\documentclass[journal,comsoc]{IEEEtran}

\usepackage{color}

\usepackage{amsmath}
\usepackage{multirow}
\usepackage{amssymb}

\usepackage{algorithm,algpseudocode}

 \usepackage{amsthm}
 
\newtheorem{corollary}{Corollary}
\newtheorem{proposition}{Proposition}
\newtheorem{lemma}{Lemma}

\newcommand{\norm}[1]{\lVert#1\rVert_2}

\usepackage[font=small,labelfont=bf]{caption}

\usepackage{subcaption}
 \usepackage{graphicx}

\usepackage{cite}
\usepackage[table]{xcolor}

\begin{document}
%
\title{Hemispherical Antenna
Array Architecture for High-Altitude Platform Stations (HAPS) for
Uniform Capacity Provision}
\author{Omid~Abbasi, \IEEEmembership{Senior Member,~IEEE}, Halim~Yanikomeroglu, \IEEEmembership{Fellow,~IEEE}, Georges~Kaddoum, \IEEEmembership{Senior Member,~IEEE}
\thanks{O. Abbasi and H. Yanikomeroglu are with Non-Terrestrial Networks (NTN) Lab, Department of Systems and 
Computer Engineering, Carleton University, Ottawa, ON  K1S 5B6, Canada. E-mail: omidabbasi@sce.carleton.ca; halim@sce.carleton.ca. (\textit{Corresponding author: Omid Abbasi.})}
\thanks{G. Kaddoum is with Department of Electrical Engineering, \'Ecole de Technologie Sup\'erieure (ETS), Universit\'e du Qu\'ebec, Montr\'eal, QC  H3C 1K3, Canada. E-mail: Georges.Kaddoum@etsmtl.ca.}
\thanks{A preliminary version of this work was presented at the 2024 IEEE Wireless Communications and Networking Conference (WCNC) \cite{Omid_hemi_wcnc}.}
}

\maketitle

\begin{abstract}
In this paper, we present a novel hemispherical antenna array (HAA) designed for high-altitude platform stations (HAPS). A significant limitation of traditional rectangular antenna arrays for HAPS is that their antenna elements are oriented downward, resulting in low gains for distant users. Cylindrical antenna arrays were introduced to mitigate this drawback; however, their antenna elements face the horizon leading to suboptimal gains for users located beneath the HAPS.
 To address these challenges, in this study, we introduce our HAA. 
 An HAA's antenna elements are strategically distributed across the surface of a hemisphere to ensure that each user is directly aligned with specific antenna elements.
To maximize users’ minimum signal-to-interference-plus-noise ratio (SINR), we formulate an optimization problem. After performing analog beamforming, we introduce an antenna selection algorithm 
and show that this method achieves optimality when a substantial number of antenna elements are selected for each user.
Additionally, we employ the bisection method to determine the optimal power allocation for each user. Our simulation results convincingly demonstrate that the proposed HAA outperforms the conventional arrays, and provides uniform rates across the entire coverage area. 
With a $20~\mathrm{MHz}$ communication bandwidth, and a $50~\mathrm{dBm}$ total power, the proposed approach reaches sum rates of $14~\mathrm{Gbps}$.

\end{abstract}


\section{Introduction}

\subsection{Background}
The high-altitude platform station (HAPS) has emerged as a prominent contender for integration in the forthcoming 6th generation (6G) and beyond, as has been evidenced by recent research \cite{abbasi2023haps_mag, survey_haps,Sustainable_HAPS,dicandia2022space}. To date, HAPS has showcased its versatility across a broad spectrum of applications, including communication~\cite{3GPP_haps}, computing~\cite{Qiqi_Caching, Qiqi_Handoff}, localization~\cite{Hongzhao}, and sensing~\cite{Gunes_sensing}. HAPS \footnote{In this paper, we use HAPS for both singular and plural forms.} are conventionally positioned in the stratosphere at an altitude of approximately $20~\mathrm{km}$, and maintain a quasi-stationary position relative to Earth~\cite{grace2011broadband}. This strategic deployment enables them to offer line-of-sight (LoS) communication with an expansive coverage radius spanning $50-500~\mathrm{km}$. Furthermore, to ensure uninterrupted operation, HAPS installations come equipped with robust computing resources and integrated battery systems ~\cite{survey_haps}. Recent research (see, e.g., \cite{Qiqi_Caching, Qiqi_Handoff}) has explored integrating HAPS-based computing as a promising extension of edge computing. Meanwhile, the authors of \cite{Gunes_sensing} envisioned HAPS as an instrumental technology for communication, computing, caching, and sensing in next-generation aerial delivery networks.

HAPS used as international mobile telecommunication (IMT) base stations (BSs) are referred to as HIBS \cite{IMTBS}. In recent years, the use of HAPS for communication has garnered significant attention in both academic \cite{abbasi2022uxnb, Sahabul_haps, Alouini,Dahrouj} and industrial \cite{3GPP_haps, 3GPP_haps_RF, hoshino2021service, noerpel2016multibeam} sectors.
In \cite{Sahabul_haps}, for instance, the authors envisaged HAPS as super macro BSs to provide connectivity across a wide range of applications. Unlike conventional HAPS systems that are primarily used to provide extensive coverage for remote areas and disaster recovery, the authors of \cite{Sahabul_haps} envisioned deploying their proposed HAPS in densely populated metropolitan areas. 
Additionally, in \cite{abbasi2022uxnb}, an HAPS was proposed to support the backhauling of aerial BSs. The HAPS's capabilities as an HIBS are explored in more details in \cite{3GPP_haps} and \cite{3GPP_haps_RF}, where specifications and performance characteristics are thoroughly investigated.

\subsection{State of the Art}

In the literature on HAPS, two primary types of antennas to establish ground cells have been proposed. The first and more conventional category is aperture-type antennas, which are positioned beneath a HAPS to generate beams, as discussed in \cite{Thornton_2003}. Each antenna is responsible for creating one beam and, consequently, one cell. To establish cells in different locations, each of multiple antennas needs to be mechanically steered toward their respective designated spot.
The second category is array-type antennas, which are arranged in a geometric grid pattern to form multiple beams (see \cite{El_Jabu2001}). In this latter category, beamforming allows beams to be steered towards various locations. Conducting beamforming on the HAPS level enables generating multiple beams in different directions. When a substantial number of antenna elements is employed, these beams can achieve a high array gain. To effectively manage intercell
interference, both types of antennas can operate with various frequency reuse factors.

Elements in antenna arrays can be arranged in various architectures. In \cite{El_Jabu2001}, the authors proposed planar antenna arrays for conventional HAPS systems. These planar arrays can be square or rectangular. For instance, in \cite{El_Jabu2001}, the authors explored cellular communication using a HAPS equipped with a rectangular antenna array (RAA). The results of the aforementioned study demonstrated feasibility of constructing cells with a radius as small as $100~\mathrm{meters}$ using square arrays measuring less than $12~\mathrm{meters} \times 12~\mathrm{meters}$ and a frequency band of $2~\mathrm{GHz}$. In \cite{Capacity_HAPS_2016}, the authors presented an analysis of the capacity of both sparse users and hotspot users in a massive multiple-input multiple-output (MIMO) system implemented with an RAA and a HAPS. Furthermore, \cite{Softbank_giga} showcased that a HAPS equipped with a planar array can achieve data rates of up to $1~\mathrm{Gbps}$.
In another relevant study \cite{Opport_BF2016}, a user-clustered opportunistic beamforming approach was proposed for a HAPS equipped with a uniform linear array. Likewise, the authors of \cite{Cylindrical_VTC} and \cite{hoshino2021service} incorporated a HAPS with cylindrical antenna array (CAA) in a massive MIMO system to achieve ultra-wide coverage and high capacity. The aforementioned work demonstrated that the proposed system can achieve more than twice the capacity of conventional massive MIMO systems with planar arrays.
In \cite{Hsieh+Gosh}, the authors proposed a hexagonal antenna array featuring a downward-facing panel serving the center cell and six outward-facing panels serving six outer cells. The study mentioned above revealed that the proposed system can attain a downlink sector throughput of $26~\mathrm{Mbps}$ with a radius of $100~\mathrm{km}$ and a bandwidth of $20~\mathrm{MHz}$.


Considering that the number of users to be served frequently by far exceeds the number of beams created, user scheduling must be employed in various domains, including time, frequency, code, and power. Additionally, high channel correlation among users may hinder users' spatial separation and necessitate user scheduling. In \cite{Opport_BF2016}, the authors proposed using a beamformer partitioned into two parts: a dominant part directed towards the user cluster, on the one hand, and a random part to enhance user fairness, on the other hand. Slowly varying expected channel correlation was chosen as the metric for user clustering. It is interesting to note that in this work, only one user from each cluster was assumed to be served by a given beam in a resource block. The other users in that beam are served in other resource blocks, thereby leaving to only inter-beam interference. The authors of \cite{Two-Stage_Precoding_Rician} proposed a two-stage precoding design for the Ricean channel in the HAPS massive MIMO systems. 

In contrast, in \cite{Location-Assisted-Precoding2017} and \cite{User_groupinh_HAPS_letter}, all users in a given cluster were served in the same resource blocks which resulted in both inter-beam and intra-beam interference. In \cite{Location-Assisted-Precoding2017}, a 3-D massive MIMO system with a HAPS equipped with a 2-D RAA  was investigated for air-to-ground transmission. Using leveraged slow-varying parameters, such as channel correlation and the angles of departure of users (UEs), the authors proposed a two-layer location-assisted precoding scheme for downlink transmission, followed by UE clustering. The results demonstrated that location-assisted precoding significantly outperforms matched filter precoding.
The authors of \cite{User_groupinh_HAPS_letter} explored user grouping and beamforming for HAPS-equipped massive MIMO systems. They introduced user grouping based on the average chordal distance and an outer beamformer that took into account the user's statistical eigenmode. The authors found that their proposed scheme outperformed conventional channel correlation matrix-based schemes.

Another relevant investigation \cite{Beamspace_NOMA_Pingping_COMML_2021} introduced a beamspace HAPS-NOMA scheme with a RAA was proposed that enabled the non-orthogonal multiple access (NOMA) scheme to simultaneously serve multiple users in each beam. The authors obtained a zero-forcing digital precoder based on an equivalent channel.
Furthermore, in \cite{Cylindrical_VTC}, a cylindrical massive MIMO system was introduced to enhance the capacity and coverage of a HAPS. Each of the beams was allocated to only one user.
The authors of \cite{Softbank_Movements} and \cite{Softbank_Prototype} developed a beamforming method that considered the movements of the solar panel in HAPS systems. They created cells using a CAA with the users in each beam served in different resource blocks. Finally, in \cite{Softbank_Cell_config}, the cell configuration in HAPS systems was optimized considering a frequency reuse factor of one.


\subsection{Motivation and Contributions}

RAAs are ideal for users located directly beneath a HAPS, as all their antenna elements are oriented downward, so maximum gain is achieved in this direction. However, for users situated at a considerable distance from the HAPS, each antenna element's gain becomes negligible, thereby rendering RAAs suboptimal. 
To address this challenge, the authors of \cite{Cylindrical_VTC} proposed a CAA. In their innovative design, some antenna elements were placed on the facade of a cylindrical structure, while others were located on a circular surface beneath the cylinder. This arrangement enabled maximizing the utility of antenna elements for users located at greater distances. Nevertheless, it's worth noting that their proposed architecture was not optimal, as many users still did not have a direct view of antenna elements to obtain maximum gains.

In this paper, aiming to mitigate the limitations of RAA and CAA systems, we introduce a hemispherical antenna array (HAA). In our proposed scheme, each user is directly aligned with specific antenna elements to ensure they receive the maximum gain possible from these elements. Simulation results demonstrate that our proposed HAA system outperforms RAA and CAA systems. Furthermore, in contrast to the baseline schemes, the proposed scheme achieves uniform data rates across the entire coverage area.

The main contributions of this paper can be summarized as follows:
\begin{itemize}
     \item We propose an HAA scheme where antenna elements are strategically placed over the surface of a hemisphere. A subset of these elements is selected to create focused beams on the ground. The number of elements selected determines the size of the resulting beam.


    \item In the proposed HAA system, we derive the achievable rate for each user when the number of elements selected is relatively high.

    \item To optimize the system's performance, we formulate an optimization problem that aims to maximize users' minimum SINR while adhering to constraints regarding the total power available at the HAPS and the number of elements selected for each user.

    \item Our system employs analog beamforming techniques that capitalize on steering vectors of antenna elements selected for each user.

    \item We introduce an antenna selection algorithm that takes into account the gains of the selected antenna elements and demonstrate that this method is optimal when a substantial number of antenna elements is selected.

    \item We utilize the bisection method to determine the optimal power allocation for each user and thus further enhance the system's efficiency and performance.

    \end{itemize}


\subsection{Outline and Notations}
The remainder of this paper is organized as follows. Section II presents the system model and the channel model. Section III presents the proposed transmitter scheme at the HAPS. Section IV formulates the optimization problem, provides its solution, and derives the achievable rate. Section V provides simulation
results to validate the proposed scheme's performance. Finally, Section VI concludes the paper.

\textit{Notations:} 
In this paper, scalar variables are denoted by regular lowercase letters, e.g., $x$.
Matrices and vectors are represented by boldface uppercase and lowercase letters, respectively (e.g., $\mathbf{X}$ and $\mathbf{x}$).
The element at the $(i,j)$-th position of matrix $\mathbf{X}$ is denoted as $[\mathbf{X}]_{i,j}$.
Transpose, conjugate, and Hermitian operators on a matrix $\mathbf{X}$ are denoted as $\mathbf{X}^T$, $\mathbf{X}^*$, and $\mathbf{X}^H$, respectively.
The Hadamard product is represented by $\odot$.
The diagonal matrix of a vector $\mathbf{x}$ is denoted as $\mathrm{diag}(\mathbf{x})$.
The trace of a matrix $\mathbf{X}$ is given by $\mathrm{Tr}(\mathbf{X})$.
The absolute value of a scalar $x$ is denoted as $|x|$.
The Frobenius norm of a vector $\mathbf{x}$ is represented by $\norm{\mathbf{x}}$.

\section{System Model and Channel Model}
\subsection{System Model}
\begin{figure}[!t]
  \centering
\includegraphics[width=\linewidth]{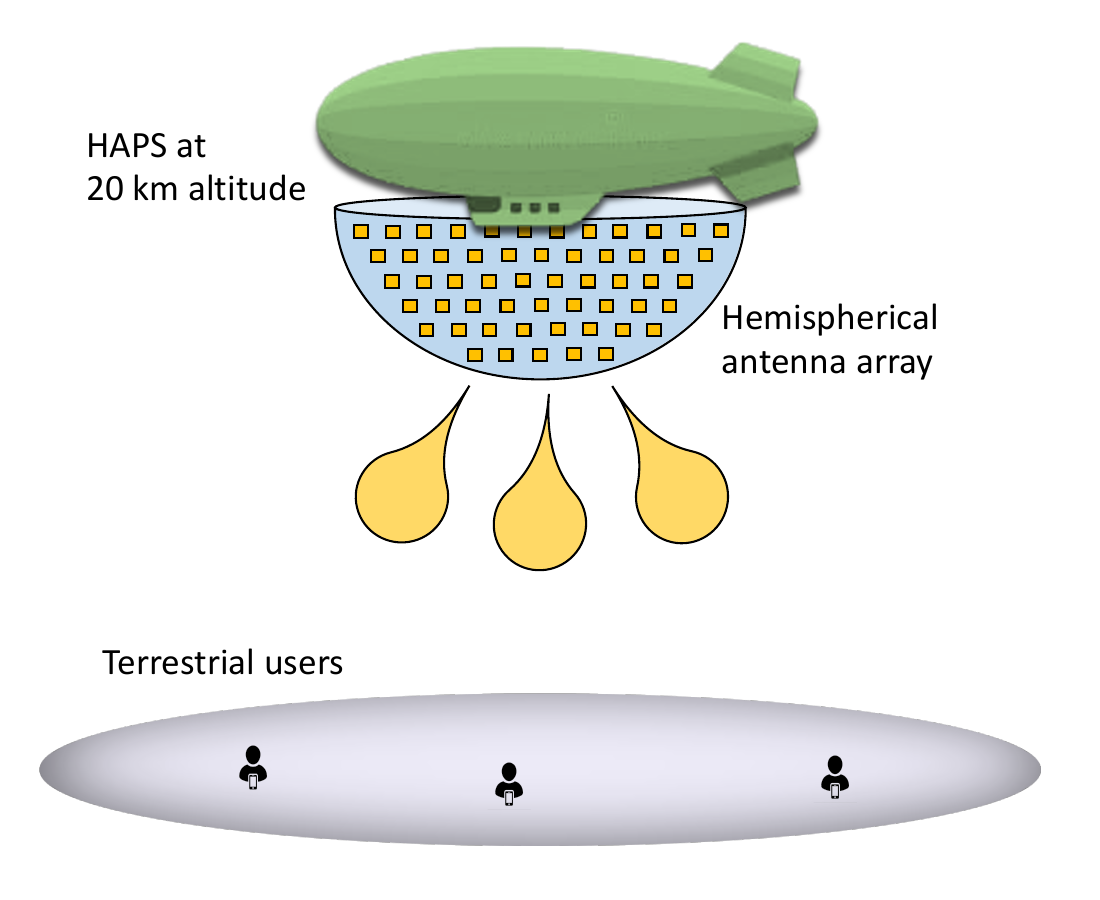}
  \caption{System model for the proposed HAA.}\label{system-model}
  \end{figure}

 The proposed HAA is illustrated in Fig. \ref{system-model}. We consider a scenario in which a HAPS serves $K$ terrestrial users in the downlink mode, in the sub-6 GHz frequency band. The HAPS is equipped with $M$ antenna elements placed on a hemispherical surface. Subsets of these elements are selected to create focused beams on the ground. The number of elements selected determines the size of the resulting beam. Note that, as depicted in Fig. \ref{system-model}, the HAA is positioned beneath the HAPS. This array, being affixed to the underside of the HAPS, constitutes a distinct equipment from the HAPS platform. In this setup, we assume that the terrestrial users are each equipped with a single antenna. In our proposed scheme, each user is directly aligned with specific antenna elements to ensure that they receive the maximum possible gains from these elements. To this end, we assume a large number of elements is positioned on the hemispherical surface so that each user can be aligned with a high number of antenna elements facing towards them.  Moreover, analog beamforming is used to ensure that the signals transmitted from the HAPS are aligned with intended users. This beamforming process relies on steering vectors associated with the HAA at the HAPS. 

In this paper, the terms \textit{cell} and \textit{beam} are used interchangeably. More specifically, a cell represents the footprint of a beam created by the proposed antenna array. Importantly, multiple users can be connected to each beam by being scheduled in orthogonal resource blocks. In the special case where only one user is served by each beam, the array's multi-cell functionality can be simplified to multi-user massive MIMO.

\subsection{Channel Model}
In our scheme, the channel between user $k$ and antenna element $m$ of the HAPS is denoted as $h_{km}$, which includes both large-scale fading 
and small-scale multipath fading effects. We assume the existence of a LoS link between the users and the HAPS; hence Ricean distribution is considered for the channel between user $k$ and antenna element $m$ of HAPS as follows:
\begin{equation}
    h_{km}=10^{-\dfrac{\mathsf{PL_{k}}}{20}}(\sqrt{P_{k}^{\mathsf{LoS}}}b_{km}+\sqrt{P_{k}^{\mathsf{NLoS}}}\rm{CN}(0,1)),
\end{equation}
where $\rm{CN}(0,1)$ shows a complex normal random variable with a mean of $0$ and variance (power) equal to $1$. Also, 
\begin{equation}
\begin{split}
      b_{km}=\exp(j2\pi(\frac{d_{km}}{\lambda_\mathsf{sub6}}))
\end{split}
\end{equation} 
indicates the phase shift of the LoS link's signal due to distance, with $d_{km}$ indicating the distance between user $k$ and antenna element $m$ of the HAPS (see Section IV for its derivation), $\lambda_\mathsf{sub6}=\frac{c}{f_\mathsf{sub6}}$ representing the wavelength, and $c=3\times10^8 ~\mathrm{m/s}$ being the speed of light. 
Each user $k$'s coordinates are denoted by $(x_{k}, y_{k},0)$. 
$\theta_{k}$ and $\phi_{k}$ represent the elevation and azimuth angles of departure, respectively, of the transmitted signal directed towards user $k$.
It should be noted that $\bold{b}_{k}=[b_{km}]_{1\times M}$ is the steering vector of the HAPS's transmit antenna array for user $k$.  
In addition, we have $\mathsf{PL_{k}}=P_{k}^{\mathsf{LoS}}\mathsf{PL}_{k}^{\mathsf{LoS}}+P_{k}^{\mathsf{NLoS}}\mathsf{PL}_{k}^{\mathsf{NLoS}}$, where the LoS and non-LoS (NLoS) path losses of the link between the HAPS and user $k$ are equal to $\mathsf{PL}_{k}^{\mathsf{LoS}}=\mathsf{FSPL}_{k}+\eta_{\mathsf{LoS}}^{\mathsf{dB}}$, and $\mathsf{PL}_{k}^{\mathsf{NLoS}}=\mathsf{FSPL}_{k}+\eta_{\mathsf{NLoS}}^{\mathsf{dB}}$, respectively \cite{Hourani}. In these equations, $\mathsf{FSPL}_{k}=10\log(\frac{4\pi f_\mathsf{sub6}d_{k}}{C})^2$ represents the free-space path loss (FSPL), while $\eta_{\mathsf{LoS}}^{\mathsf{dB}}$ and $\eta_{\mathsf{NLoS}}^{\mathsf{dB}}$ indicate excessive path losses (in $\mathrm{dB}$)  affecting the air-to-ground links in the LoS and NLoS cases, respectively \cite{Irem2016}. Next, $P_{k}^{\mathsf{LoS}}=\frac{1}{1+A\exp({-B(90-\theta_{k})-A})}$ denotes the probability of establishing a LoS link between user $k$ and the HAPS, with $\theta_{k}$ (in $\mathrm{degrees}$) referring to the elevation angle between user $k$ and HAPS, and $A$ and $B$ being environment-dependent parameters \cite{Hourani}. Clearly, $P_{k}^{\mathsf{NLoS}}=1- P_{k}^{\mathsf{LoS}}$ indicates the probability of establishing a NLoS link between user $k$ and the HAPS.

The large-scale channel power gain of the link between user $k$ and the HAPS is equal to
\begin{equation}
\begin{split}
    \beta_{k}^2=E\{|h_{km}|^2\}&=E\{h_{km}h_{km}^*\}=10^{-\frac{\mathsf{PL_{k}}}{10}}\\&=10^{-\frac{P_{k}^{\mathsf{LoS}}\mathsf{PL}_{k}^{\mathsf{LoS}}+P_{k}^{\mathsf{NLoS}}\mathsf{PL}_{k}^{\mathsf{NLoS}}}{10}}.
\end{split}
\end{equation}
By considering $\beta_{0}=(\frac{4\pi f_\mathsf{sub6}}{c})^{-2}$ to be the channel gain at the reference distance $d_{0}=1~\mathrm{m}$, the large-scale channel power gain can be rewritten as
  $\beta_{k}^2=\eta_{k}\beta_{0}(d_{k})^{-2}$,
where $\eta_{k}=10^{-\frac{P_{k}^{\mathsf{LoS}}\eta_{\mathsf{LoS}}^{\mathsf{dB}}+P_{k}^{\mathsf{NLoS}}\eta_{\mathsf{NLoS}}^{\mathsf{dB}}}{10}}$ represents the excessive path loss.
We consider independent additive white Gaussian noise (AWGN) with the distribution $\rm{CN}(0,\sigma^{2})$ at each user's receive antenna element. 


\section{Transmitter Scheme Proposed for the  HAPS}
\begin{figure}[!t]
  \centering  \includegraphics[width=\linewidth]{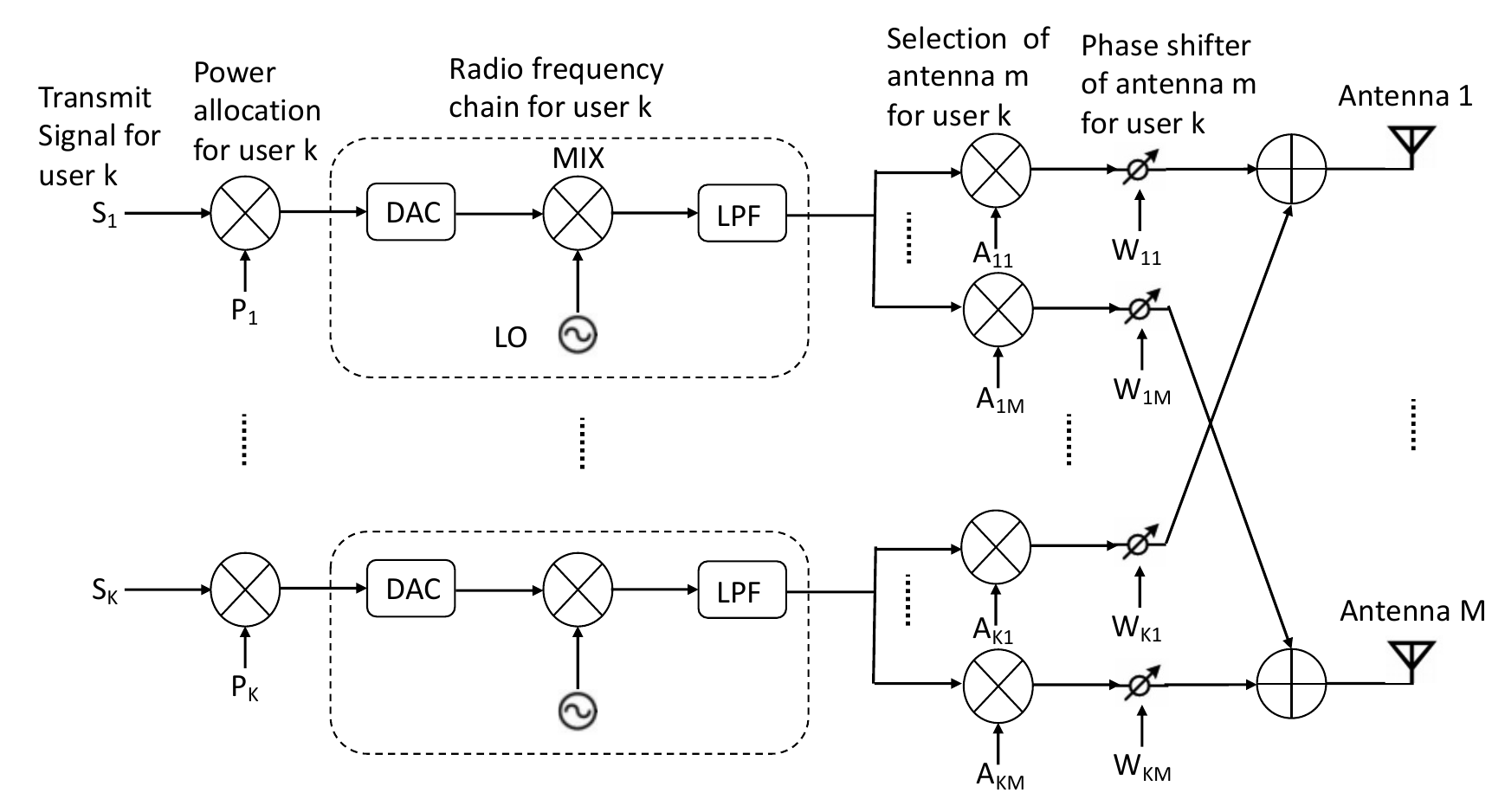}
  \caption{The proposed transmitter structure at the HAPS.}\label{transmitter}
\end{figure}
We can see the transmitter scheme proposed for the HAPS in Fig. \ref{transmitter}. We assume that the HAPS transmits all users' signals in the same time-frequency resources.  User $k$'s transmit signal is denoted as $\sqrt{p_k}s_k$, where $p_k$ (for $k\in \mathcal{K}=\{1,2,...,K\}$) represents the transmit power allocated to each user $k$ at the HAPS while $s_k$ ($E{|s_{k}|^{2}}=1$) is the symbol transmitted for user $k$. Next, users' digital signals are converted to analog signals and then up-converted to the carrier frequency band utilizing radio frequency (RF) chains. Importantly, we assume that the number of RF chains is equal to the number of users ($K$). Once all users' signals have passed through the RF chains, we determine which antenna elements have been selected to transmit each user's signal. Then, we employ analog beamforming to direct signals towards users. The signal received at user $k$ is expressed as 
\begin{equation} \label{y_k}
\begin{split}
    y_{k}&=\mathbf{h}_k\mathbf{G}_k(\mathbf{W}\odot \mathbf{A})\mathbf{P}\mathbf{s}+z_{k}\\&=\sum_{k'=1}^{K}\sum_{m=1}^{M}h_{km}g_{km}w_{k'm}a_{k'm}\sqrt{p_{k'}}s_{k'}+z_{k},
\end{split}
\end{equation}
where $\mathbf{s}=[s_1,...,s_K]^T \in \mathbb{C}^K$ and $\mathbf{P}=\mathrm{diag}(\sqrt{p_1},...,\sqrt{p_K}) \in \mathbb{R}^{K\times K}$. $\mathbf{A}=[\mathbf{a}_1,...,\mathbf{a}_K]^T \in \mathbb{B}^{K\times M}$ indicates the antenna selection matrix, where $[\mathbf{A}]_{km}=1$ if antenna $m$ is selected to transmit user $k$'s signal, and $[\mathbf{A}]_{km}=0$  otherwise. In this case, $k\in \mathcal{K}=\{1,2,...,K\}$, $m\in \mathcal{M}=\{1,2,...,M\}$, and $\mathbf{a}_{k} \in \mathbb{C}^{M\times 1}$ is the antenna selection vector for user $k$. Furthermore,  $\mathbf{W}=[\mathbf{w}_1,...,\mathbf{w}_K]^T \in \mathbb{C}^{K \times M}$ represents the analog beamforming matrix at the HAPS, with $[\mathbf{W}]_{km}$ indicating the phase shifter (PS) of the $m$-th antenna element for user $k$ and $\mathbf{w}_{k} \in \mathbb{C}^{M\times 1}$ being the beamforming vector for user $k$. 

It is worth noting that in the proposed HAA scheme, we introduce antenna and power pooling among all users within the coverage region. This implies that each user has access to various options for antenna element selection, which enhances diversity. In addition, the proposed HAA scheme offers great flexibility for power allocation among users. This differs from terrestrial networks with a smaller pool for antenna element selection and power allocation.

The transmit vector from the HAPS is equal to $\mathbf{t}=(\mathbf{W}\odot \mathbf{A})\mathbf{P}\mathbf{s} \in \mathbb{C}^{M\times 1}$, where $\mathbf{s}=[s_1,...,s_K]^T \in \mathbb{C}^K$ is the data vector with $\mathrm{E}\{\mathbf{s}\mathbf{s}^H\}=\mathbf{I}_K$. The transmit vector must satisfy the total power constraint at the HAPS, which is given by $\mathrm{E}\{\mathbf{t}^H\mathbf{t}\}=\mathrm{E}\{\norm{\mathbf{t}}^2\}\leq P_{\mathsf{HAPS}}$, where $P_\mathsf{HAPS}$ is the total power at the HAPS. Moreover, $\mathbf{G}_k=\mathrm{diag}(\sqrt{g_{k1}},...,\sqrt{g_{kM}}) \in \mathbb{R}^{M\times M}$ represents the antenna gain matrix at the HAPS for user $k$, with $g_{km}$ indicating the $m$-th antenna element's gain at user $k$. Furthermore, $\mathbf{h}_k=[h_{k1},...,h_{kM}] \in \mathbb{C}^{1\times M}$ represents the channel vector to user $k$ while $z_k \sim \mathrm{CN}(0,\sigma^2)$ is 
complex AWGN noise with variance $\sigma^2$ at user $k$'s receiver.

\section{Optimization problem}
On formulating an optimization problem that maximizes the users' minimum SINR, we find the optimum powers allocated to users at the HAPS, i.e., $\mathbf{P}=\mathrm{diag}(\sqrt{p_1},...,\sqrt{p_K}) \in \mathbb{R}^{K\times K}$, the analog beamforming matrix at the HAPS, i.e.,$\mathbf{W}=[\mathbf{w}_1,...,\mathbf{w}_K]^T \in \mathbb{C}^{K \times M}$,
 and the antenna selection matrix at the HAPS, i.e., $\mathbf{A}=[\mathbf{a}_1,...,\mathbf{a}_K]^T \in \mathbb{B}^{K\times M}$. Note that matrix $\mathbf{A}_k$, that is the diagonal matrix version of vector $\mathbf{a}_{k}$, denoted as $\mathbf{A}_k=\mathrm{diag}(\mathbf{a}_{k}) \in \mathbb{B}^{M\times M}$.
We can write the optimization problem as follows:
\begin{alignat}{2}
\mathcal{(P):~~~~    }&\underset{\bold{P},\bold{A},\bold{W}}{\mathrm{max}}~~~        && \underset{k}{\mathrm{min}}  ~~\mathsf{SINR_k}\label{eq:obj_fun}\\
&\text{s.t.} &      &  \mathrm{E}\{\norm{\mathbf{t}}^2\}=
\mathrm{E}\{\sum_{k=1}^{K}\mathbf{w}_{k}^H\mathbf{A}_{k}\mathbf{w}_{k}p_{k}\} \leq P_{\mathsf{HAPS}}, \label{eq:constraint-sum}\\
               &&& p_{k}>0, ~~\forall k \in \mathcal{K}, \label{eq:constraint+}\\&&& \mathrm{Tr}(\mathbf{A}_k)\leq M_{k},~~\forall k \in \mathcal{K}, \label{eq:constraint-sel}\\&&&\sum_{k=1}^K[\mathbf{A}]_{km}\leq M_{\mathsf{element}},~~\forall m \in \mathcal{M}, \label{eq:constraint-element}\\&&&
               [\mathbf{A}]_{km}\in \{0,1\}, ~~\forall k \in \mathcal{K}, ~~\forall M \in \mathcal{M}, \label{eq:constraint01}\\&&&
               |[\mathbf{W}]_{km}|=\frac{1}{\sqrt{M_k}}, ~~\forall k \in \mathcal{K}, ~~\forall M \in \mathcal{M}, \label{eq:constraintCM}
\end{alignat}
 where constraint in (\ref{eq:constraint-sum}) shows the maximum total power constraint at the HAPS, constraint in(\ref{eq:constraint-sel}) limits the total number of antennas that can be selected for each user, constraint in (\ref{eq:constraint-element}) limits the total number of users that can be selected for each antenna element, and constraint in (\ref{eq:constraintCM}) refers to the constant modulus condition that applies because PSs are used. We derive an expression for each user's SINR in Proposition \ref{propos_sinr}.
 
 The problem ($\mathcal{P}$) is a mixed integer programming (MIP) problem with a discrete variable $\mathbf{A}$ and continuous variables $\bold{P}$ and $\bold{W}$, which are known to be NP-hard. Also, the objective function of ($\mathcal{P}$) is not concave with respect to the variables. Moreover, constraint in \eqref{eq:constraint-sum} is not convex due to its coupled variables $\mathbf{A}$ and $\bold{W}$, and the constant modulus constraint in \eqref{eq:constraintCM} is also non-convex. We solve this problem by splitting it into three sub-problems. In the first sub-problem, we employ users' steering vectors at the HAPS to derive analog beamforming matrix $\bold{W}$. Then, we utilize $\bold{W}$ and the use-and-then-forget bound \cite{marzetta2016fundamentals} to derive the achievable rates of users for the proposed HAA. In the second sub-problem, we propose a heuristic method to find antenna selection matrix $\bold{A}$. Finally, in the third sub-problem, we use the bisection method to calculate optimal values for power allocation matrix $\bold{P}$.
 
 \subsection{Steering Vector-Based Analog Beamforming}
 In this paper, we assume that the analog beamforming matrix at the HAPS is derived from steering vectors of the antennas selected for each user, i.e., $\mathbf{w}_{k}=\frac{1}{\sqrt{M_k}}\mathbf{b}_k^*$, where $M_k$ indicates the number of antenna elements selected for user $k$. To direct the signals transmitted from the HAPS toward users, this beamforming is performed with the aid of PSs. In proposition 1, we utilize the beamforming matrix to derive an achievable rate for users in the proposed system.

 \begin{proposition}\label{propos_sinr}
The achievable rate of user $k$ in the proposed HAA scheme for HAPS is given by
\begin{equation}
   R_k=\mathsf{BW} \log_2(1+\mathsf{SINR}_k), 
\end{equation}
where
\begin{equation}    \label{SINR}
     \mathsf{SINR}_k=\frac{\frac{p_k\beta_k^2}{M_k}\mathrm{Tr}(\mathbf{G}_k\mathbf{A}_k)^2}{\beta_k^2\sum_{k'=1}^K \frac{p_{k'}}{M_{k'}}\mathrm{Tr}(\mathbf{G}_k^2\mathbf{A}_{k'})+\sigma^2},
\end{equation}
where $\beta_k^2$ represents large-scale fading, $p_k$ indicates the power allocated to user $k$, $\mathbf{G}_k \in \mathbb{R}^{M\times M}$ represents user $k$'s diagonal antenna gain matrix at the HAPS, $\mathbf{A}_k \in \mathbb{R}^{M\times M}$ represents user $k$'s diagonal antenna selection matrix, $M_{k}$ indicates the total number of antennas selected for user $k$, and $\mathsf{BW}$ is the communication bandwidth. 
\end{proposition}
\begin{proof}
See Appendix \ref{proof_sinr}.
\end{proof}

Now, we can update ($\mathcal{P}$) as follows:
\begin{alignat}{2}
\text{($\mathcal{P}$):~~~~    }&\underset{\bold{P},\bold{A}}{\mathrm{max}}~~~        && \underset{k}{\mathrm{min}}  ~~\frac{\frac{p_k\beta_k^2}{M_k}\mathrm{Tr}(\mathbf{G}_k\mathbf{A}_k)^2}{\beta_k^2\sum_{k'=1}^K \frac{p_{k'}}{M_{k'}}\mathrm{Tr}(\mathbf{G}_k^2\mathbf{A}_{k'})+\sigma^2}\label{eq:obj_fun_no_W}\\
&\text{s.t.} &      &  \sum_{k=1}^K p_k \leq P_{\mathsf{HAPS}}, \label{eq:constraint-sum_no_W}\\
               &&& \eqref{eq:constraint+}, \eqref{eq:constraint-sel}, \eqref{eq:constraint-element}, \eqref{eq:constraint01}.\nonumber
\end{alignat}

\begin{corollary}\label{Corrolary_sinr}
In the interference-limited regime, assuming allocation of equal power for all users, i.e., $p_k=\frac{P}{K},~\forall k$, and assuming equal number of selected antenna elements for all users, i.e., $M_1=...=M_K$, the achievable rate of user $k$ in the proposed HAA scheme for HAPS is given by
\begin{equation}
   R_k^{\infty}=\mathsf{BW} \log_2(1+\mathsf{SINR}_k), 
\end{equation}
where
\begin{equation}    \label{SINR_regime}
     \mathsf{SINR}_k^{\infty}=\frac{\mathrm{Tr}(\mathbf{G}_k\mathbf{A}_k)^2}{\sum_{k'=1}^K \mathrm{Tr}(\mathbf{G}_k^2\mathbf{A}_{k'})}.
\end{equation}
\end{corollary}
\begin{proof}
  The proof is straightforward; hence, we omit it.
\end{proof}

\begin{corollary}\label{Corrolary_sumr}
In the interference-limited regime, assuming omnidirectional antennas with gains equaling $1$, the sum rate of the proposed HAA scheme for HAPS is given by
\begin{equation}
   \sum_{k=1}^K R_k^{\infty}=\mathsf{BW}  M_k \log_2(e), 
\end{equation}
where $M_k$ is the number of selected antenna elements for users.
\end{corollary}
\begin{proof}
    Assuming omnidirectional gains of $1$ for elements, \eqref{SINR_regime} can be written as $\mathsf{SINR}_k^{\infty}=\frac{M_k^2}{KM_k}=\frac{M_k}{K}$. Therefore, for the sum rate, we can write
    \begin{equation}
        \sum_{k=1}^K R_k^{\infty}=\sum_{k=1}^K \mathsf{BW} \log_2(1+\frac{M_k}{K})=\mathsf{BW} \log_2(1+\frac{M_k}{K})^K.
    \end{equation}
    When $K$ tends to infinity, knowing $\lim_{x\to\infty} (1+\frac{1}{x})^x=e$, we have
    \begin{equation}
        \sum_{k=1}^K R_k^{\infty}= \mathsf{BW} M_k\log_2 e.
    \end{equation}
    Thus the proof is completed.
\end{proof}

\subsection{Antenna Selection Based on the Gains of the Antenna Elements}
\begin{figure}[!t]
  \centering
  \includegraphics[width=\linewidth]{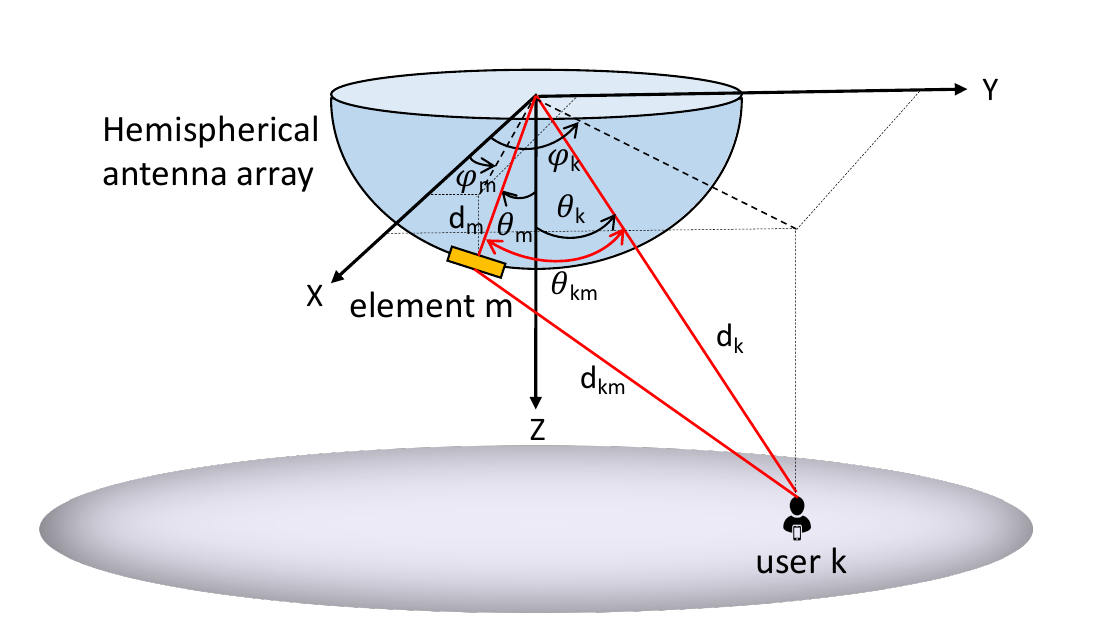}
  \caption{Polar coordinates for antenna element $m$ expressed as $(d_m,\theta_m,\phi_m)$ and for user $k$ denoted as $(d_k,\theta_k,\phi_k)$. In this figure, $d_{km}$ is the distance between antenna element $m$ and user $k$, while $\theta_{km}$ represents the angle between them.}\label{polar}
\end{figure}

In this section, we propose a heuristic algorithm to find the antenna selection matrices $\mathbf{A}_k, \forall k \in \mathcal{K}$. To this end, we first derive the gain of each of the HAPS's antenna elements at each user's location based on antenna pattern formulas provided in \cite{3GPP_haps}. In  Fig. \ref{polar}, one can see the polar coordinates for antenna element $m$ expressed as $(d_m,\theta_m,\phi_m)$ and those for user $k$ denoted as $(d_k,\theta_k,\phi_k)$. As can be seen in Fig. \ref{polar}, the distance between antenna element $m$ and user 
$k$, i.e., $d_{km}$, is calculated in accordance with the triangle law as follows:
\begin{equation}
    d_{km}=\sqrt{d_k^2+d_m^2-2d_kd_m\cos{\theta_{km}}}.
\end{equation}

 \begin{proposition}\label{propos_gain}
The gain of antenna element $m$ (in $\mathrm{dB}$) at the location of user $k$ in the proposed HAA scheme is given by

\begin{equation}\label{G}
 g_{km}=\begin{cases}
      G_{\mathsf{E,max}}+\gamma_{km}, &\text{if $0<\theta_{km}<90$,}\\
      0, &\text{if $90<\theta_{km}<180$,}
  \end{cases}
\end{equation}
where $G_{\mathsf{E,max}}$ is the maximum directional gain of an antenna element as

\begin{equation}    \label{g_max}
    G_{\mathsf{E,max}}=\frac{32400}{\theta_{\mathsf{3dB}}^2},
\end{equation}
and
\begin{equation}    \label{gamma}
     \gamma_{km}=-\min(12(\frac{\theta_{km}}{\theta_{\mathsf{3dB}}})^2,\gamma_{\mathsf{max}}),
\end{equation}
where $\theta_{\mathsf{3dB}}$ is the $3~\mathrm{dB}$ beamwidth of each antenna element, while $\gamma_{\mathsf{max}}$ is the front-to-back ratio for each element. Furthermore,  $\theta_{km}$ is the angle between antenna $m$ and user $k$ (in $\mathrm{degrees}$) that is given by
\begin{equation} \label{theta-km}
\begin{split}  \theta_{km}&=\arccos(\sin{\theta_{k}}\sin{\theta_{m}}\cos{\phi_{k}}\cos{\phi_{m}}\\&+\sin{\theta_{k}}\sin{\theta_{m}}\sin{\phi_{k}}\sin{\phi_{m}}+\cos{\theta_{k}}\cos{\theta_{m}}).
\end{split}
\end{equation}

\end{proposition}
\begin{proof}
See Appendix \ref{proof_G}.
\end{proof}

\begin{figure}[t]
\begin{subfigure}{0.24\textwidth}
\includegraphics[width=\linewidth]{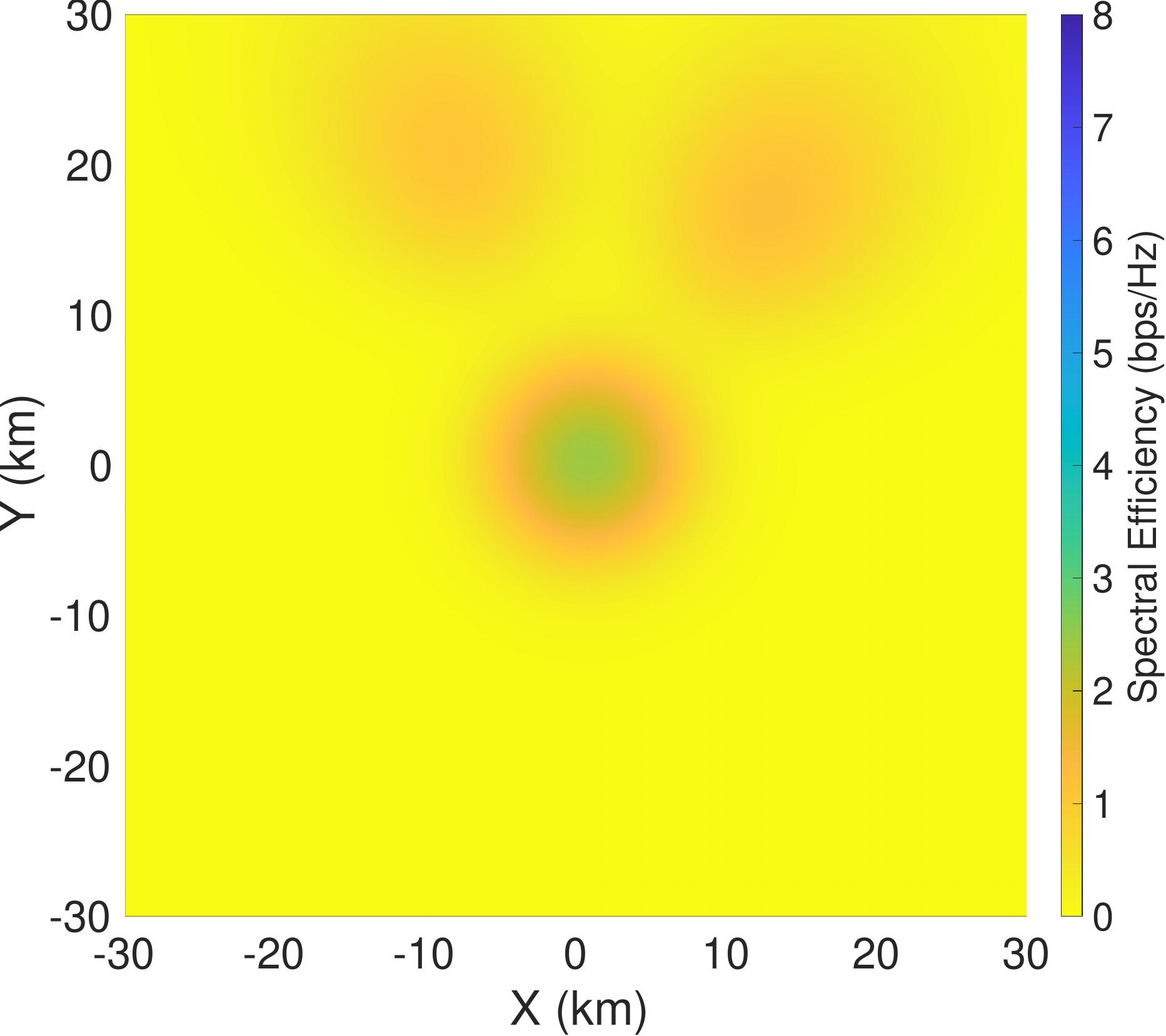} 
\caption{$\theta_{\mathsf{3dB}}=25$}
\label{hear_25}
\end{subfigure}
\begin{subfigure}{0.24\textwidth}
\includegraphics[width=\linewidth]{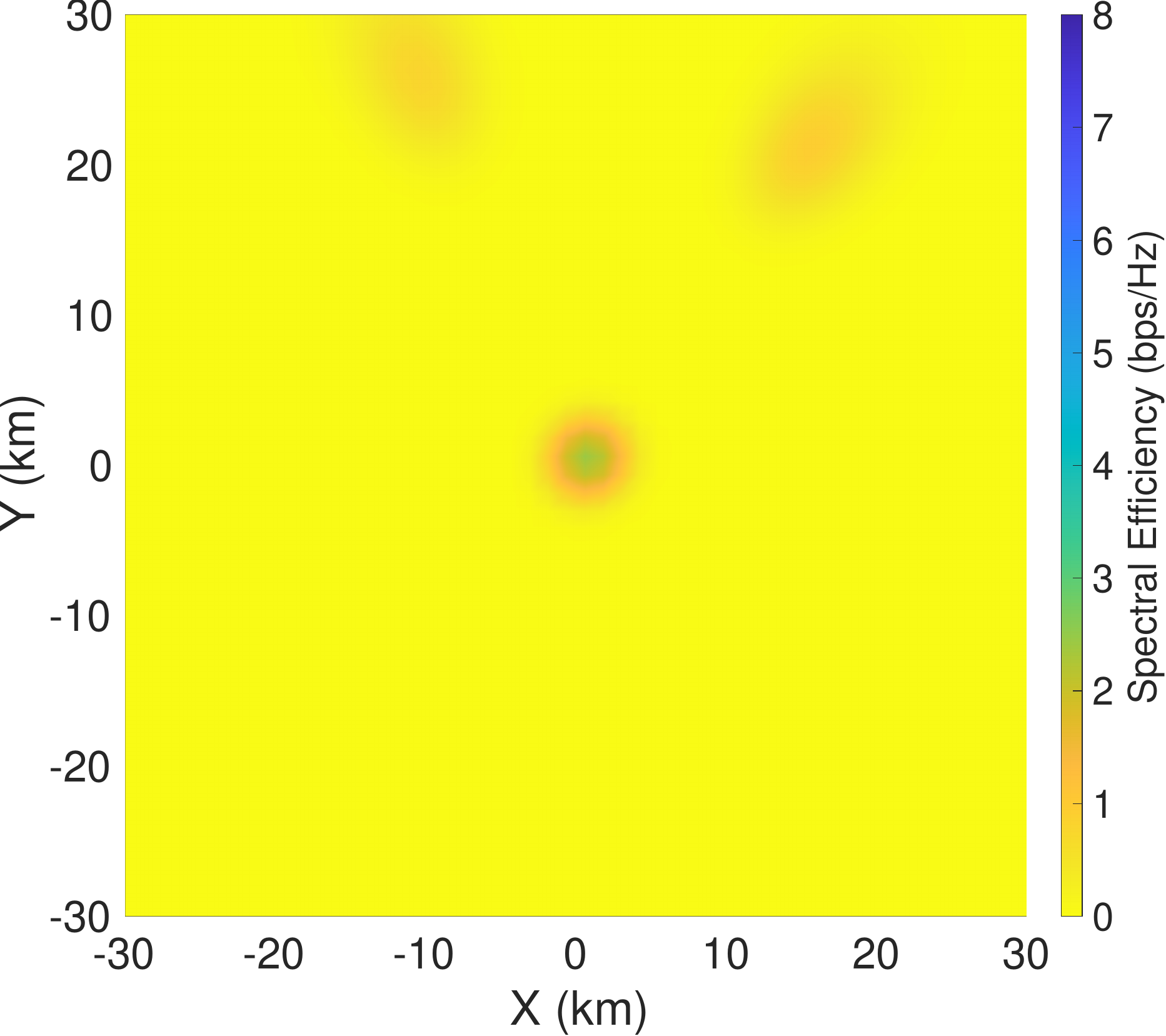} 
\caption{$\theta_{\mathsf{3dB}}=10$}
\label{heat_10}
\end{subfigure}
\caption{Heatmaps depicting the spectral efficiencies of three individual antenna elements in the proposed HAA, showcasing the impact of two distinct $3~\mathrm{dB}$ beamwidth settings. Each element is allocated a fixed power of $1~\mathrm{Watt}$.}
\label{Heat_single_element}
\end{figure}

 Fig. \ref{Heat_single_element} shows heatmaps depicting spectral efficiencies of three selected individual antenna elements determined from the formula derived in Proposition \ref{propos_gain} for the antenna gain of each element in the proposed HAA. Each element is allocated a fixed power of $1~\mathrm{Watt}$. As shown in Fig. \ref{hear_25}, while elements facing the nadir of the HAPS exhibit a circular pattern, those directed towards other spots in the proposed HAA produce ellipsoidal footprints. In Fig. \ref{heat_10}, the $3~\mathrm{dB}$ beamwidth was reduced from $25$ to $10$ for the same selected antenna elements as are shown in Fig. \ref{hear_25}, which demonstrates it is possible to create narrower beams in this configuration.

Now, we construct the antenna selection matrix based on the gains derived by the HAPS's antenna elements at the location of each user, which are denoted by vector $\mathbf{g}_k=[\sqrt{g_{k1}},...,\sqrt{g_{kM}}] \in \mathbb{R}^{1\times M},~\forall k \in \mathcal{K}$. To this end, we first sort, in descending order, the entries in each user's gain vector $\mathbf{g}_k$ to create new vectors $\mathbf{g}_{k,\mathsf{sorted}}$.
We then generate index vectors $\mathbf{i}_k \in \mathbb{Z}^{1\times M}$ with entries containing the corresponding indices of the entries of the sorted vector $\mathbf{g}_{k,\mathsf{sorted}} \in \mathbb{R}^{1\times M}$ in the original vector $\mathbf{g}_k$.
Next, we select the first $M_k$ entries of each user $k$'s index vectors $\mathbf{i}_k$ to transmit signals to user $k$ and place them in vector $\mathbf{i}_{k,\mathsf{selected}} \in \mathbb{Z}^{1\times M_k}$.
Finally, the antenna elements selected for user $k$ are given by

\begin{equation}
 a_{km}=\begin{cases}
      1, &\text{if}~ m \in \mathbf{i}_{k,\mathsf{selected}},\\
      0, &\text{otherwise.}
  \end{cases}
\end{equation}

\begin{algorithm}[t]
\caption{The method proposed for antenna selection based on the gains of antenna elements.}\label{alg1}
\begin{algorithmic}[1]
\State Derive the gain vector of the antenna elements for each user $k$, i.e., $\mathbf{g}_k \in \mathbb{R}^{1\times M}$, utilizing \eqref{G} in Proposition \ref{propos_gain}.
\State Sort  entries of the gain vector for each user, i.e., $\mathbf{g}_k$, in descending order, and save the sorted vector as $\mathbf{g}_{k,\mathsf{sorted}} \in \mathbb{R}^{1\times M}$.
\State Generate index vectors $\mathbf{i}_k \in \mathbb{Z}^{1\times M}$ with entries containing the corresponding indices of the entries of the sorted vector $\mathbf{g}_{k,\mathsf{sorted}}$ in the original vector $\mathbf{g}_k$.
\State For each user $k$, select the first $M_k$ entries of the $\mathbf{i}_k$, and place them in $\mathbf{i}_{k,\mathsf{selected}} \in \mathbb{Z}^{1\times M_k}$.
\For{$m = 1,2,\ldots,M$}
\If{$m\in \mathbf{i}_{k,\mathsf{selected}}$} 
    \State $[\mathbf{A}]_{km}=1$
\Else
        \State $[\mathbf{A}]_{km}=0$
\EndIf
\EndFor
\end{algorithmic}
\end{algorithm}

Algorithm 1 summarizes the method proposed to find the antenna selection matrix. Now, we first present a lemma that will serve as the foundation for demonstrating the optimality of the antenna selection method outlined in Algorithm 1. This demonstration specifically applies when a substantial number of antenna elements are selected for each user.

\begin{lemma}
The expression derived in (\ref{SINR}) for each user's SINR increases with the gain of each antenna element.     
\end{lemma}

\begin{proof}
    The proof is straightforward; hence, we omit it.
\end{proof}

\begin{proposition}\label{propos_selection_optimal}
    The proposed method for antenna selection based on the gains of the antenna elements that is set out in Algorithm \ref{alg1} is optimal when a substantial number of antenna elements are selected for each user.
\end{proposition}

\begin{proof}
See Appendix \ref{proof_selection_optimal}.
\end{proof}
\subsection{Optimal Power Allocation via the Bisection Method}

\begin{algorithm}[t]
\caption{Bisection method to solve the power allocation problem ($\mathcal{P}1$).}\label{alg2}
\begin{algorithmic}[1]
\State  Initialize the values of $\eta_{\mathsf{min}}$ and $\eta_{\mathsf{max}}$,
where $\eta_{\mathsf{min}}$ and $\eta_{\mathsf{max}}$ show a range for the minimum $\mathsf{SINR}$ of users. Choose a tolerance $\epsilon>0$.
\Repeat
\State Set $\eta=\frac{\eta_{\mathsf{min}} +\eta_{\mathsf{max}}}{2}$ and solve the linear feasibility
problem ($\mathcal{P}2$) by the interior-point method utilizing the CVX solver \cite{cvx}.
\State If the problem ($\mathcal{P}2$) is feasible, set $\eta_{\mathsf{min}}=\eta$; else set $\eta_{\mathsf{max}}=\eta$.
\Until {$\eta_{\mathsf{max}} -\eta_{\mathsf{min}}<\epsilon$.}
\end{algorithmic}
\end{algorithm}

With beamforming and antenna selection matrices derived in previous sections, we can express the power allocation sub-problem as follows:
\begin{alignat}{2}
\text{($\mathcal{P}1$):~~~~    }&\underset{\bold{P}}{\max}~~~        && \underset{k}{\min}  ~~\mathsf{SINR_k}\label{eq:obj_fun_P}\\
&\text{s.t.} &      &   \eqref{eq:constraint+},\eqref{eq:constraint-sum_no_W}.\nonumber
\end{alignat}
This problem is still non-convex with respect to the power allocation coefficients $\bold{P}$.
\begin{proposition}\label{propos_quasi}
The optimization problem ($\mathcal{P}1$) is a quasi-linear optimization problem. 
\end{proposition}
\begin{proof}
See Appendix \ref{proof_quasi}.
\end{proof}

Since problem ($\mathcal{P}1$) is a quasi-linear optimization problem, its optimal solution can be efficiently found using the bisection
method \cite{boyd2020disciplined}. To this end, we rewrite problem ($\mathcal{P}1$) by adding slack variable $\eta$ as follows:
\begin{alignat}{2}
\text{($\mathcal{P}2$):~~~~    }&\underset{\bold{P},\eta}{\text{max}}~~~ \eta      && \label{eq:obj_fun_P}\\
&\text{s.t.} &      &  \mathsf{SINR_k}\geq \eta, \forall k, \in \mathcal{K},\\&&&\eqref{eq:constraint+},\eqref{eq:constraint-sum_no_W}.\nonumber
\end{alignat}
It is easy to prove that ($\mathcal{P}2$) is equivalent to ($\mathcal{P}1$). We do this by noting that the optimal solution for ($\mathcal{P}2$) must satisfy $\eta=\min (\mathsf{SINR_1},...,\mathsf{SINR_K})=\mathsf{SINR_1}=...=\mathsf{SINR_K}$, which is identical to the optimal solution for ($\mathcal{P}1$); hence, the proof is completed. 
For any given value of $\eta$, ($\mathcal{P}2$) will be a linear feasibility
problem that can be optimally solved using convex optimization techniques, such as the interior-point method \cite{cvx}.
The bisection method is summarized in Algorithm 2. Since each iteration of Algorithm 2 requires
solving only a convex problem, Algorithm 2's overall complexity is polynomial at worst. Algorithm 3 summarizes the methods proposed to solve problem ($\mathcal{P}$) and find variables $\bold{P},\bold{A},\bold{W}$.

\begin{algorithm}[t]
\caption{The methods proposed to solve the problem ($\mathcal{P}$) and find variables $\bold{P},\bold{A}, ~\text{and}~\bold{W}$.}\label{alg3}
\begin{algorithmic}[1]
\Require{$\bold{G} \in \mathbb{R}^{K \times M},\bold{\beta} \in \mathbb{R}^{1\times K},\bold{m} \in \mathbb{R}^{1\times K}, P_{\mathsf{HAPS}} \in \mathbb{R}^{1\times 1}$} 
\Ensure{$\bold{P} \in \mathbb{R}^{1\times K},\bold{A} \in \mathbb{B}^{K \times M},\bold{W} \in \mathbb{C}^{K \times M}$}
\Statex
\State Perform analog beamforming at HAPS utilizing the steering vectors of antennas for each user, i.e., $\mathbf{w}_{k}=\frac{1}{\sqrt{M_k}}\mathbf{b}_k^*$.
\State Derive the achievable SINR of each user by the bound \eqref{SINR} in Proposition \ref{propos_sinr} and update the objective function in ($\mathcal{P}$) accordingly.
\State Calculate antenna selection vectors for each user utilizing the gains of antenna elements as in Algorithm 1.
\State Calculate the optimal power allocation for each user utilizing the bisection method in Algorithm 2.
\end{algorithmic}
\end{algorithm}

\section{Numerical Results}
In this section, we report numerical results to demonstrate that our proposed HAA scheme outperforms baseline array schemes. Simulation parameters are summarized in Table I. In addition, the following default parameters were employed in the simulations, with variations specified as needed.
The carrier frequency was set to $f_\mathrm{C}=2~\mathrm{GHz}$, the communication bandwidth is assumed to be $\mathsf{BW}=20~\mathrm{MHz}$, and the noise power spectral density is $N_0=-174~\mathrm{dBm/Hz}$. We assume that all users are uniformly distributed over a square area measuring $60~\mathrm{km}$ on each side.
 We also assume that the HAPS is deployed in the middle of the square area. Default values for the number of users and antenna elements in HAPS's transmit array are $K=16$ and $M=2650$, respectively. We assume that all antenna elements are placed on a hemisphere with a radius of $d_m=3~\rm{meters}$.
 We set the maximum transmit power at the HAPS to $P_{\mathsf{HAPS}}=50 ~\rm{dBm}=100~\mathrm{W}$. Given an urban area, the excessive path loss affecting air-to-ground links in the LoS and NLoS cases is assumed to be $\eta_{\mathsf{LoS}}^{\mathsf{dB}}=1~\mathrm{dB}$ and $\eta_{\mathsf{NLoS}}^{\mathsf{dB}}=20~\mathrm{dB}$, respectively \cite{Irem2016}. Also, for the urban area, $A=9.61$ and $B=0.16$.
  In Algorithm 2, we use initial values $\eta_{\mathsf{min}}=0$ and $\eta_{\mathsf{max}}=1500$ and set the tolerance value to $\epsilon=0.01$. We also consider the HAPS's altitude to be $20 ~\mathrm{km}$. To calculate each antenna element's gain, we assume that the $3~\mathrm{dB}$ beamwidth is $\theta_{\mathsf{3dB}}=25~\mathrm{degrees}$, while the front-to-back ratio is $\gamma_{\mathsf{max}}=30~\mathrm{dB}$. We also suppose that $M_k=64$ elements are selected for each user.

 In simulation figures, we compare the proposed scheme, referred to as \textbf{hemispherical antenna array (HAA)}, with the following three baseline schemes:
 \begin{itemize}
     \item \textbf{Cylindrical antenna array (CAA):} In this configuration, the identical number of antenna elements, i.e., $M=2650=50\times53$, are placed on a cylindrical structure. In this case, $50$ indicates the number of vertical element layers, while $53$ represents the number of horizontal elements arranged around the circumference of each circular layer.
     \item \textbf{Rectangular antenna array (RAA):} Similar to the HAA and the CAA, this arrangement employs $M=2650=50\times53$ antenna elements. However, in this case, elements are situated on a rectangular surface, with $50$ representing the number of elements along the length of the rectangle, and $53$ representing the number of elements along the width of the rectangle.
     \item \textbf{Hybrid rectangular and cylindrical antenna array (HRCAA):} In this scheme, $M_{\mathsf{cyl}}=2000=50\times40$ antenna elements are placed on a cylindrical structure, with $50$ indicating the number of vertical element layers, $40$ representing the number of horizontal elements around the circumference of each circular layer, and an additional $M_{\mathsf{rect}}=M-M_{\mathsf{cyl}}=650$ elements placed on a rectangular surface underneath the cylinder.
 \end{itemize}

Note that the CAA, RAA, and HRCAA baseline schemes were previously proposed in \cite{Cylindrical_VTC}, \cite{Softbank_giga}, and \cite{Cylindrical_VTC}, respectively. For a fair comparison among all schemes, we utilize the same values for the total number of elements $M$, selected number of elements $M_k$, and $3~\mathrm{dB}$ beamwidth $\theta_{\text{3dB}}$ for the proposed and baseline schemes.
Also, identical parameters, signaling, and channel models are used for the three baseline schemes and the proposed HAA scheme. We also employ the same beamforming, antenna selection, and power allocation techniques for all four schemes. The sole distinction among the four schemes is their elements' angles with respect to terrestrial users, denoted as $\theta_{km}$. This disparity results in variation in antenna element gains and leads to the arrays with different data rates.

\begin{table}[t]
	\caption{Simulation parameters.}
	  \vspace{-11pt}
	\begin{center}
	{\rowcolors{2}{white}{black!10}
		\begin{tabular}{|c|c|} \hline \label{sim_param}
		\textbf{Parameter} & 	\textbf{Value} \\ 
			\hline \hline
		Carrier frequency, $f_\mathrm{C}$  & $2~\mathrm{GHz}$ \\ 
			Length of square area for distribution of users & $60~\mathrm{km}$ \\ 
			Communication bandwidth & $20~\mathrm{MHz}$ \\ 
			Noise power spectral density &  $-174~\rm{dBm/Hz}$\\ 
   Noise figure & $7 ~\rm{dB}$\\ 
			 Number of users, $K$ & $16$ \\ 
			Maximum transmit power at HAPS,  $P_{\mathsf{HAPS}}$ & $100~\mathrm{Watt}$ \\ 
			Number of transmit antenna elements at HAPS, $M$ & $2650$\\ 
   Hemisphere radius & $3~\rm{m}$\\ 
				 
			Front-to-back ratio for each antenna element,  $\gamma_{\mathsf{max}}$ & $30~\mathrm{dB}$\\
  $3~\mathrm{dB}$ beamwidth of each antenna element, $\theta_{\mathsf{3dB}}$ & $25~\mathrm{degrees}$ \\
			Number of selected antenna elements for each user, $M_k$ & 64 \\
			Altitude of HAPS & $20 ~\rm{km}$\\
			\hline
		\end{tabular}}
	\end{center}
	\vspace{-19pt}
\end{table}

\begin{figure}[t]
\begin{subfigure}{0.5\textwidth}
\includegraphics[width=1\linewidth, height=4.5cm]{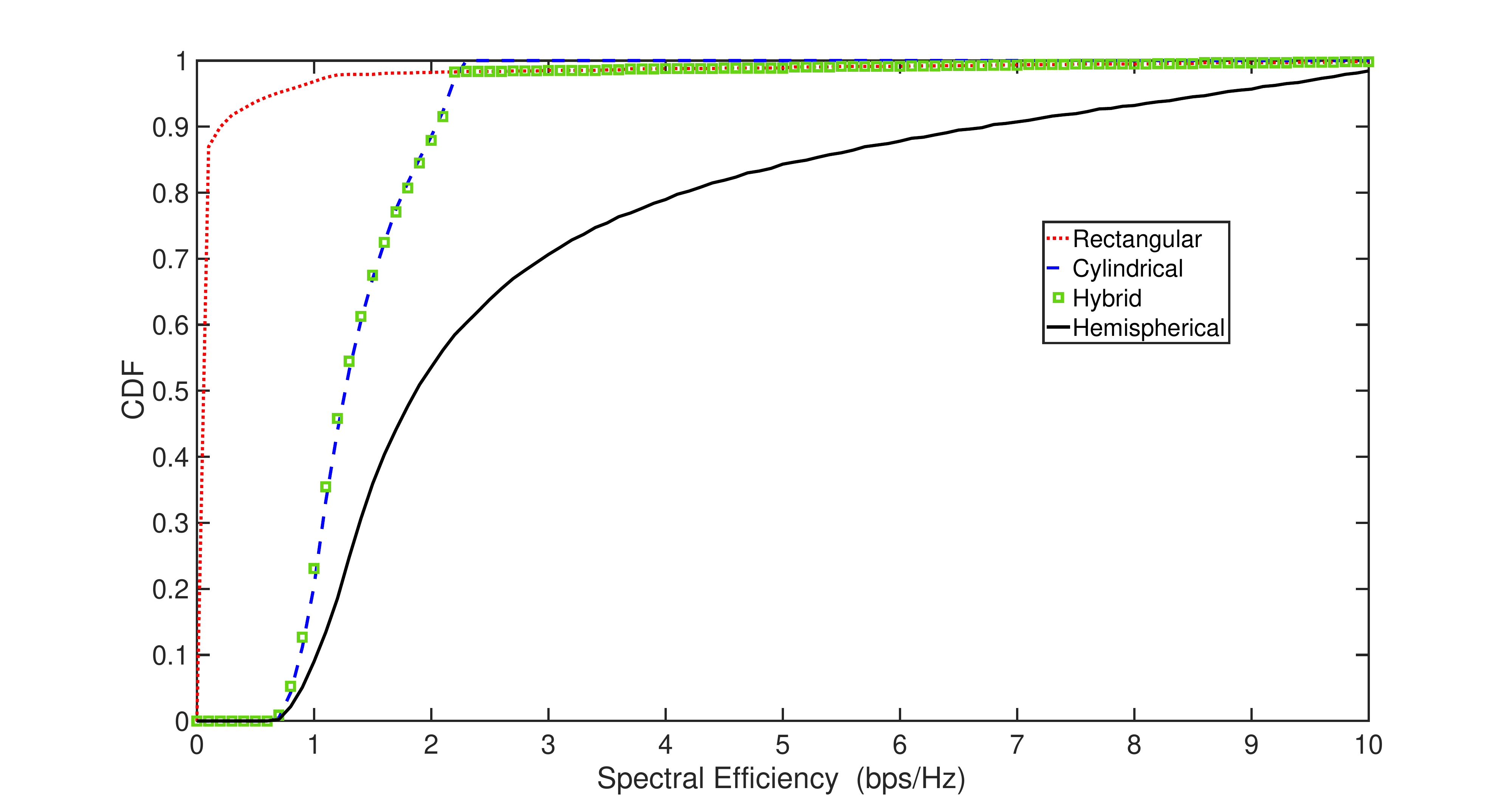} 
\caption{$200 ~\mathrm{km} \times 200 ~\mathrm{km}$}
\label{CDF_WP_200}
\end{subfigure}
\begin{subfigure}{0.5\textwidth}
\includegraphics[width=1\linewidth, height=4.5cm]{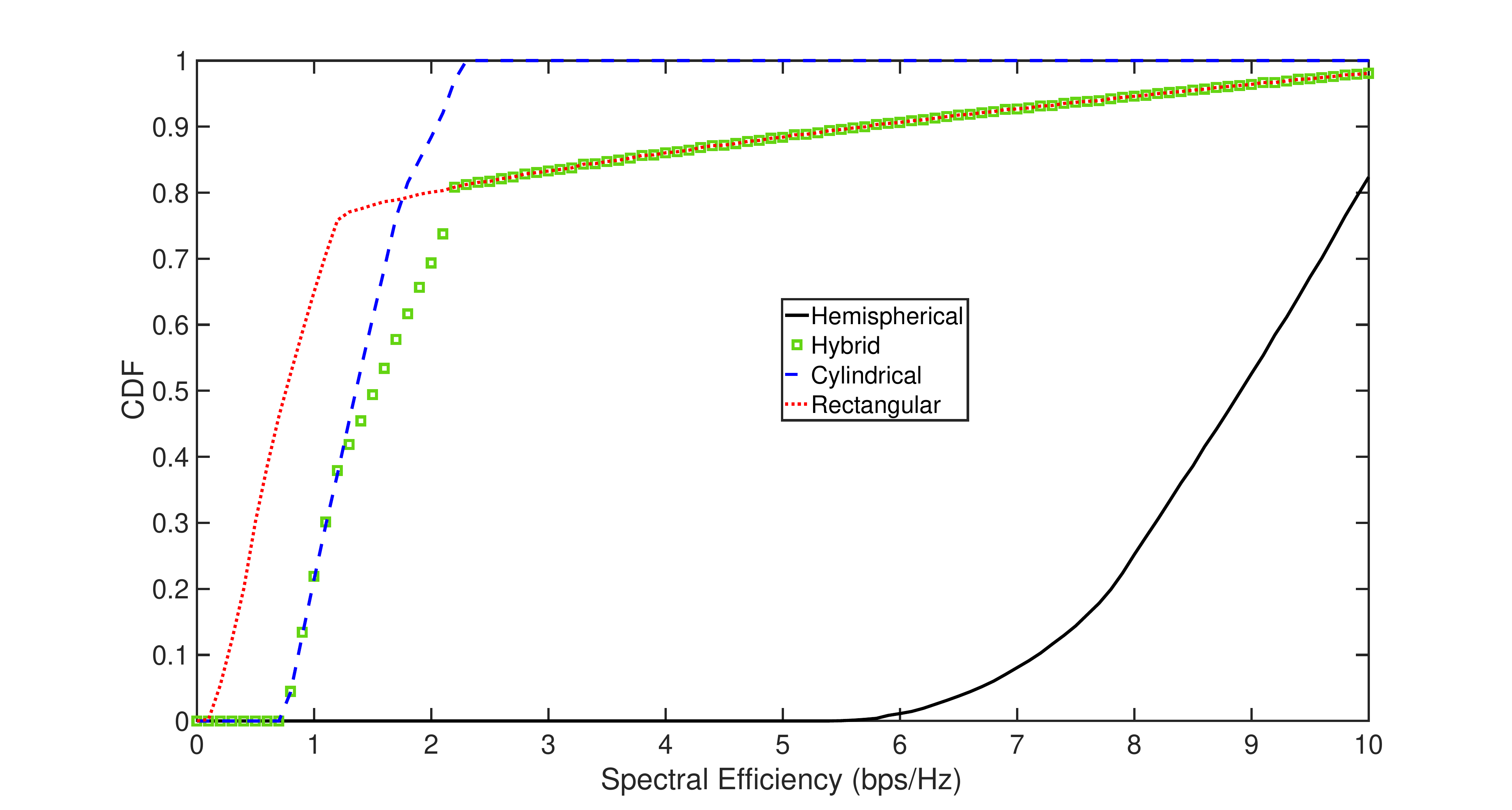}
\caption{$60 ~\mathrm{km} \times 60 ~\mathrm{km}$}
\label{CDF_WP_60}
\end{subfigure}
\caption{
CDF of spectral efficiency for a user uniformly distributed across $10,000$ different locations in two square urban areas with dimensions of $200~\mathrm{km} \times 200~\mathrm{km}$ and $60~\mathrm{km} \times 60~\mathrm{km}$. The user is allocated a fixed power of $1~\mathrm{Watt}$ and served by 64 antenna elements ($M_k=64$).
}
\label{CDF_WP}
\end{figure}

Fig.~\ref{CDF_WP} shows the cumulative distribution function (CDF) of spectral efficiency for a user that is uniformly distributed across $10,000$ different locations in two square urban areas with dimensions $200~\mathrm{km} \times 200~\mathrm{km}$ and $60~\mathrm{km} \times 60~\mathrm{km}$. We assumed the user is allocated a fixed power of $1~\mathrm{Watt}$ and served by 64 antenna elements ($M_k=64$).
As can be seen in both Fig.~\ref{CDF_WP_200} and Fig.~\ref{CDF_WP_60}, the HAA scheme outperforms the three other antenna array configurations.
In Fig.~\ref{CDF_WP_200}, where the coverage area is relatively large, the CAA scheme outperforms the RAA scheme for most users. This is because in the CAA, the antenna elements oriented towards the horizon ensure greater gains for distant users. However, in Fig.~\ref{CDF_WP_60}, the coverage area is smaller, and most users are in a closer proximity to the HAPS. In this case, the RAA scheme, whose antenna elements are oriented towards the nadir of the HAPS, offers users better gains than the CAA scheme.
Importantly, the HAA scheme performs the best since its antennas are directed towards every point on the ground.
Additionally, in both figures, the HRCAA scheme outperforms the RAA and CAA schemes. By capitalizing on the advantages of cylindrical and rectangular surfaces for distant and nearby users, respectively, the HRCAA achieves a CDF upper-bounded by the lower envelope of the CDF of the RAA and CAA schemes.
\begin{figure}[t]
\begin{subfigure}{0.24\textwidth}
\includegraphics[width=\textwidth]{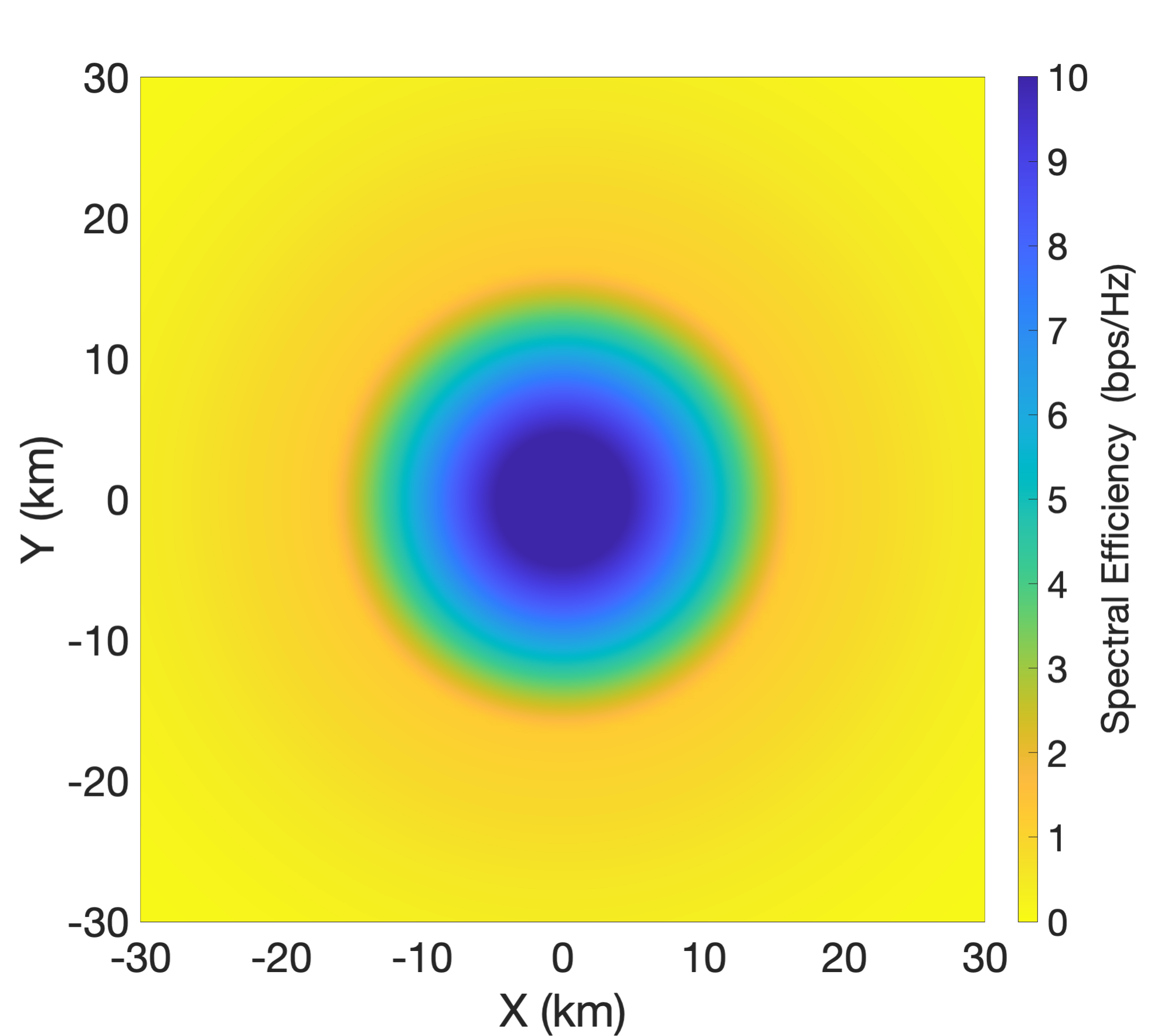} 
\caption{Rectangular}
\label{hear_wp_rect60}
\end{subfigure}
\begin{subfigure}{0.24\textwidth}
\includegraphics[width=1\linewidth]{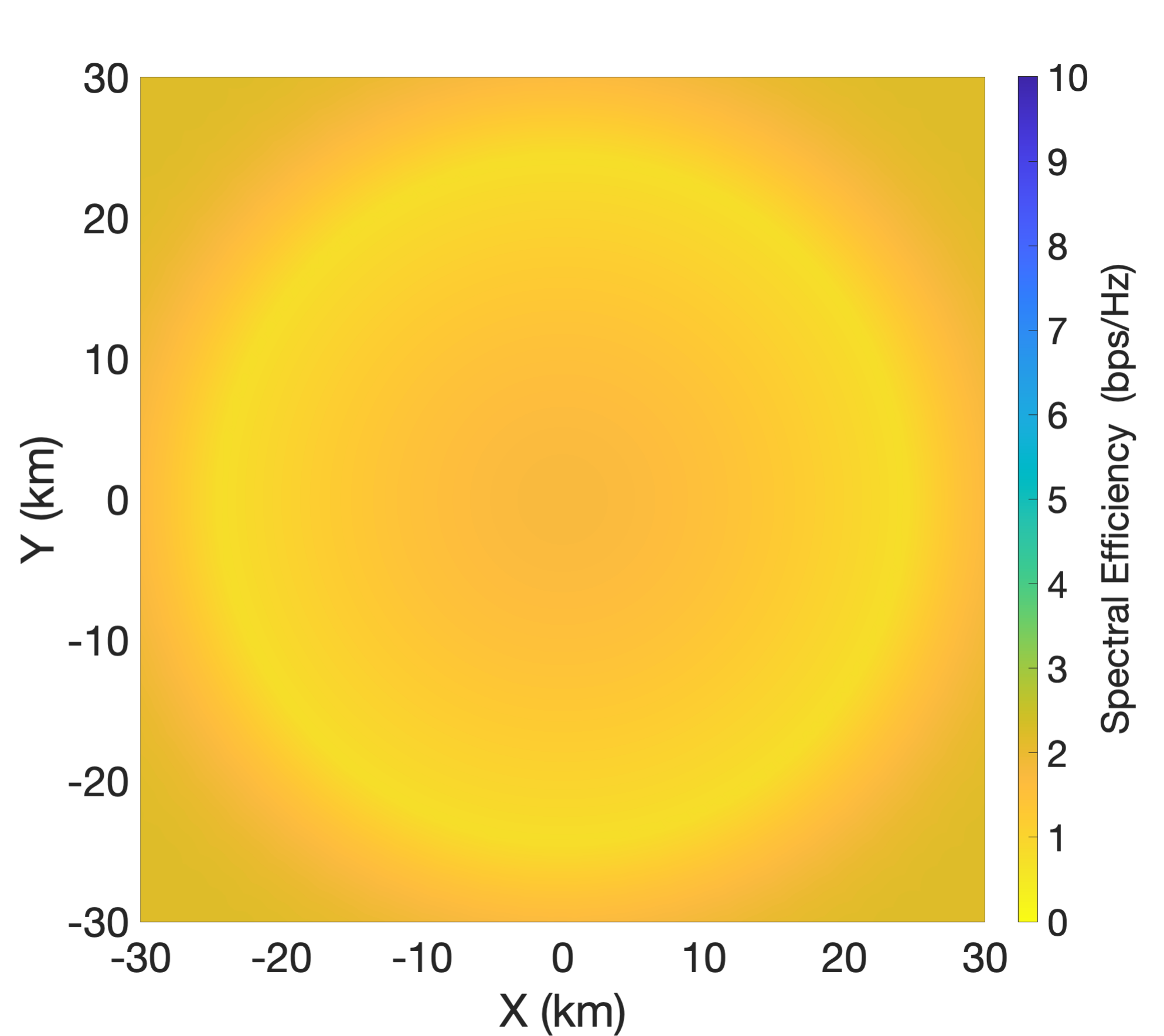} 
\caption{Cylindrical}
\label{heat_wp_cyl60}
\end{subfigure}
\begin{subfigure}{0.24\textwidth}
\includegraphics[width=1\linewidth]{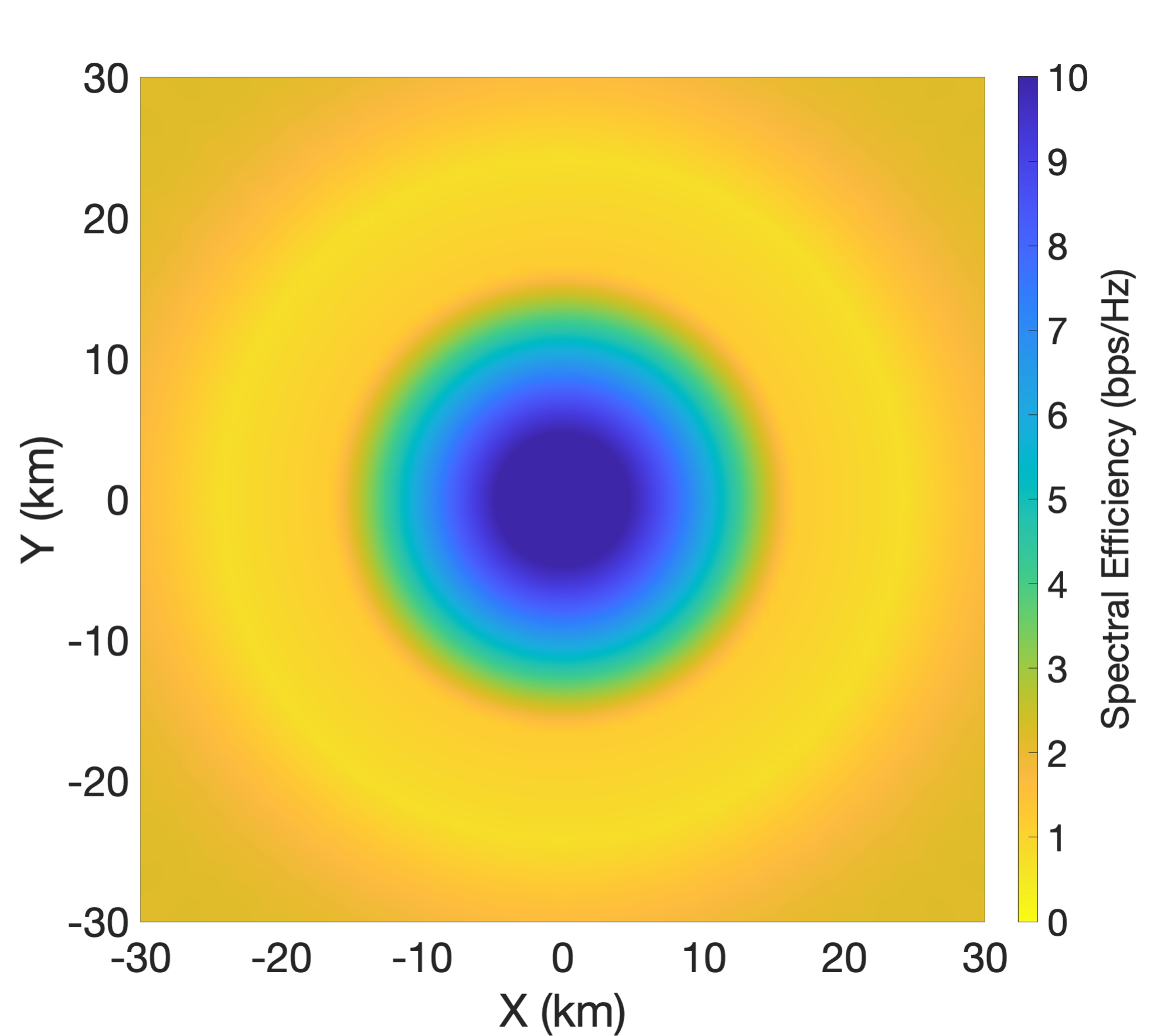}
\caption{Hybrid rectangular and cylindrical}
\label{heat_wp_hybrid60}
\end{subfigure}
\begin{subfigure}{0.24\textwidth}
\includegraphics[width=1\linewidth]{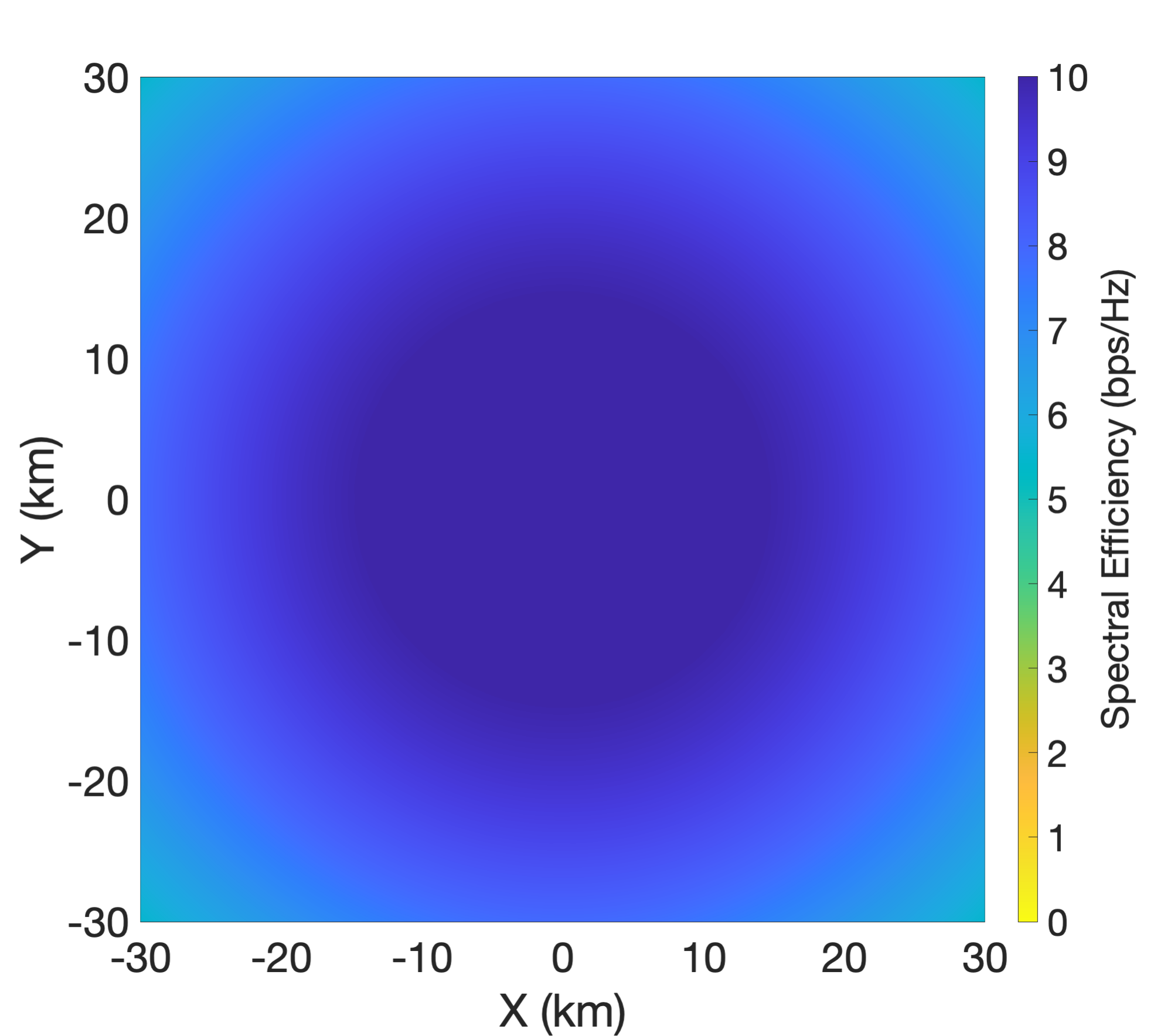}
\caption{Proposed hemispherical}
\label{heat_wp_hemi60}
\end{subfigure}
\caption{
Heatmaps of spectral efficiency for rectangular, cylindrical, hybrid, and hemispherical antenna arrays.  We assumed a user is uniformly distributed across $10,000$ different locations in a square urban area with dimensions $60~\mathrm{km} \times 60~\mathrm{km}$. The user is allocated a fixed power of $1~\mathrm{Watt}$ and served by 64 antenna elements ($M_k=64$).
}
\label{Heat_WP60}
\end{figure}

Fig.~\ref{Heat_WP60} presents heatmaps illustrating the spectral efficiency distribution of the rectangular, cylindrical, hybrid, and hemispherical antenna arrays. For this figure, we assumed a user is uniformly distributed across $10,000$ different locations in a square urban area with dimensions $60~\mathrm{km} \times 60~\mathrm{km}$. The user is allocated a fixed power of $1~\mathrm{Watt}$ and served by 64 antenna elements ($M_k=64$).
As can be seen in Fig.~\ref{hear_wp_rect60}, regions beneath the HAPS exhibit higher performance than distant locations. This outcome is attributable to the orientation of antenna elements in the RAA, which face downward and thus provide maximum gains to the areas directly beneath them.
In Fig.~\ref{heat_wp_cyl60}, the CAA scheme showcases an inverse trend. As the antenna elements are oriented toward the horizon in the CAA, remote areas experience superior spectral efficiencies compared to those of nearby ones due to their higher antenna gains.
Fig.~\ref{heat_wp_hybrid60} illustrates the advantages of the HRCAA scheme. This configuration combines the benefits of the RAA and CAA schemes to offer commendable data rates to both nearby and distant users.
Finally, Fig.~\ref{heat_wp_hemi60} highlights the HAA scheme, which stands out by providing favorable spectral efficiencies across all regions. This result can be attributed to the fact that users at every location can align with antenna elements oriented toward them and have a robust spectral efficiency for all regions.

\begin{figure}[t]
\begin{subfigure}{0.24\textwidth}
\includegraphics[width=\textwidth]{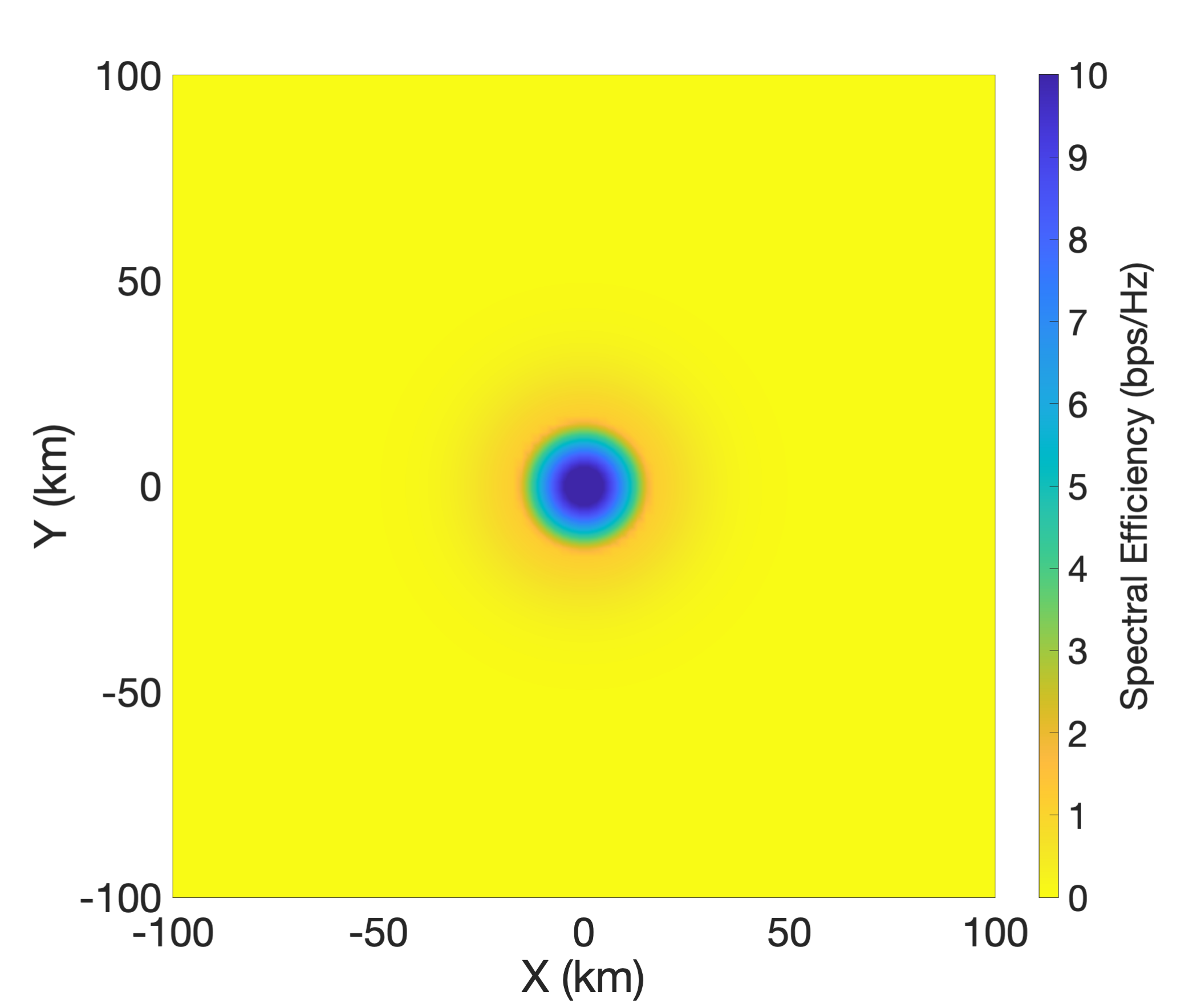} 
\caption{Rectangular}
\label{hear_wp_rect200}
\end{subfigure}
\begin{subfigure}{0.24\textwidth}
\includegraphics[width=1\linewidth]{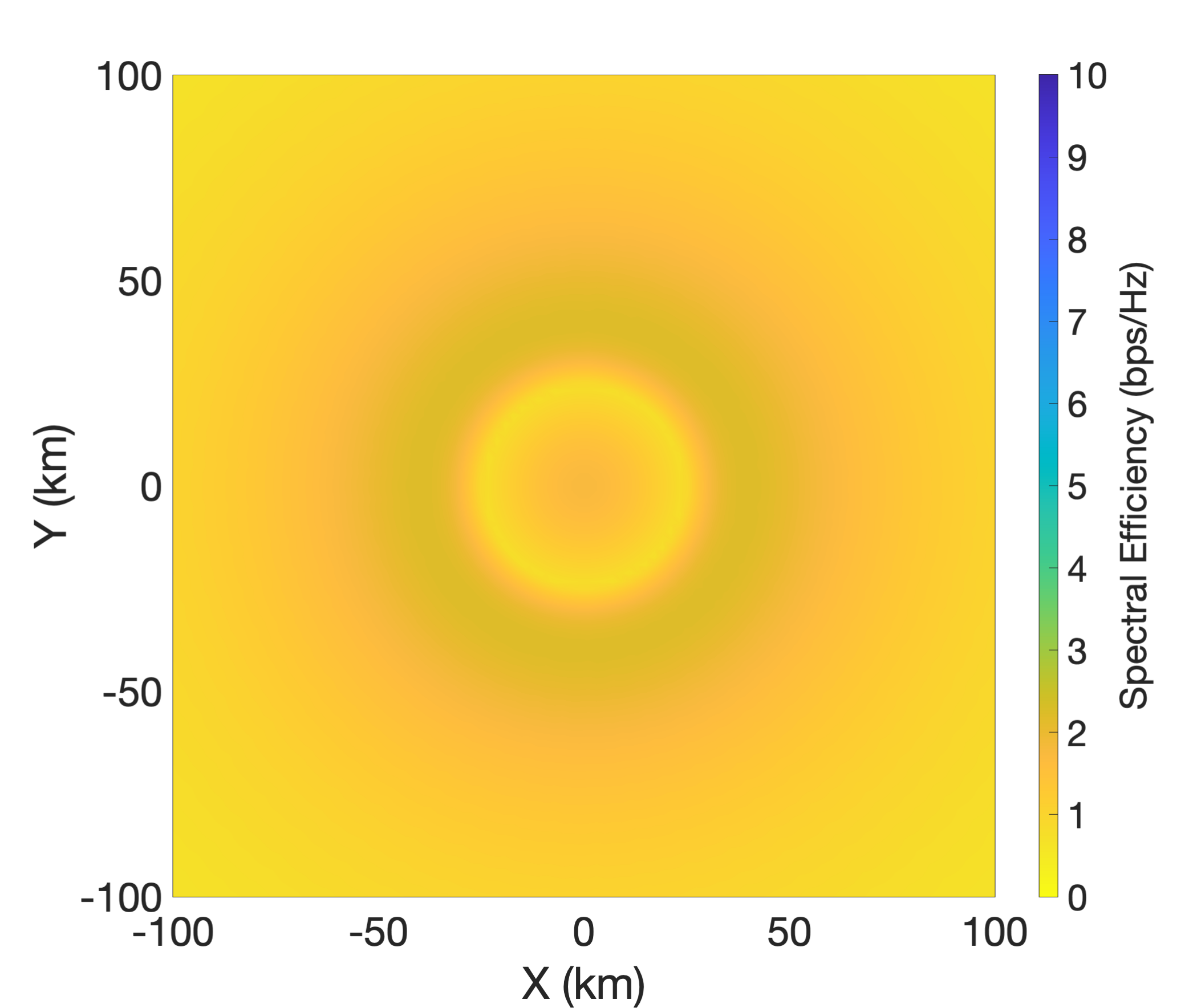} 
\caption{Cylindrical}
\label{heat_wp_cyl200}
\end{subfigure}
\begin{subfigure}{0.24\textwidth}
\includegraphics[width=1\linewidth]{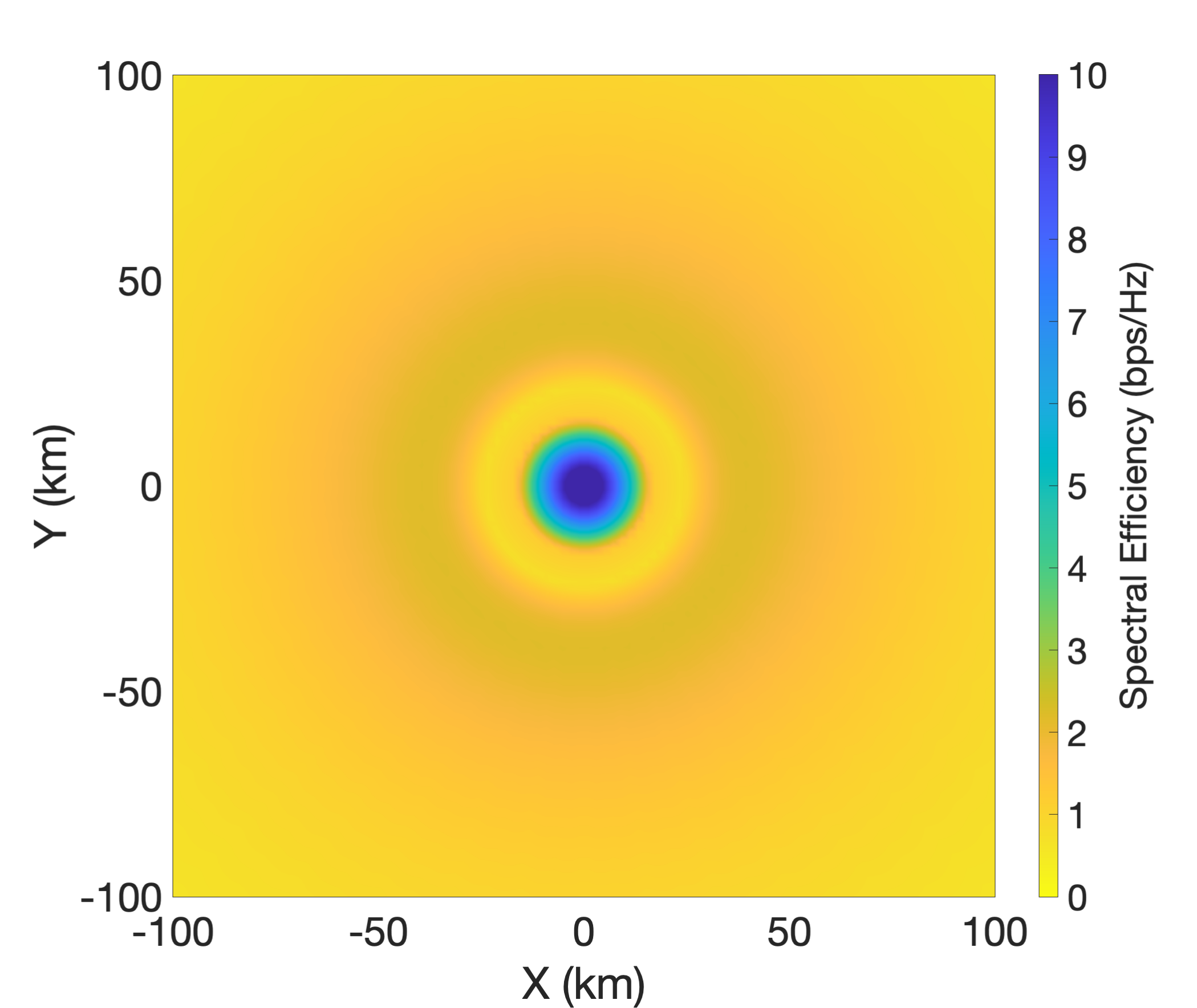}
\caption{Hybrid rectangular and cylindrical}
\label{heat_wp_hybrid200}
\end{subfigure}
\begin{subfigure}{0.24\textwidth}
\includegraphics[width=1\linewidth]{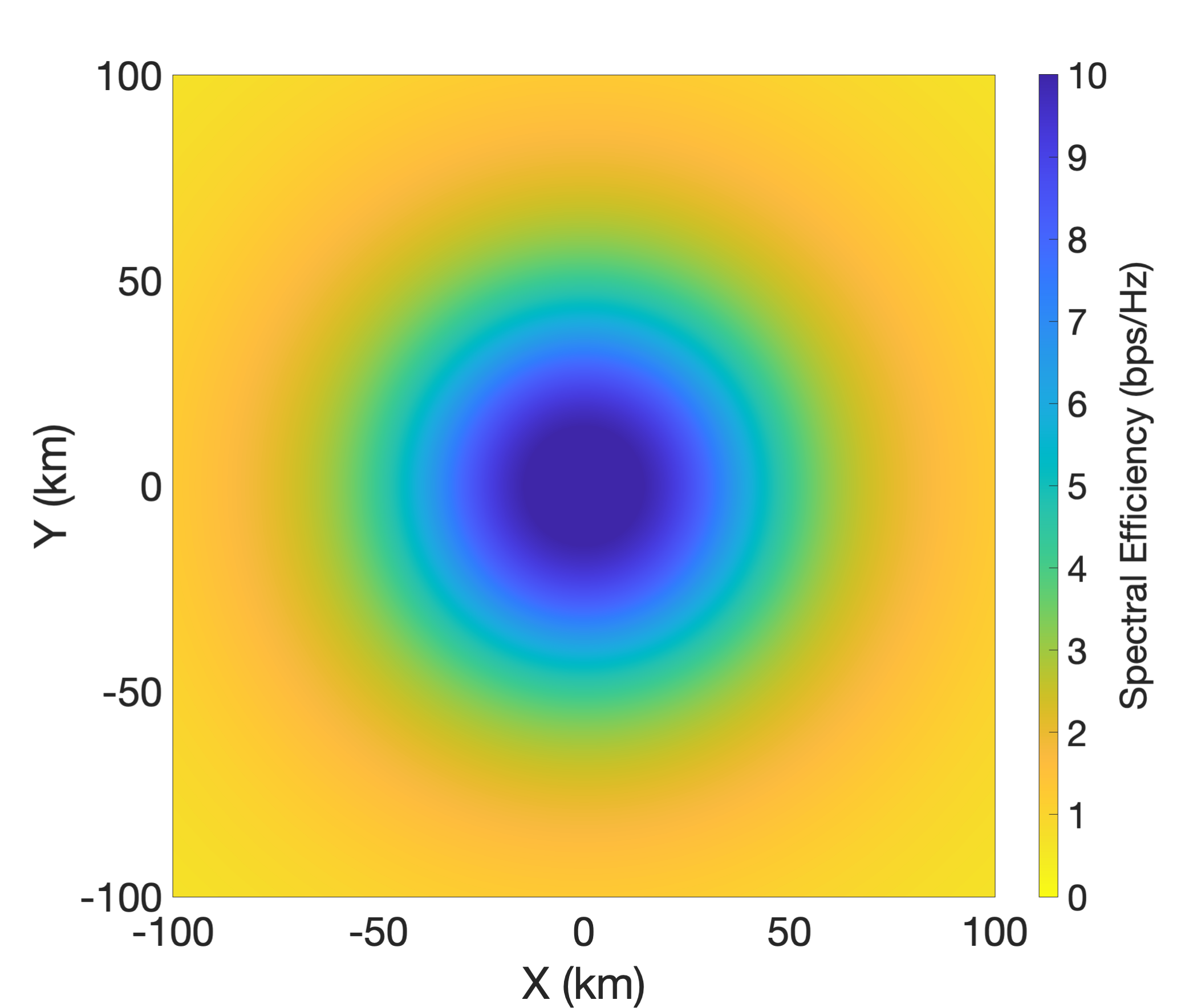}
\caption{Proposed hemispherical}
\label{heat_wp_hemi200}
\end{subfigure}
\caption{Heatmaps of spectral efficiency for rectangular, cylindrical, hybrid, and hemispherical antenna arrays.  We assumed a user is uniformly distributed across $10,000$ different locations in a square urban area with dimensions $200~\mathrm{km} \times 200~\mathrm{km}$. The user is allocated a fixed power of $1~\mathrm{Watt}$ and served by 64 antenna elements ($M_k=64$).}
\label{Heat_WP_200}
\end{figure}

Fig.~\ref{Heat_WP_200} displays another set of heatmaps of spectral efficiency distribution of the rectangular, cylindrical, hybrid, and hemispherical antenna arrays. This time, we assumed a user is uniformly distributed across $10,000$ different locations in a square urban area with dimensions $200~\mathrm{km} \times 200~\mathrm{km}$. The user is allocated a fixed power of $1~\mathrm{Watt}$ and served by 64 antenna elements ($M_k=64$).
Fig.~\ref{hear_wp_rect200} and Fig.~\ref{heat_wp_cyl200} show that peak spectral efficiencies of the RAA and CAA schemes are achieved in the regions with small and medium-sized radii, respectively. However, remote users experience much lower spectral efficiencies due to substantial path loss.
As can be seen in Fig.~\ref{heat_wp_hybrid200}, the HRCAA scheme can provide peak spectral efficiencies to regions with both small and medium-sized radii. Finally, in Fig.~\ref{heat_wp_hemi200}, the HAA scheme stands out for its capacity to deliver competitive spectral efficiencies across a broader region than the three baseline schemes. 
In the HAA scheme, antenna elements' gains uniformly influence the spectral efficiencies of all users. Therefore, the path loss is the only factor that impacts spectral efficiencies. Note that as one can see in the figure, the proposed hemispherical scheme has much better rates as compared to the baseline schemes. Said differently, in order to provide a specific rate, the proposed scheme requires less power consumption, which is essential for HAPS with its limited power resources.

\begin{figure}[t]
\begin{subfigure}{0.24\textwidth}
\includegraphics[width=\textwidth]{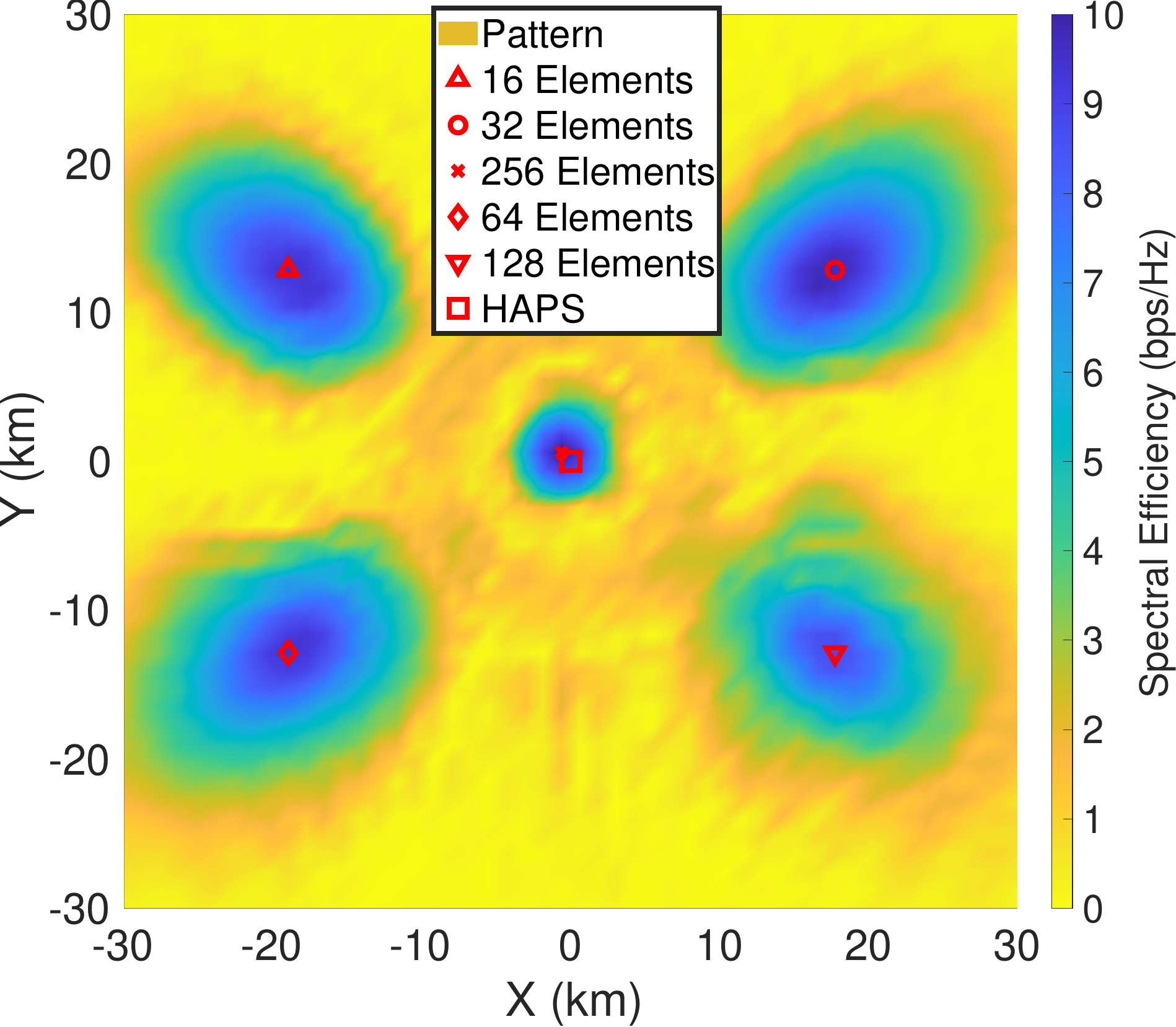} 
\caption{$\theta_{\text{3dB}}=10$}
\label{hear_M_10}
\end{subfigure}
\begin{subfigure}{0.24\textwidth}
\includegraphics[width=1\linewidth]{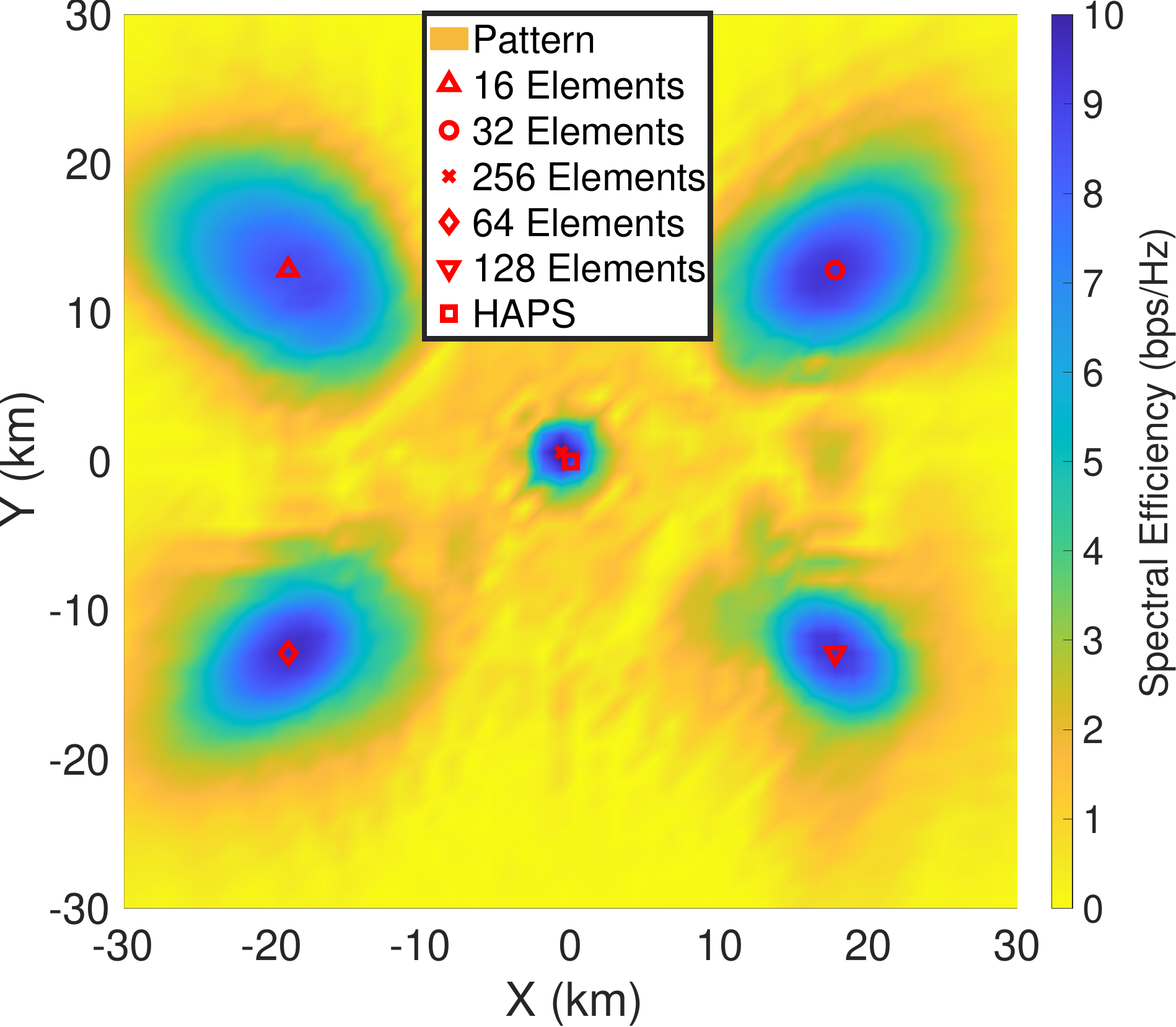} 
\caption{$\theta_{\text{3dB}}=15$}
\label{heat_M_15}
\end{subfigure}
\begin{subfigure}{0.24\textwidth}
\includegraphics[width=1\linewidth]{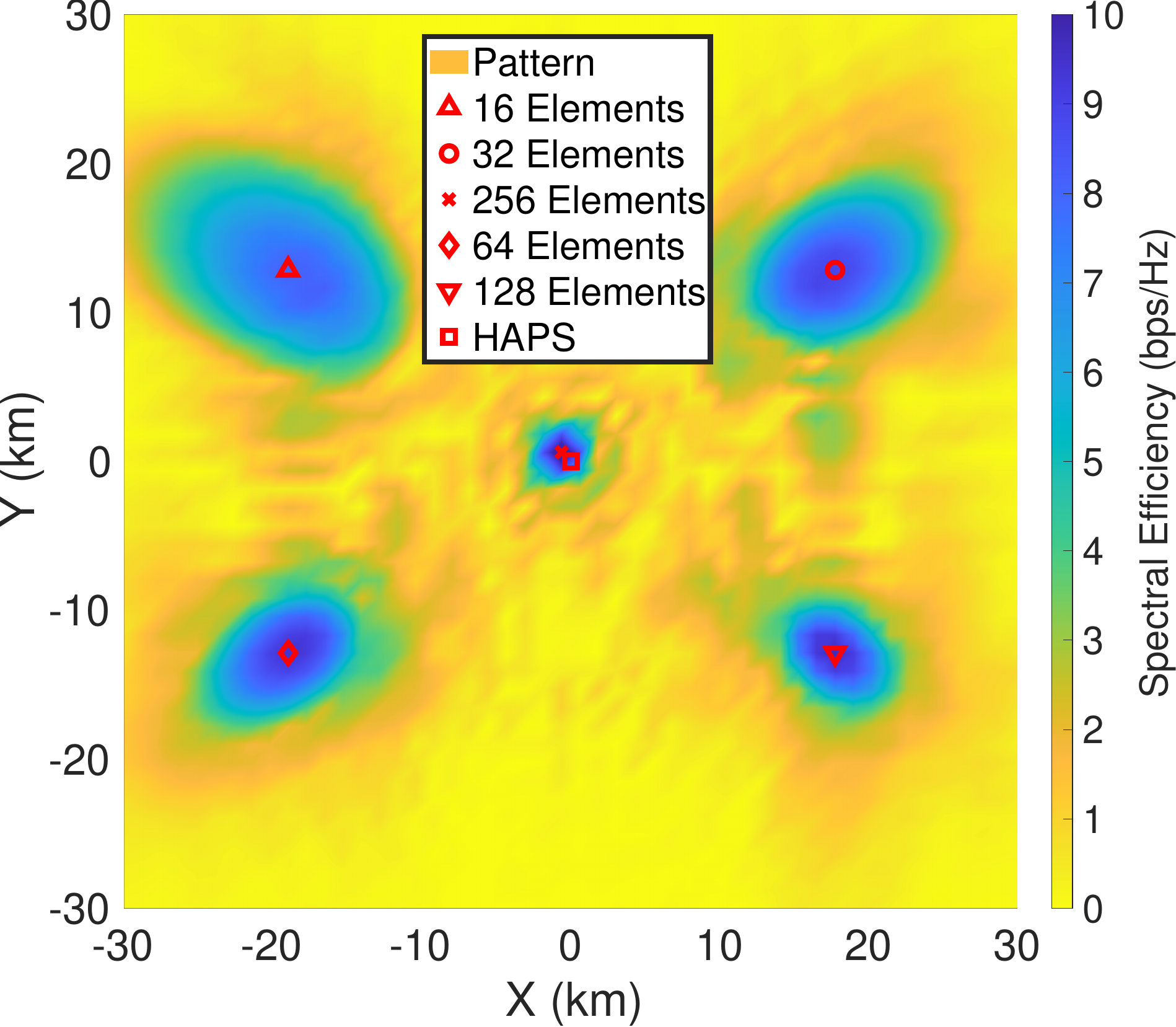} 
\caption{$\theta_{\text{3dB}}=20$}
\label{heat_M_20}
\end{subfigure}
\begin{subfigure}{0.24\textwidth}
\includegraphics[width=1\linewidth]{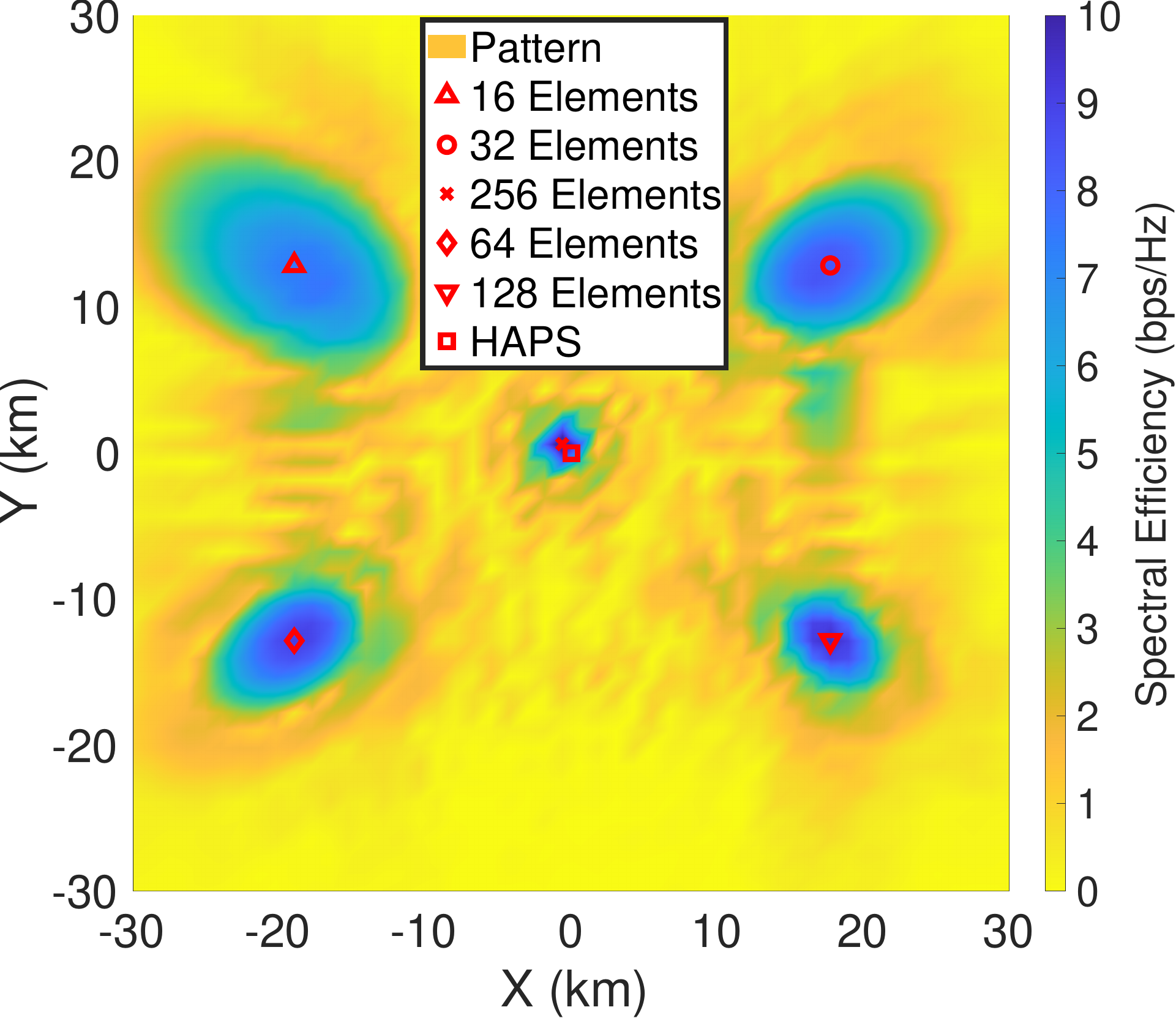} 
\caption{$\theta_{\text{3dB}}=25$}
\label{heat_M_25}
\end{subfigure}
\caption{
Heatmaps displaying the spectral efficiencies for five beams with different numbers of selected antenna elements  are presented for the proposed HAA. Each beam is allocated a fixed power of $1~\mathrm{Watt}$. These heatmaps consider four distinct $3~\mathrm{dB}$ beamwidth values for the antenna elements.
}
\label{Heat_M_element}
\end{figure}

Fig.~\ref{Heat_M_element} presents a set of heatmaps showing the spectral efficiencies for five beams with different numbers of antenna elements selected, i.e., $M_k=16, 32, 64, 128,$ and $ 256$, for the proposed HAA scheme. Each beam is allocated a fixed power of $1~\mathrm{Watt}$.
In Fig.~\ref{hear_M_10}, we visualize the beams with a $3~\mathrm{dB}$ beamwidth of $\theta_{\mathsf{3dB}}=10$. As shown in the figure, in cases where the beams of single antenna elements are narrow, increasing the number of elements selected from $M_k=16$ to $M_k=256$ does not enhance the array gain. This finding indicates that when the $3~\mathrm{dB}$ beamwidth of elements is small, selecting a massive number of elements for beam creation may not yield substantial benefits. For instance, when $M_k=128$, the power of the beam is equally divided among all $128$ elements selected. However, most antennas exhibit negligible gain for the beam, which makes it impractical to allocate power to them. Conversely, when $M_{k}=16$, the power is distributed among antennas that offer significant gains for the beam.
Fig.~\ref{heat_M_15}, Fig.~\ref{heat_M_20}, and Fig.~\ref{heat_M_25} illustrate the footprints of beams with $3~\mathrm{dB}$ beamwidth of $\theta_{\mathsf{3dB}}=15$, $\theta_{\mathsf{3dB}}=20$, and $\theta_{\mathsf{3dB}}=25$, respectively. It becomes apparent that widening the beams of each element makes it possible to achieve an array gain with a high number of antenna elements selected. For instance, in Fig.~\ref{heat_M_25} where $\theta_{\mathsf{3dB}}=25$, a narrow beam with a high array gain is created with $M=128$ antenna elements. It is worth noting that, as one can see in this figure, there is increased interference on neighboring beams when $3~\mathrm{dB}$ beamwidth is a higher value.
Therefore, our results highlight that, when the $3~\mathrm{dB}$ beamwidth value is smaller, array (beamforming) gains cannot be effectively harnessed and it is more advantageous to create beams with a smaller number of antenna elements selected. However, with higher $3~\mathrm{dB}$ beamwidth values, beamforming becomes a more potent tool that supports the creation of very narrow high-gain beams utilizing a larger number of antenna elements.

\begin{figure}[t]
\begin{subfigure}{0.24\textwidth}
\includegraphics[width=\textwidth]{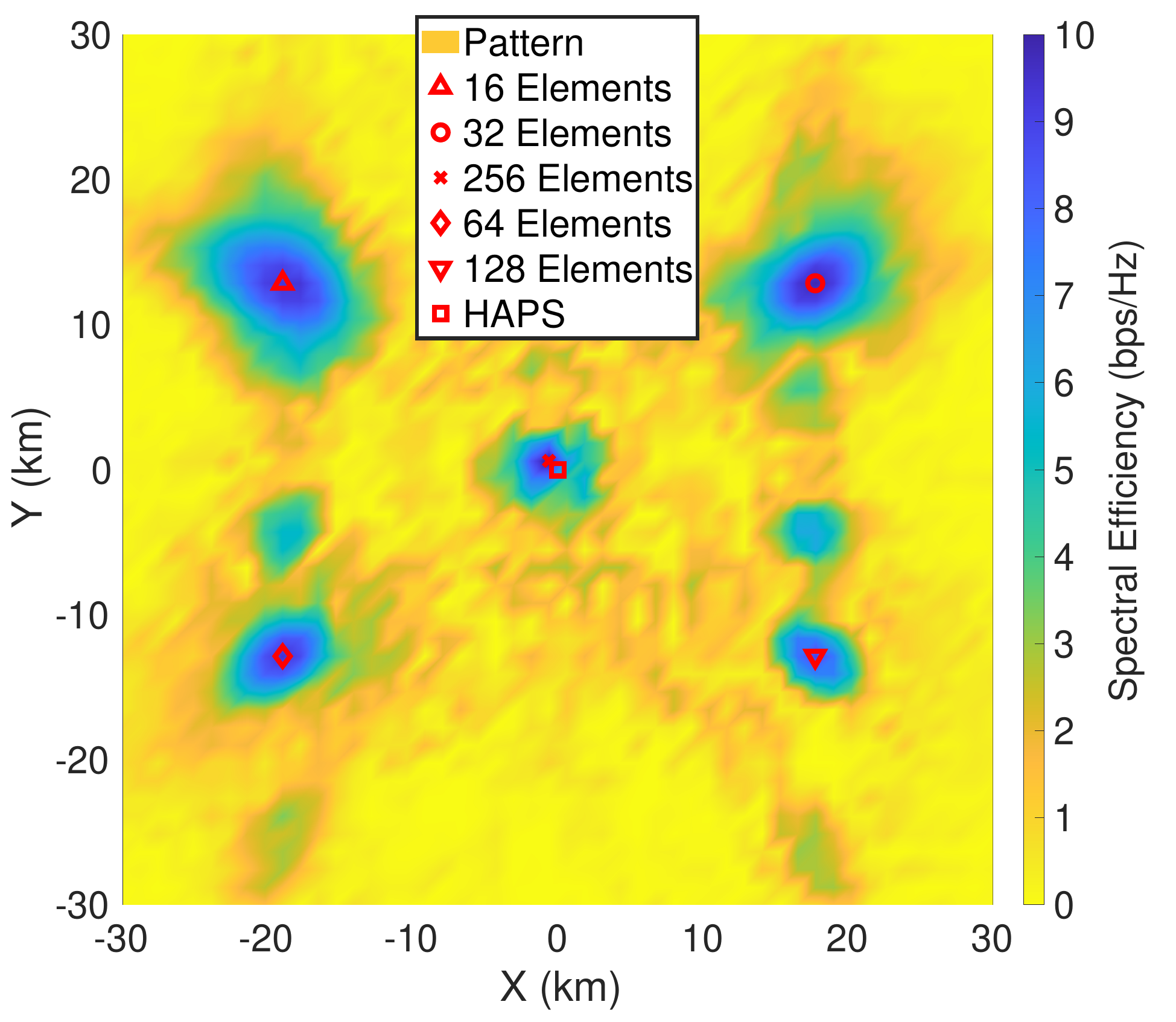} 
\caption{$6~\mathrm{m}$ radius, $\theta_{\text{3dB}}=10$}
\label{heat_r6_10d}
\end{subfigure}
\begin{subfigure}{0.24\textwidth}
\includegraphics[width=1\linewidth]{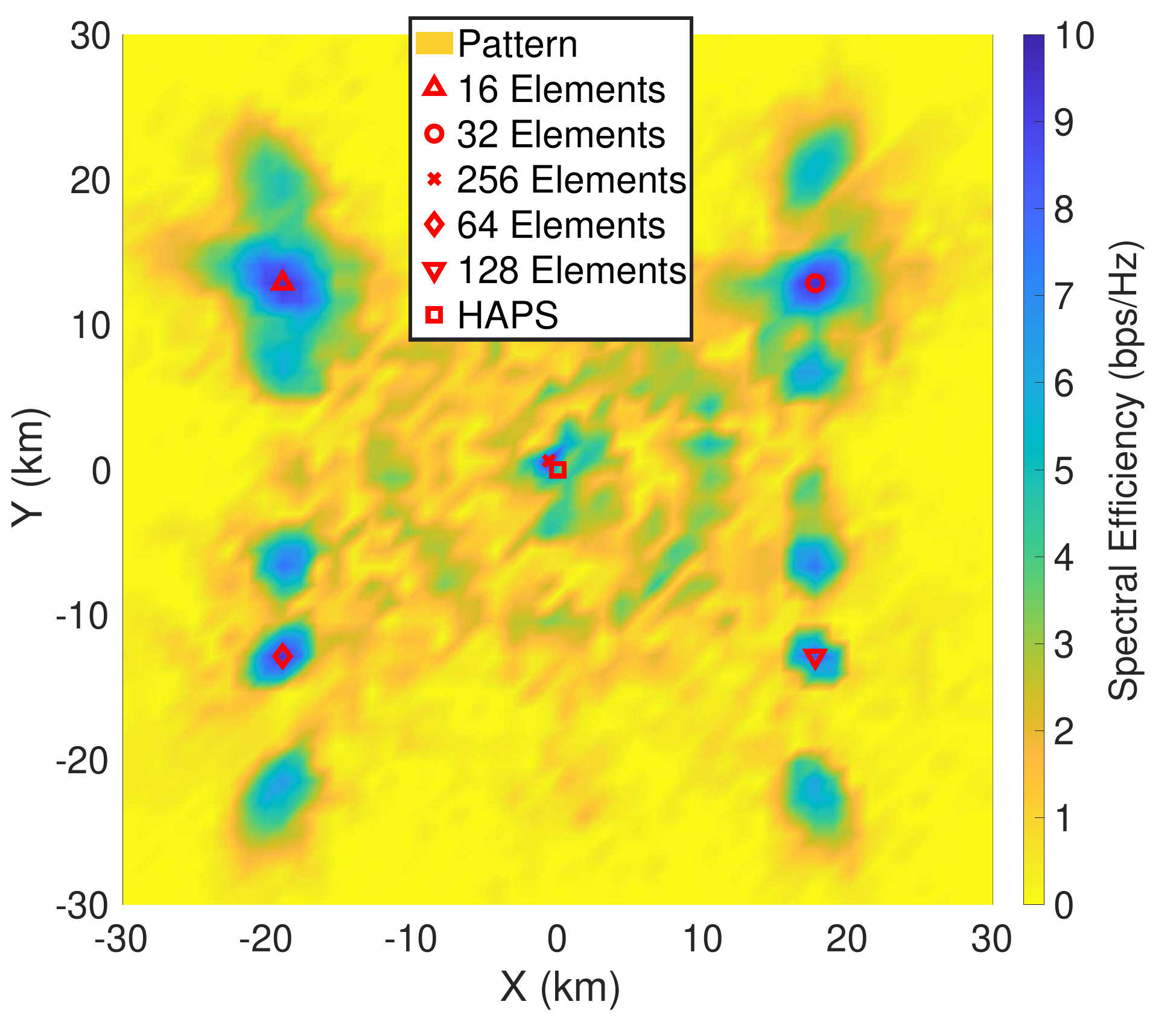} 
\caption{$9~\mathrm{m}$ radius, $\theta_{\text{3dB}}=10$}
\label{heat_r9_10d}
\end{subfigure}
\begin{subfigure}{0.24\textwidth}
\includegraphics[width=1\linewidth]{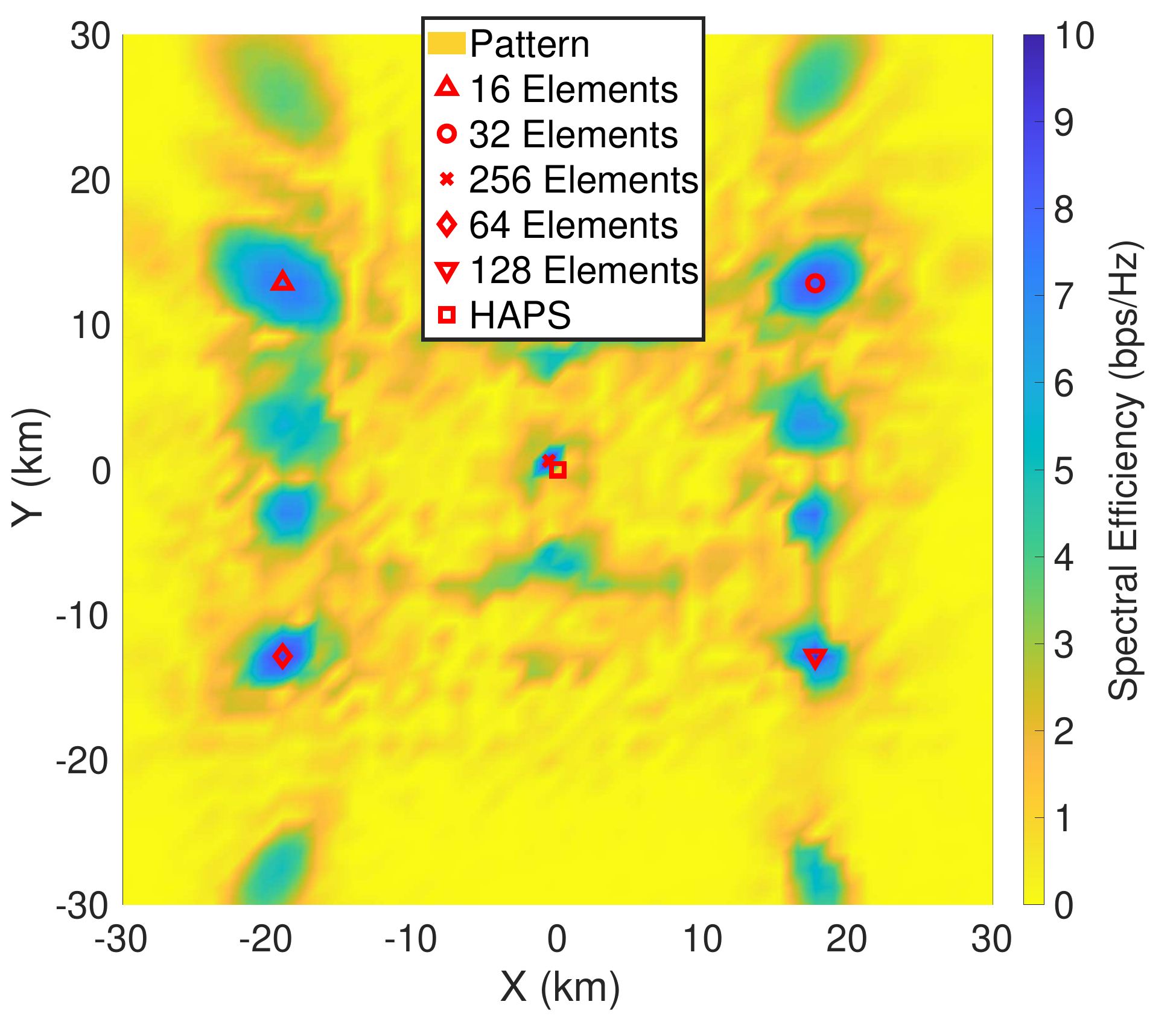} 
\caption{$6~\mathrm{m}$ radius, $\theta_{\text{3dB}}=25$}
\label{heat_r6_25d}
\end{subfigure}
\begin{subfigure}{0.24\textwidth}
\includegraphics[width=1\linewidth]{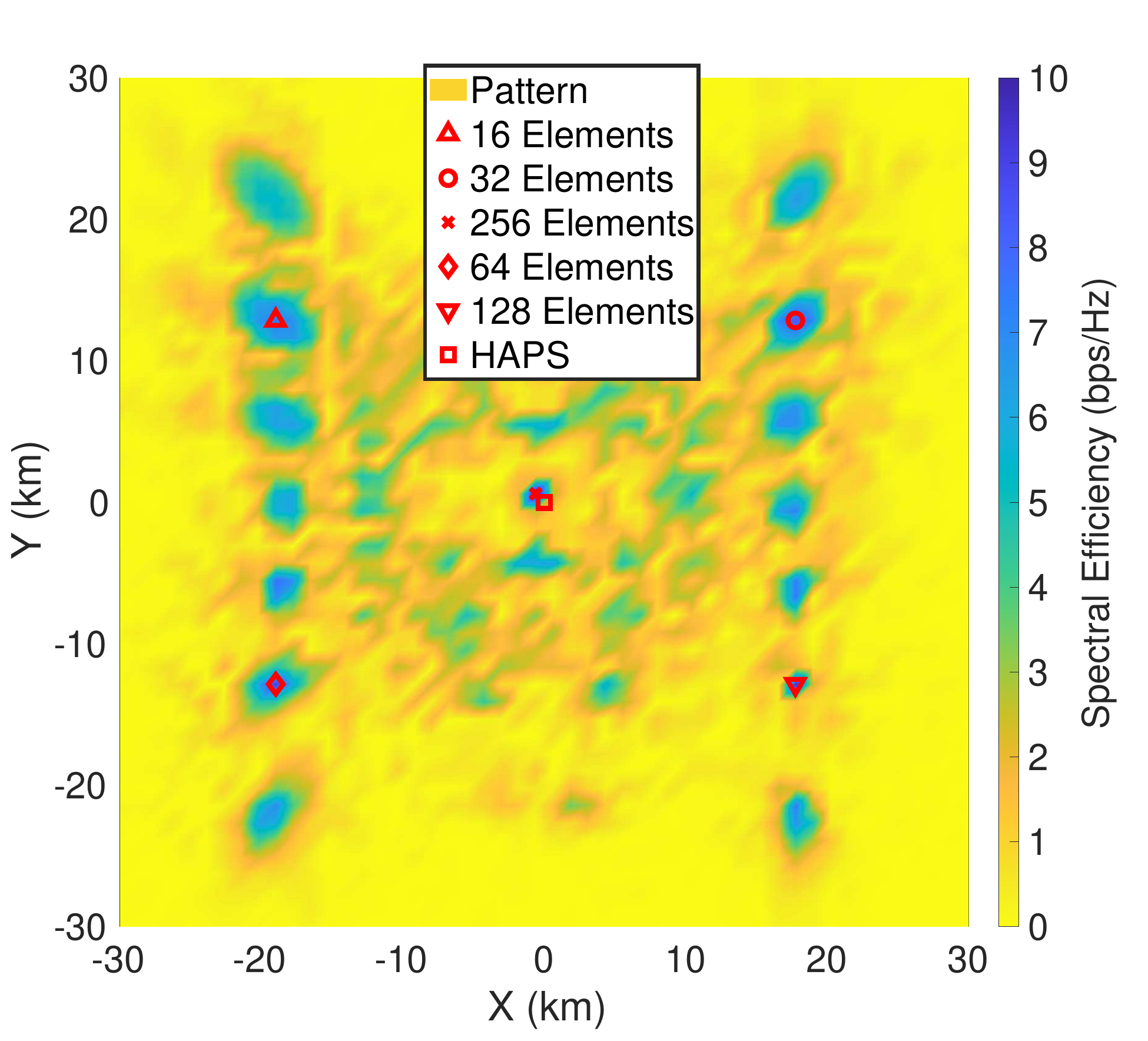} 
\caption{$9~\mathrm{m}$ radius, $\theta_{\text{3dB}}=25$}
\label{heat_r9_25d}
\end{subfigure}
\caption{
Heatmaps of spectral efficiencies for five beams with different numbers of selected antenna elements are presented for the proposed HAA. Each beam is allocated a fixed power of $1~\mathrm{Watt}$. The heatmaps consider two distinct $3~\mathrm{dB}$ beamwidth values for the antenna elements, and two distinct radii for the hemispherical antenna array.
}
\label{Heat_radius}
\end{figure}

 Fig.~\ref{Heat_radius} shows a set of heatmaps displaying the patterns of 
spectral efficiencies for five beams with different numbers of selected antenna elements — i.e., $M_k=16, 32, 64, 128,$ and $ 256$ — for the proposed HAA scheme. Five red markers represent the centers of the five beams formed by the HAPS on the ground . 
In this figure, as indicated by the color bar, the spectral efficiency ranges from $0 ~\rm{bps/Hz}$ (light yellow) to $10 ~\rm{bps/Hz}$ (dark blue). Each beam is allocated a fixed power of $1~\mathrm{Watt}$.
For a $3~\mathrm{dB}$ beamwidth of $\theta_{\mathsf{3dB}}=10$, in Fig.~\ref{heat_r6_10d} and Fig.~\ref{heat_r9_10d}, we visualize the beams with a hemisphere radius of $6 ~\rm{m}$ and $9 ~\rm{m}$, respectively. Furthermore, for a $3~\mathrm{dB}$ beamwidth of $\theta_{\mathsf{3dB}}=25$, Fig.~\ref{heat_r6_25d} and Fig.~\ref{heat_r9_25d} show the beams with a hemisphere radius of $6 ~\rm{m}$ and $9 ~\rm{m}$, respectively. We can see that increasing the radius of the hemisphere, results in narrower beams with stronger grating lobes \cite{bjornson2024introduction} for both $3~\mathrm{dB}$ beamwidth values. This is due to the fact that, with an increase in the radius, the distance between the antenna elements increases as well. Note that for $\theta_{\mathsf{3dB}}=25$ case, due to a better array gain, the beams are narrower compared with the $\theta_{\mathsf{3dB}}=10$ case. However, due to their wider beams for each element, $\theta_{\mathsf{3dB}}=25$ case results in strong grating lobes.

\begin{figure}[t]
\begin{subfigure}{0.24\textwidth}
\includegraphics[width=\textwidth]{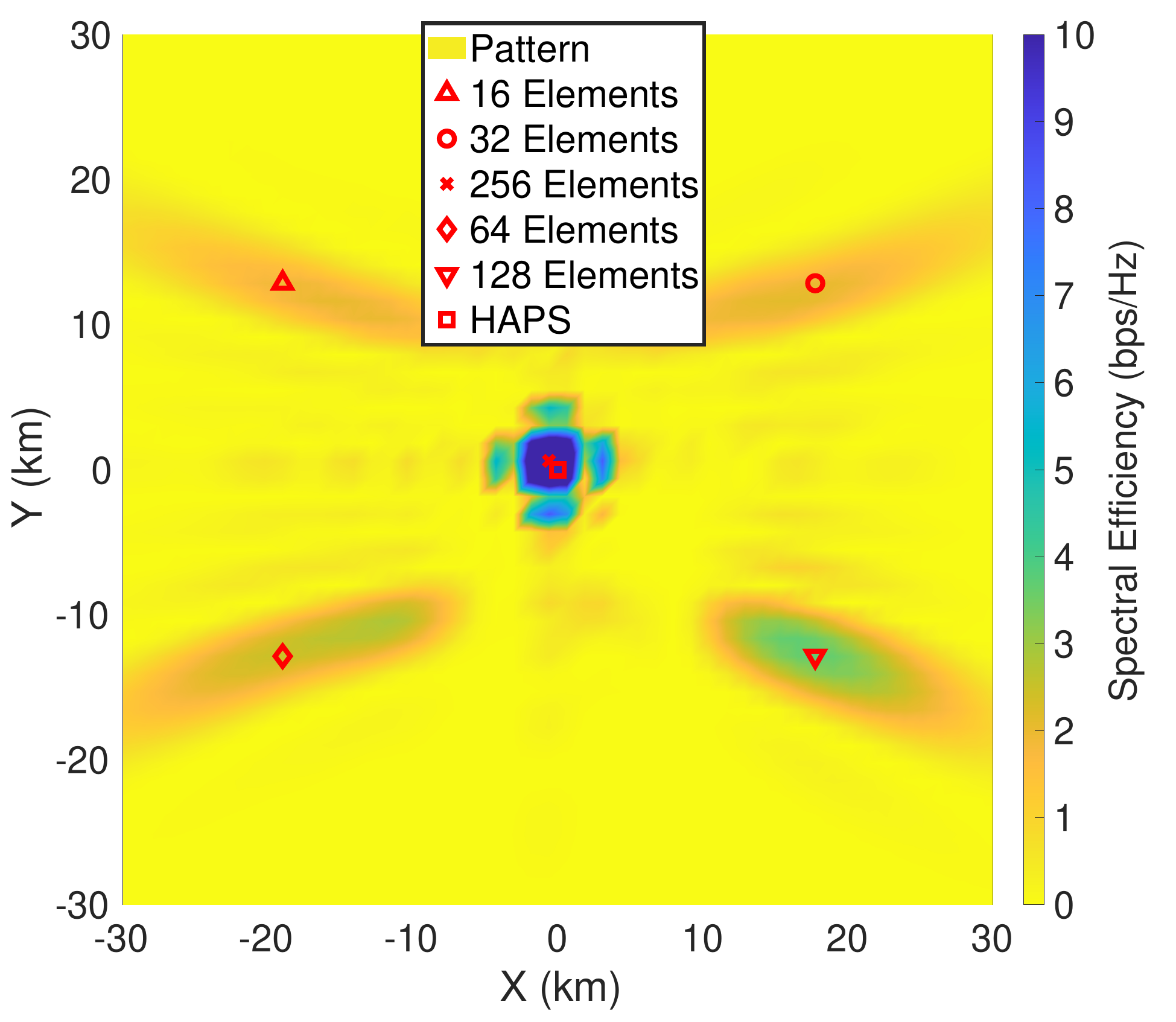} 
\caption{Rectangular, $\theta_{\text{3dB}}=10$}
\label{heat_M_R10}
\end{subfigure}
\begin{subfigure}{0.24\textwidth}
\includegraphics[width=1\linewidth]{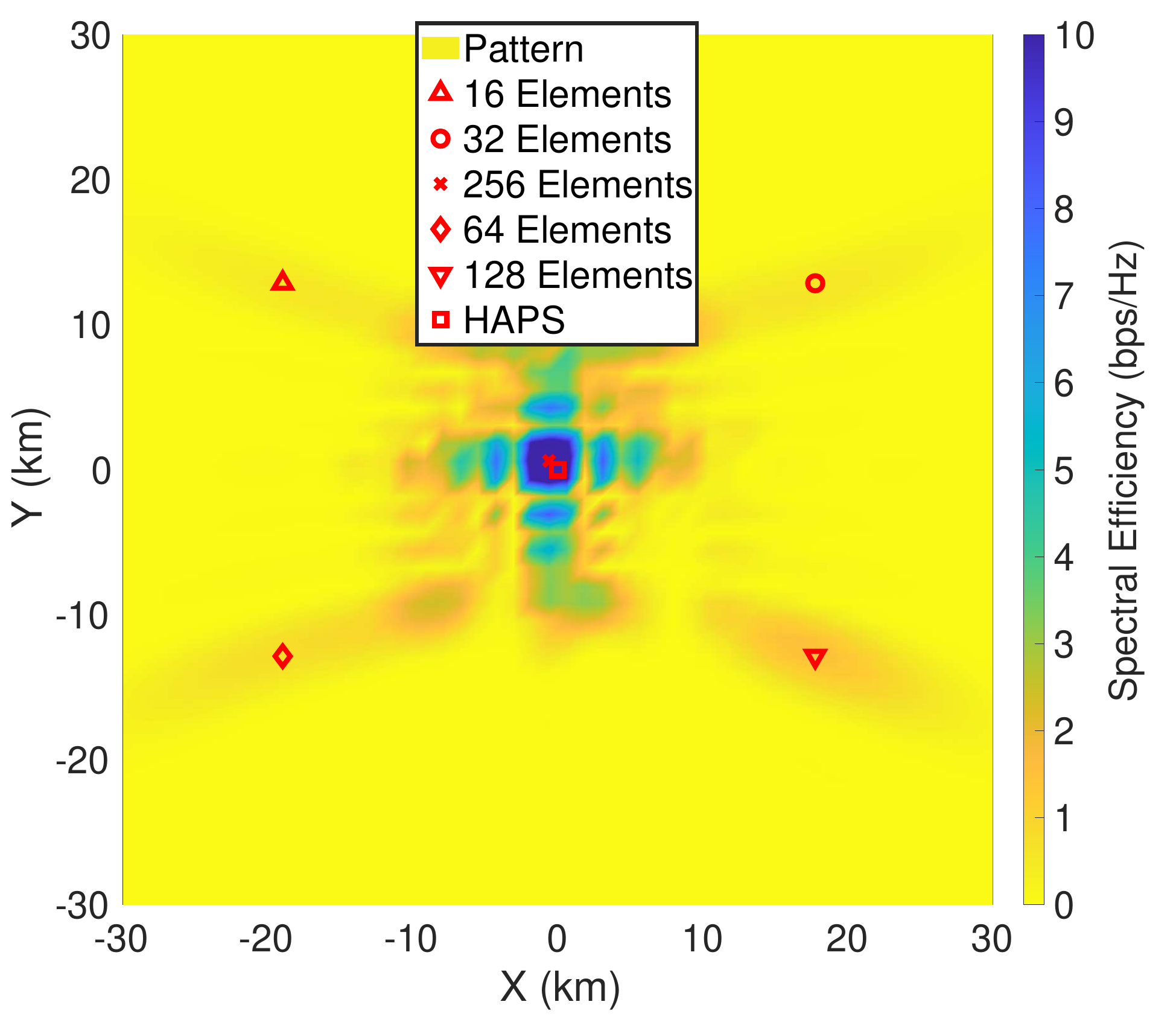} 
\caption{Rectangular, $\theta_{\text{3dB}}=25$}
\label{heat_M_R25}
\end{subfigure}
\begin{subfigure}{0.24\textwidth}
\includegraphics[width=1\linewidth]{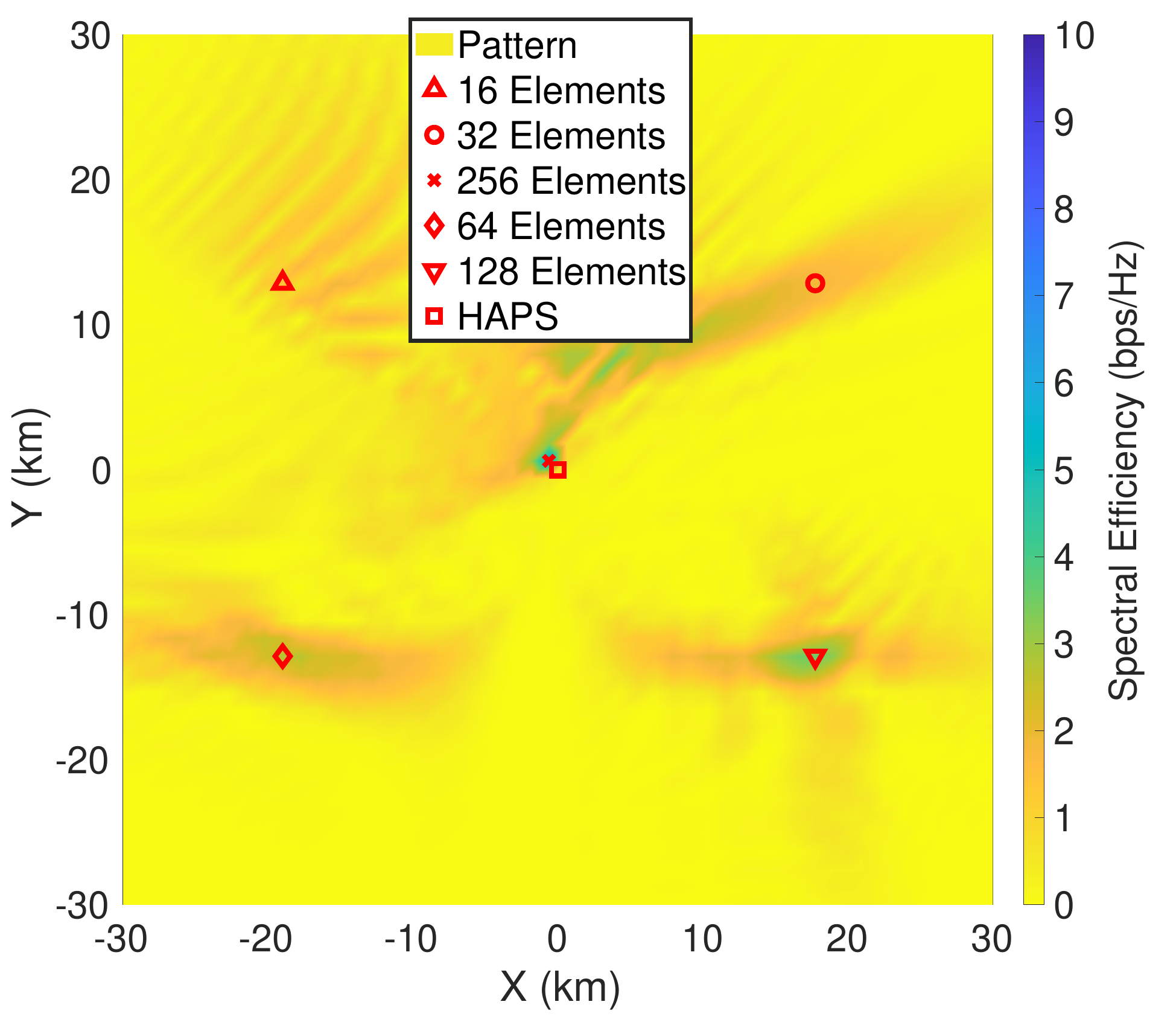} 
\caption{Cylindrical, $\theta_{\text{3dB}}=10$}
\label{heat_M_C10}
\end{subfigure}
\begin{subfigure}{0.24\textwidth}
\includegraphics[width=1\linewidth]{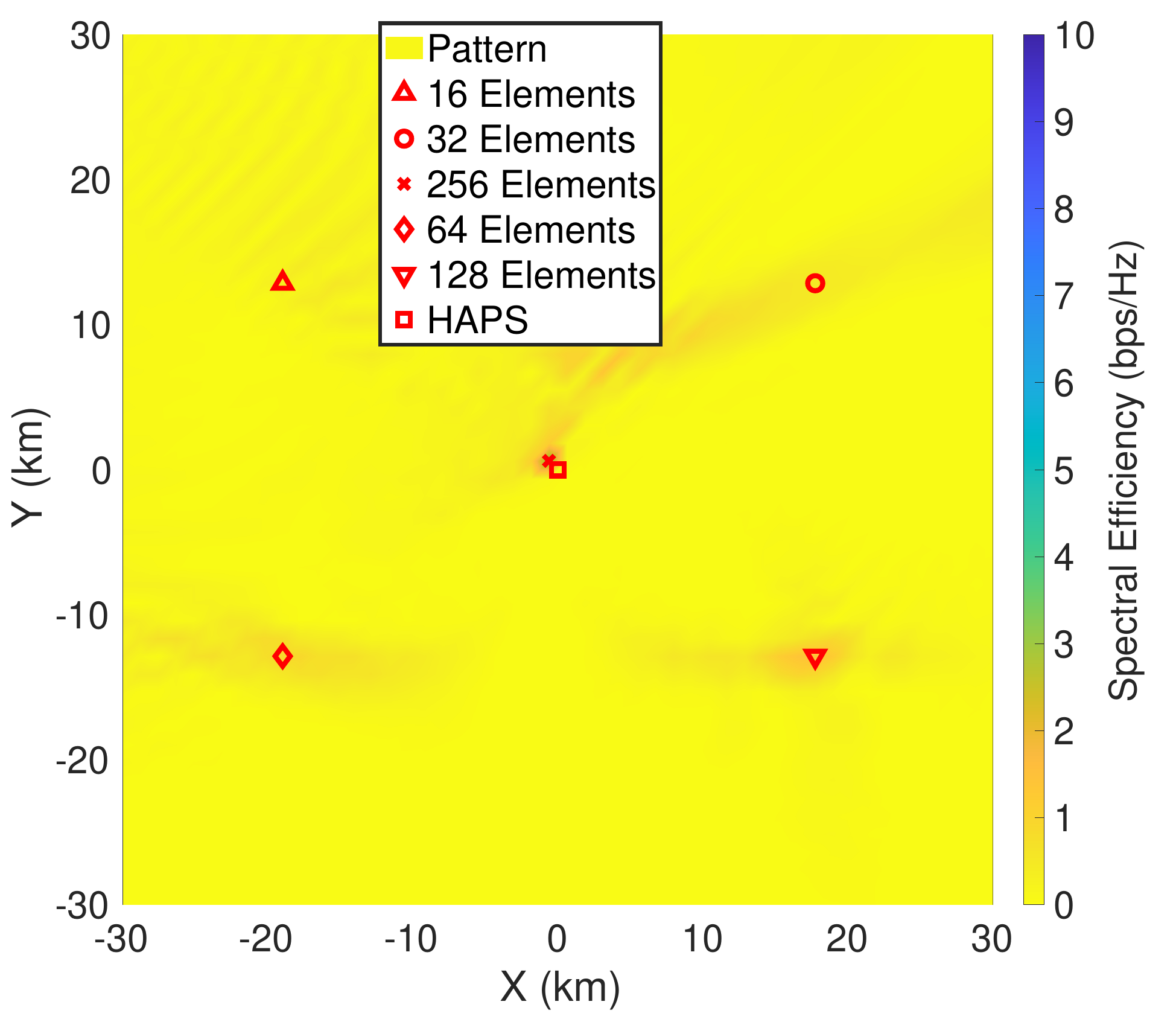} 
\caption{Cylindrical,$\theta_{\text{3dB}}=25$}
\label{heat_M_C25}
\end{subfigure}
\caption{
Heatmaps displaying spectral efficiencies for five beams with different numbers of selected antenna elements are presented for the rectangular and cylindrical baseline schemes. Each beam is allocated a fixed power of $1~\mathrm{Watt}$. These heatmaps consider two distinct $3~\mathrm{dB}$ beamwidth values for antenna elements.}
\label{Heat_M_element_baselines}
\end{figure}

Fig.~\ref{Heat_M_element_baselines} shows a set of heatmaps displaying spectral efficiencies for five beams with different numbers of selected antenna elements — i.e., $M_k=16, 32, 64, 128,$ and $ 256$ — for the rectangular and cylindrical baseline schemes. Each beam is allocated a fixed power of $1~\mathrm{Watt}$.
For the rectangular scheme, Fig.~\ref{heat_M_R10} and Fig.~\ref{heat_M_R25} show the beams with a $3~\mathrm{dB}$ beamwidth values of $\theta_{\mathsf{3dB}}=10$ and $\theta_{\mathsf{3dB}}=25$, respectively. We can see that, the rectangular scheme, with its antenna elements facing
downwards, has much stronger beams for the location under the HAPS. For the cylindrical scheme, Fig.~\ref{heat_M_C10} and Fig.~\ref{heat_M_C25} show the beams with a $3~\mathrm{dB}$ beamwidth values of $\theta_{\mathsf{3dB}}=10$ and $\theta_{\mathsf{3dB}}=25$, respectively. With its antenna elements facing towards the horizon, the cylindrical scheme has similar beams for the locations under the HAPS and those far from the HAPS. Note that for $\theta_{\mathsf{3dB}}=25$ case, due to a better array gain, narrower beams are created.

\begin{figure}[t]
\begin{subfigure}{0.5\textwidth}
\includegraphics[width=1\linewidth, height=4.2cm]{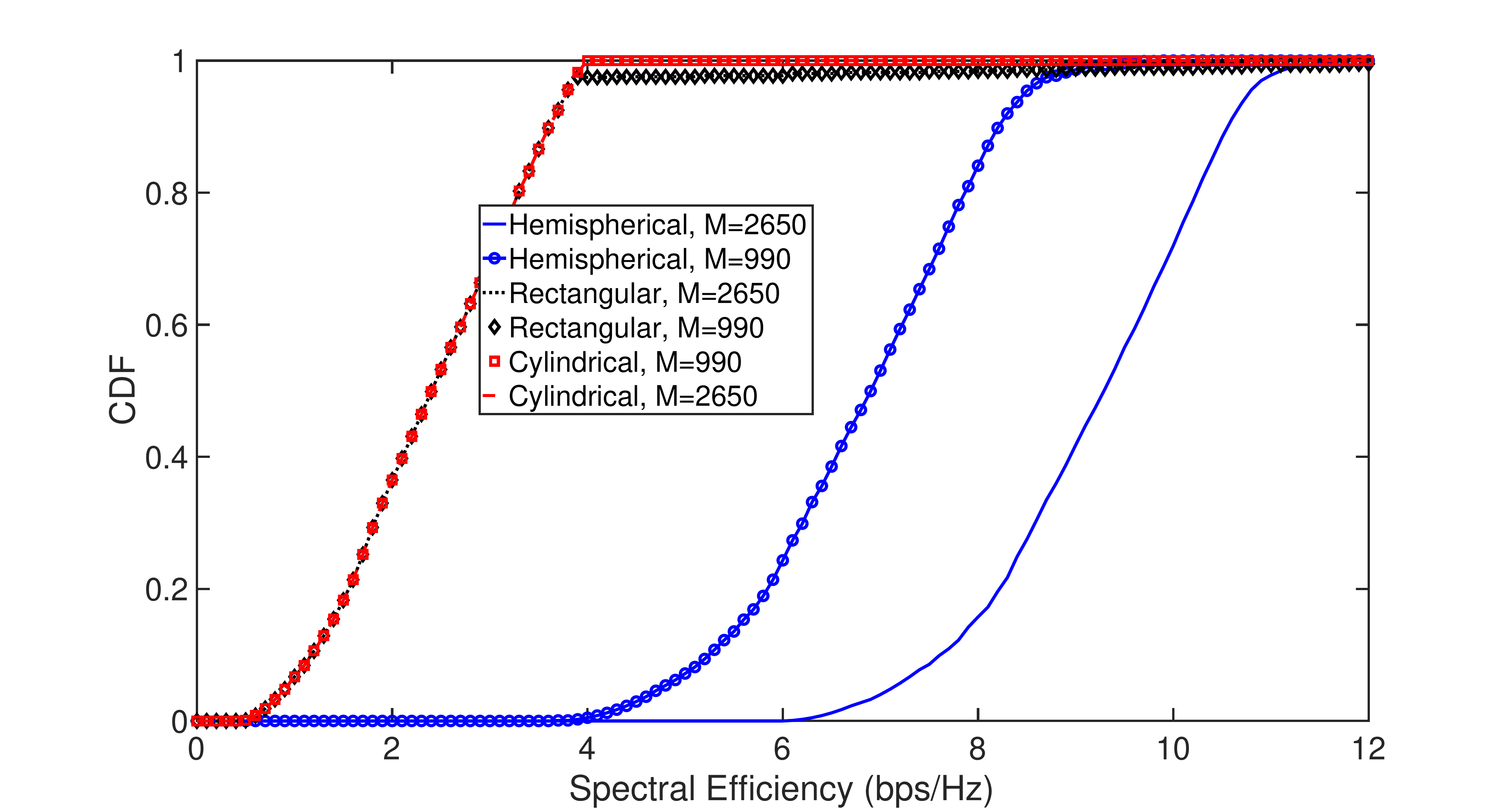} 
\caption{$\theta_{\text{3dB}}=10$}
\label{CDF_M10}
\end{subfigure}
\begin{subfigure}{0.5\textwidth}
\includegraphics[width=1\linewidth, height=4.2cm]{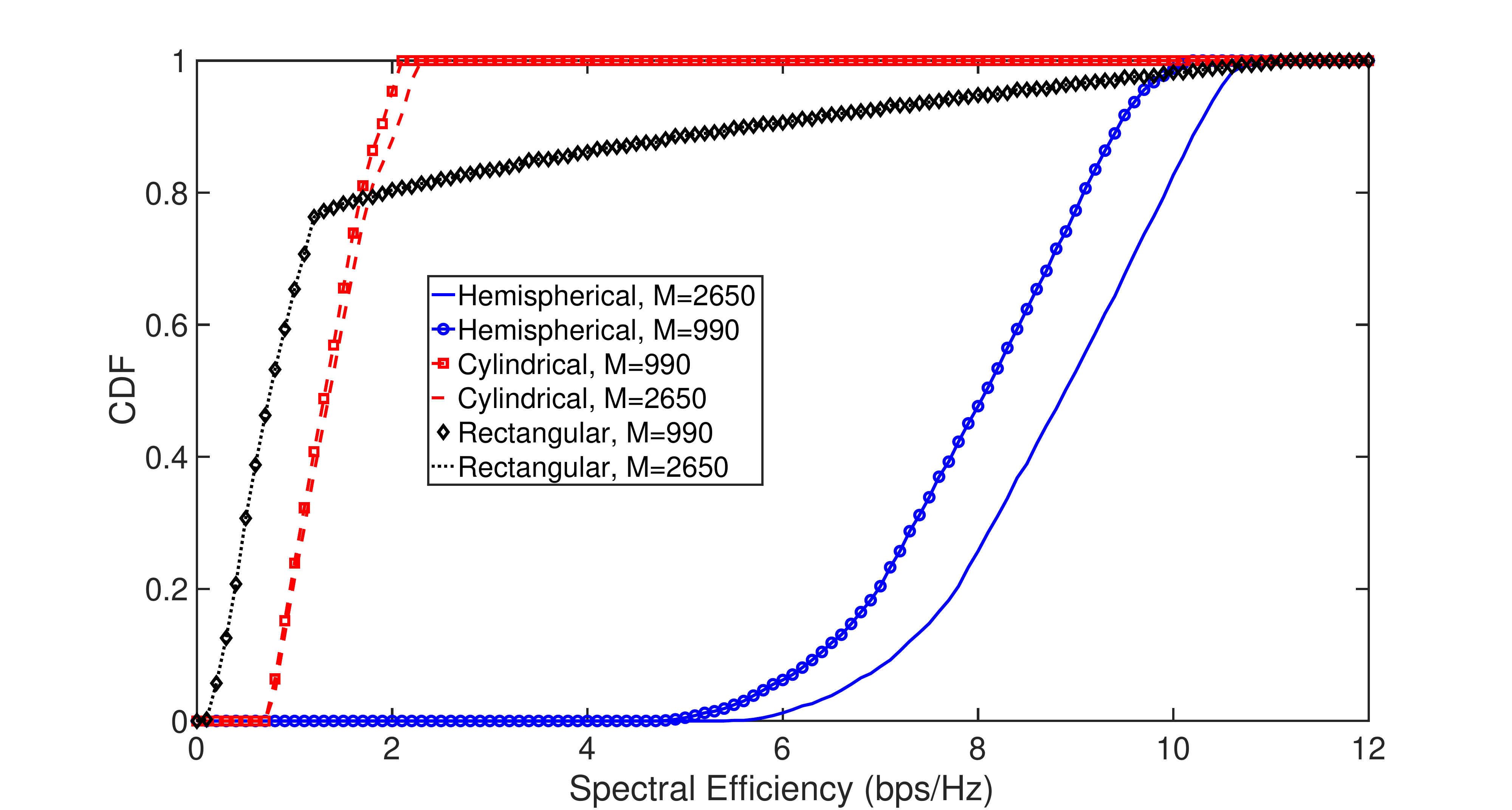}
\caption{$\theta_{\text{3dB}}=25$}
\label{CDF_M25}
\end{subfigure}
\caption{
CDF of spectral efficiency for a user uniformly distributed across $10,000$ different locations in a square urban area with dimensions $60~\mathrm{km} \times 60~\mathrm{km}$. We consider two distinct values for the total number of antenna elements in the array — namely, $M=990,~2650$. The user is allocated a fixed power of $1~\mathrm{Watt}$ and served by 64 antenna elements selected ($M_k=64$).}
\label{CDF_M10_25}
\end{figure}

Fig.~\ref{CDF_M10_25} shows the CDF of spectral efficiency for a user that is uniformly distributed across $10,000$ different locations in a square urban area with dimensions $60~\mathrm{km} \times 60~\mathrm{km}$. We consider two distinct values for the total number of antenna elements in the array — namely, $M=990,~\text{and}~2650$. We assume the user is allocated a fixed power of $1~\mathrm{Watt}$ and served by 64 selected antenna elements ($M_k=64$).
Fig. \ref{CDF_M10} and Fig. \ref{CDF_M25} show the CDF for $3~\mathrm{dB}$ beamwidth values of $\theta_{\text{3dB}}=10$ and $\theta_{\text{3dB}}=25$, respectively. For the HAA scheme, increasing $M$ improves the rates, which can be attributed to the fact that, for a bigger $M$, more elements with a higher gain will be available for each user to be selected. Note that this gain is more pronounced for the $\theta_{\text{3dB}}=10$ case, as the beam of each antenna element is narrower. For the rectangular case, a higher $M$ does not change the direction and gain of antenna elements at users' location as all elements are looking downwards. Therefore, the CDF remains fixed for both values of $M$. For the CAA, a higher $M$ leads to a higher granularization only in the azimuth angles of the elements. Therefore, we observe a small gain for increasing $M$ in the cylindrical array. Note that for the HAA, a higher $M$ leads to a higher granularization in both azimuth and elevation angles of the elements, while, for the RAA, a higher $M$ does not improve the angles of the elements.

\begin{figure}[t]
\begin{subfigure}{0.5\textwidth}
\includegraphics[width=1\linewidth, height=4.5cm]{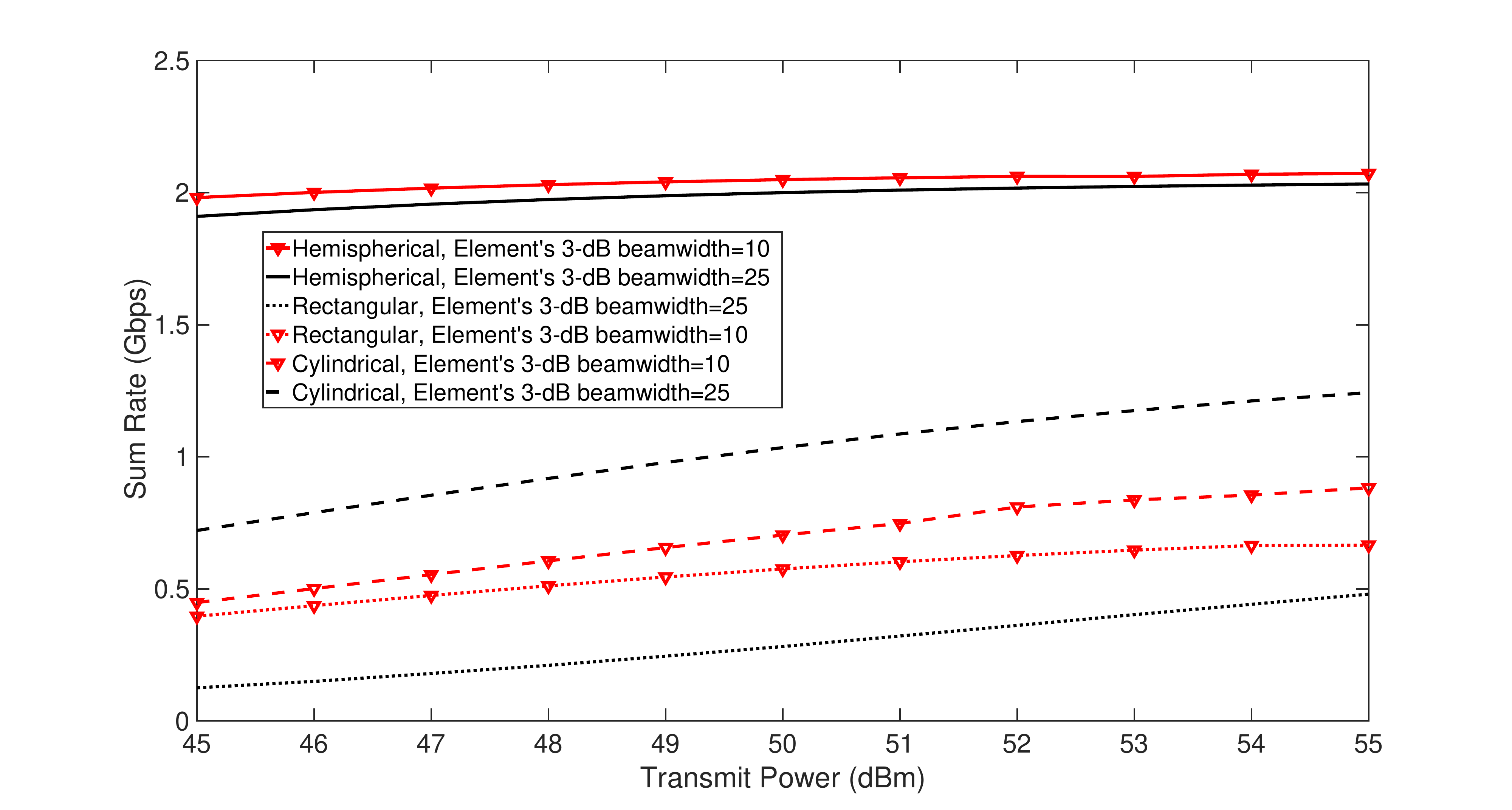} 
\caption{$M_{k}=64~ \text{Elements}$}
\label{sum_power_M64}
\end{subfigure}
\begin{subfigure}{0.5\textwidth}
\includegraphics[width=1\linewidth, height=4.5cm]{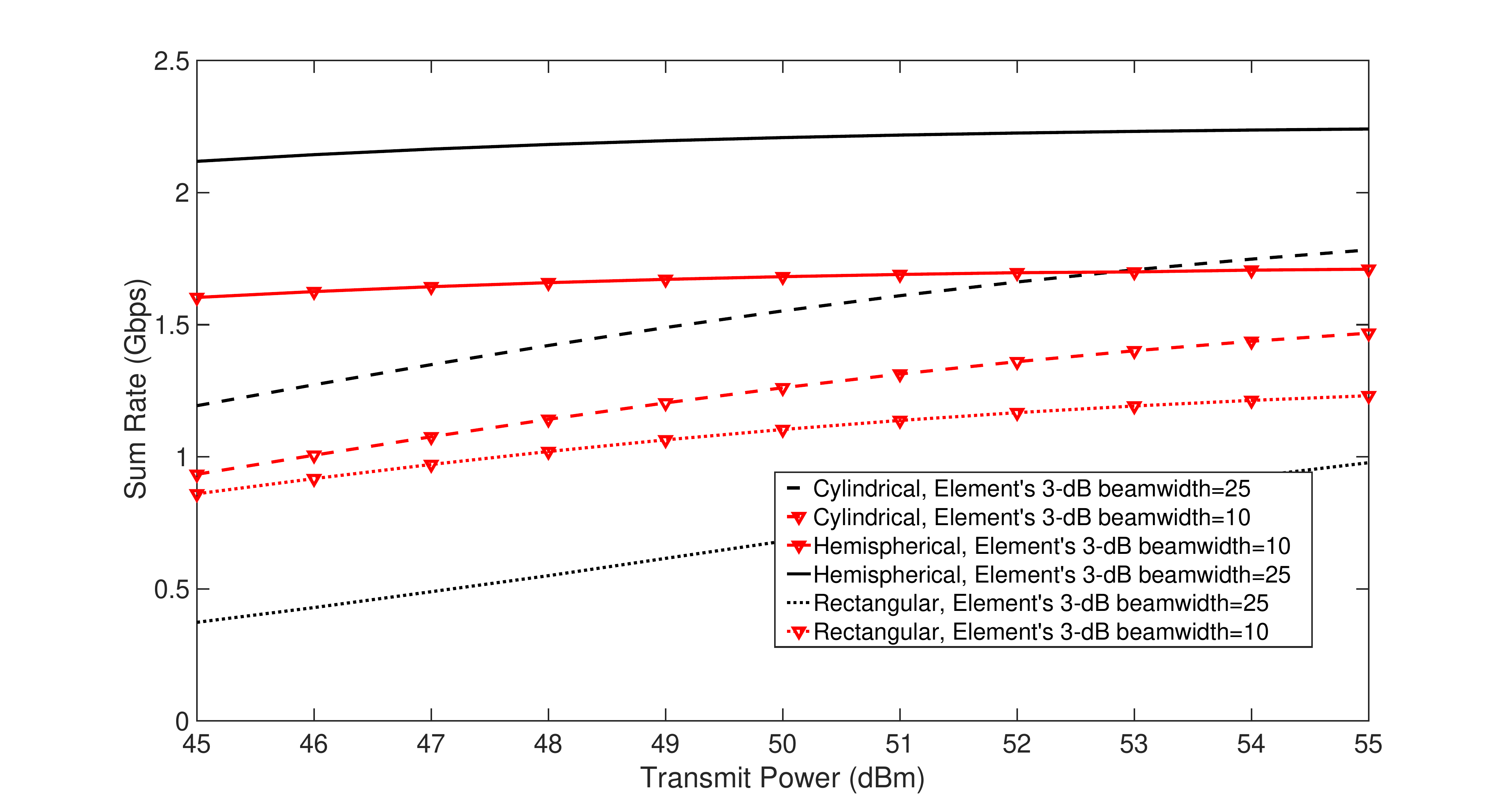}
\caption{$M_{k}=256~ \text{Elements}$}
\label{sum_power_M256}
\end{subfigure}
\caption{
Sum rate versus the total transmit power at the HAPS for $16$ users uniformly distributed in a square urban area measuring $60~\mathrm{km} \times 60~\mathrm{km}$. Simulation employs optimized parameter values proposed in Algorithm 3, a communication bandwidth of $20~\mathrm{MHz}$, and a total transmit power of $50~\mathrm{dBm}$ at the HAPS.
}
\label{sum_power_opt}
\end{figure}

Fig.~\ref{sum_power_opt} illustrates the sum rate of the HAA, CAA, and RAA schemes in relation to the total transmit power at the HAPS. Simulation was conducted with $16$ users uniformly distributed in a square urban area measuring $60~\mathrm{km} \times 60~\mathrm{km}$ using the optimal parameter values proposed in Algorithm 3. We also assumed a communication bandwidth of $20~\mathrm{MHz}$ and a total transmit power of $50~\mathrm{dBm}$ at the HAPS.
The achieved sum rate is presented for two cases: one with $M_k=64$ elements (see Fig.~\ref{sum_power_M64}) and the other with $M_k=256$ elements (see Fig.~\ref{sum_power_M256}). The HAA scheme outperforms both the CAA and RAA schemes. Additionally, the CAA scheme performs better than the RAA scheme.
Importantly, the HAA scheme's sum rate is lower with $\theta_{\mathsf{3dB}}=25$ than $\theta_{\mathsf{3dB}}=10$ when $M_k=64$. However, when $M_k=256$, the narrower beamwidth performs worse than the wider one. This discrepancy can be attributed to the fact that wider beamwidth values require more antenna elements to achieve an effective array (beamforming) gain.

\begin{figure}[t]
\begin{subfigure}{0.5\textwidth}
\includegraphics[width=1\linewidth, height=4.5cm]{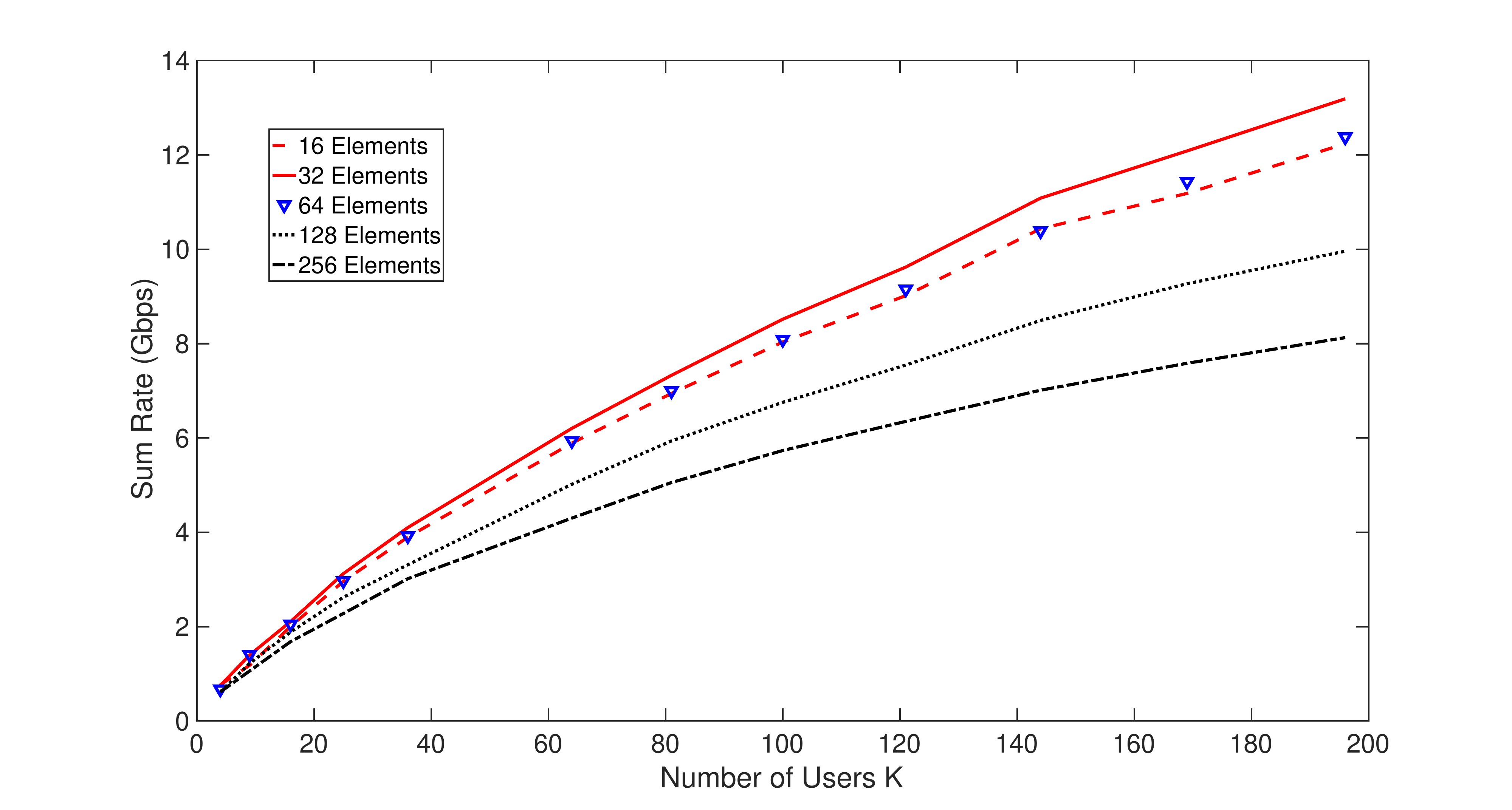} 
\caption{$\theta_{\mathsf{3dB}}=10$}
\label{sum_K_10}
\end{subfigure}
\begin{subfigure}{0.5\textwidth}
\includegraphics[width=1\linewidth, height=4.5cm]{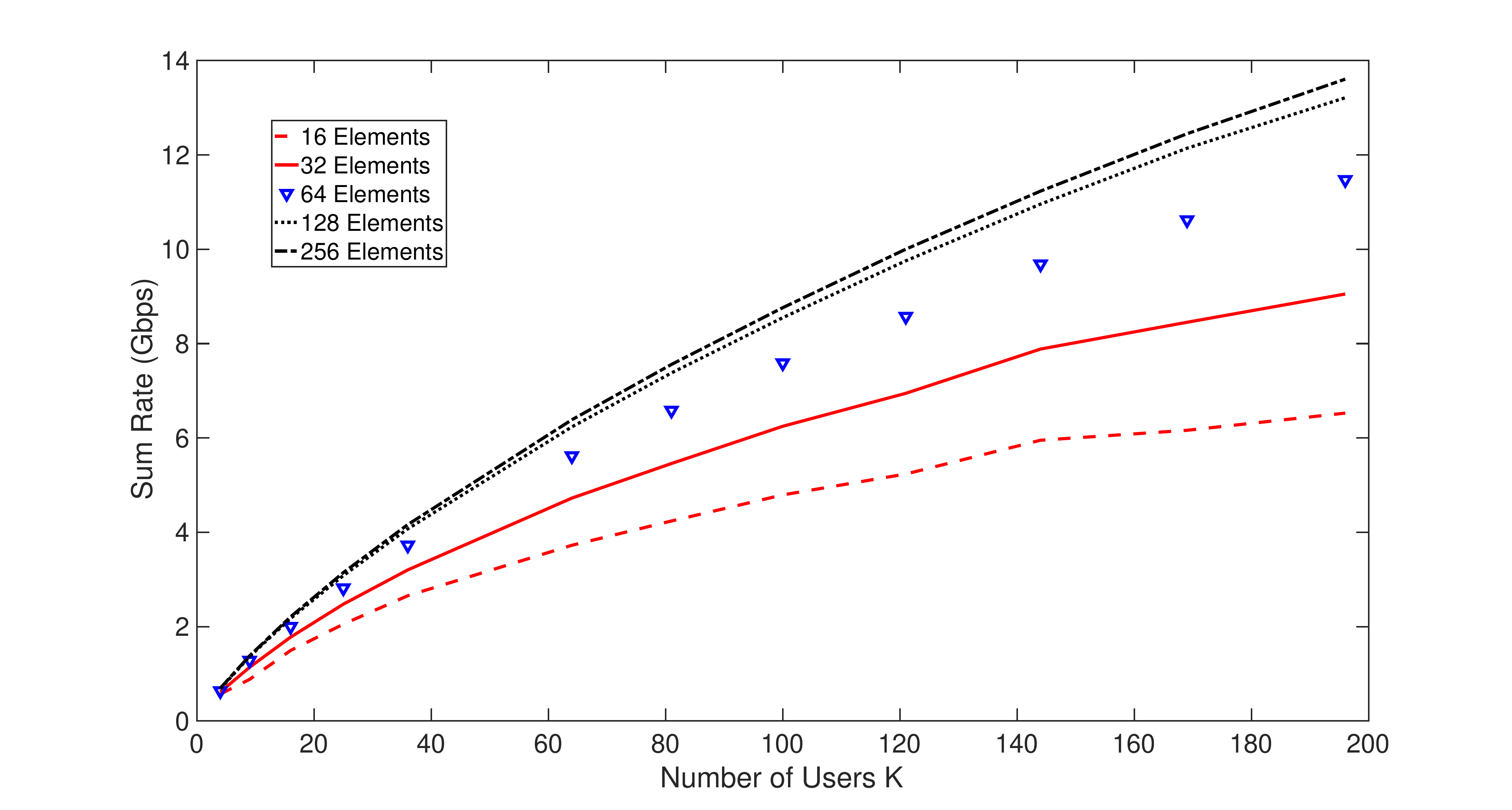}
\caption{$\theta_{\mathsf{3dB}}=25$}
\label{sum_K_25}
\end{subfigure}
\caption{Sum rate versus the number of users $K$ is shown for users uniformly distributed in a square urban area measuring $60~\mathrm{km} \times 60~\mathrm{km}$. Simulation employs optimized parameter values proposed in Algorithm 3, a communication bandwidth of $20~\mathrm{MHz}$, and a total transmit power of $50~\mathrm{dBm}$ at the HAPS.
}
\label{sum_K_opt}
\end{figure}

Fig. \ref{sum_K_opt} shows the sum rate as a function of the number of users $K$ uniformly distributed in a square urban area measuring $60~\mathrm{km} \times 60~\mathrm{km}$. Simulation employs the optimized parameter values proposed in Algorithm 3, a communication bandwidth of $20~\mathrm{MHz}$, and a total transmit power of $50~\mathrm{dBm}$ at the HAPS.
This figure presents the sum rates for five different numbers of antenna elements selected, i.e., $M_k=16, 32, 64, 128,$ and $ 256$. Increasing the number of users leads to higher sum rates. Fig. \ref{sum_K_10}, which shows the sum rate when the elements have a narrower $3~\mathrm{dB}$ beamwidths of $\theta_{\mathsf{3dB}}=10$, reveals that increasing $M_k$ from $16$ to $32$ improves the sum rate; however, further increasing it from $32$ to $256$ causes a decrease.
Fig. \ref{sum_K_25} shows the sum rate when the elements have a wider $3~\mathrm{dB}$ beamwidths of $\theta_{\mathsf{3dB}}=25$. In this case, increasing $M_k$ from $16$ to $256$ leads to an increase in the sum rate due to the array gain achieved by beamforming with a larger number of wider elements.

\begin{figure}[t]
\begin{subfigure}{0.5\textwidth}
\includegraphics[width=1\linewidth, height=4.5cm]{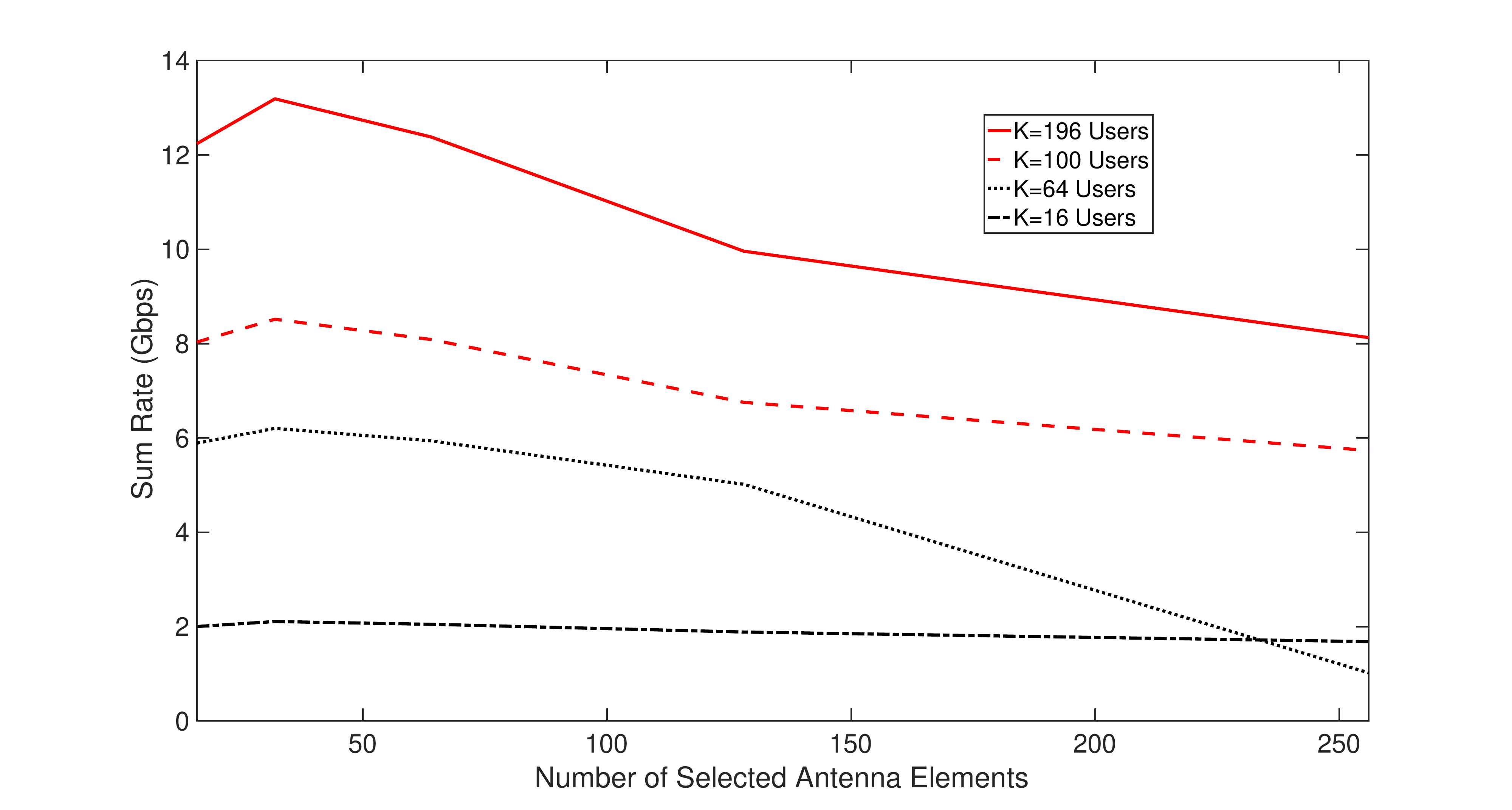} 
\caption{$\theta_{\mathsf{3dB}}=10$}
\label{sum_M_10}
\end{subfigure}
\begin{subfigure}{0.5\textwidth}
\includegraphics[width=1\linewidth, height=4.5cm]{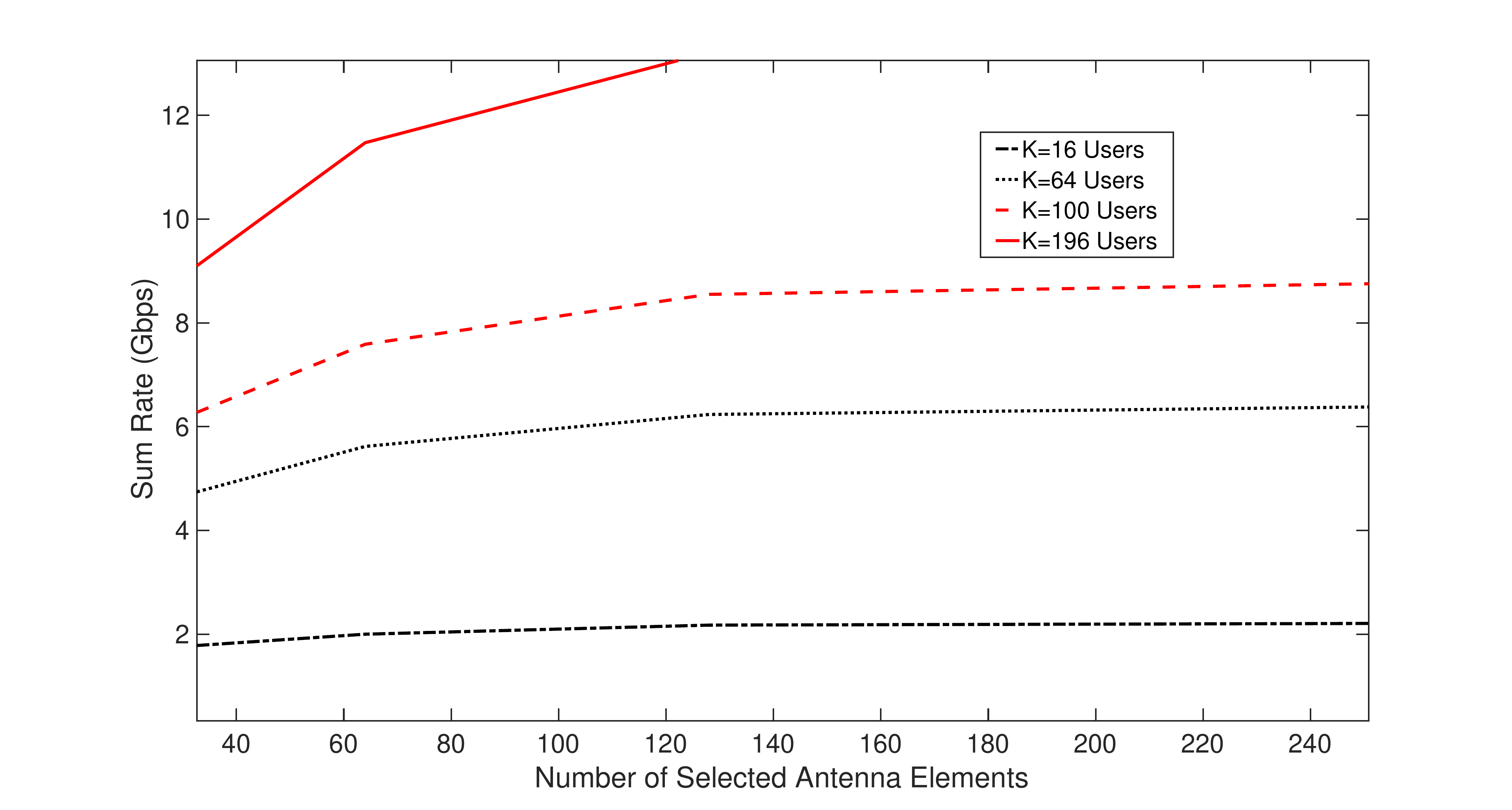}
\caption{$\theta_{\mathsf{3dB}}=25$}
\label{sum_M_25}
\end{subfigure}
\caption{
The sum rate as a function of the number of selected antenna elements $M_k$ for various numbers of users $K$ uniformly distributed in a square urban area measuring $60~\mathrm{km} \times 60~\mathrm{km}$. Simulation employs optimal parameter values proposed in Algorithm 3 and a communication bandwidth of $20~\mathrm{MHz}$.
}
\label{sum_M_opt}
\end{figure}

Fig. \ref{sum_M_opt} illustrates the sum rate as a function of the number of antenna elements selected $M_k$ for different numbers of users $K$ uniformly distributed in a square urban area measuring $60~\mathrm{km} \times 60~\mathrm{km}$. Simulation employs the optimal parameter values proposed in Algorithm 3, a communication bandwidth of $20~\mathrm{MHz}$, and a total transmit power of $50~\mathrm{dBm}$ at the HAPS.
This figure provides insight into the sum rates for different values of $K$ — specifically, $K=16, 64, 100,$ and $196$. Fig. \ref{sum_M_10} displays the sum rate when the elements have a narrower $3~\mathrm{dB}$ beamwidth of $\theta_{\mathsf{3dB}}=10$. It is evident that the sum rate first improves with an increase in $M_k$ from $16$ to $32$, but then declines with an increase of $M_k$ from $32$ to $256$. Furthermore, the optimal number of antenna elements selected to maximize the sum rate is $M_k=32$.
Fig. \ref{sum_M_25}, which shows the sum rate when the elements have a wider $3~\mathrm{dB}$ beamwidth of $\theta_{\mathsf{3dB}}=25$, suggests that, when $M_k$ increases from $16$ to $256$, the sum rate increases too, thus highlighting the array gain that can be achieved by beamforming with a larger number of wider elements.

\begin{figure}[t]
\begin{subfigure}{0.5\textwidth}
\includegraphics[width=1\linewidth, height=4.5cm]{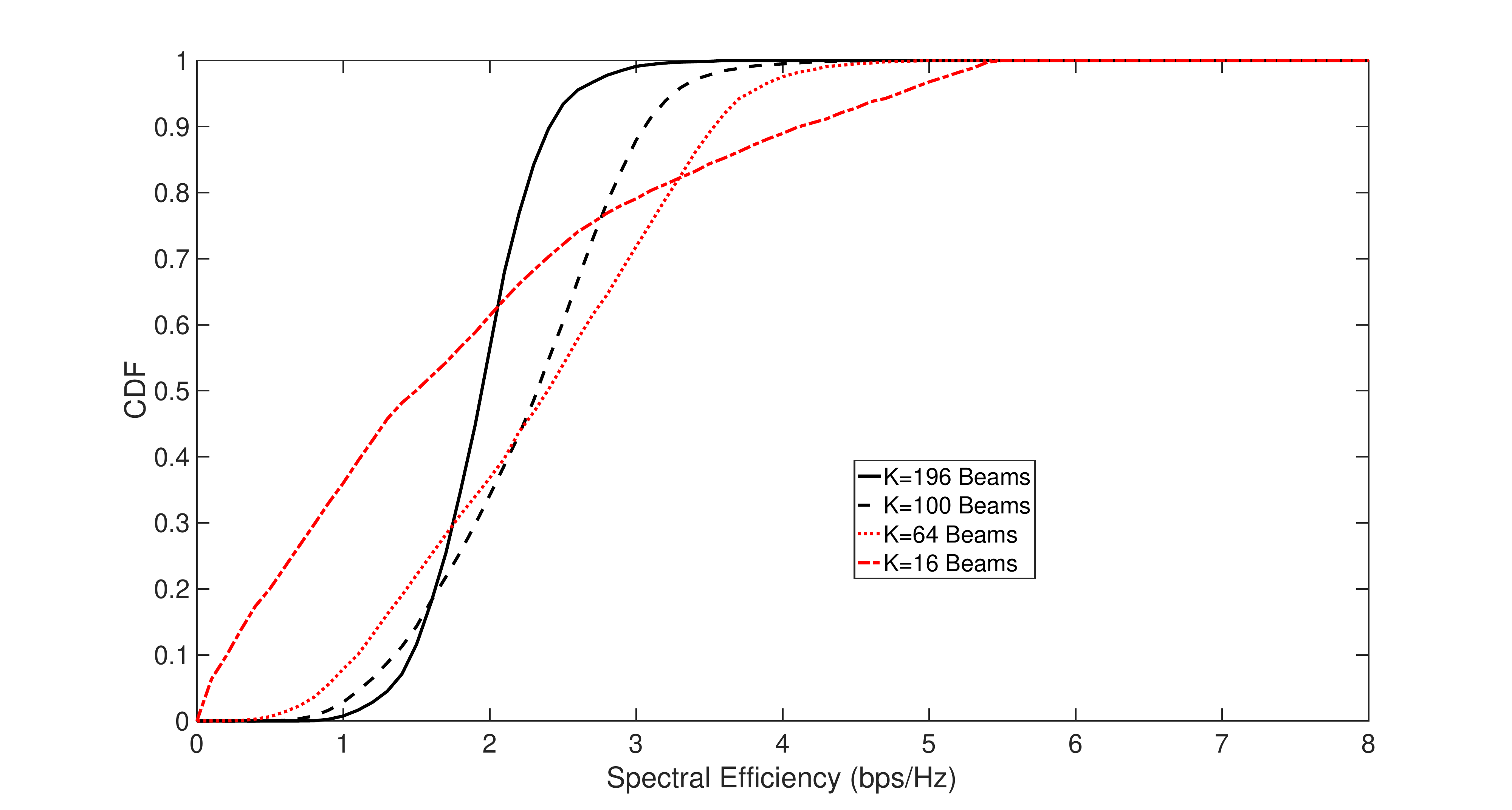} 
\caption{$\theta_{\text{3dB}}=10$}
\label{CDF_opt_10}
\end{subfigure}
\begin{subfigure}{0.5\textwidth}
\includegraphics[width=1\linewidth, height=4.5cm]{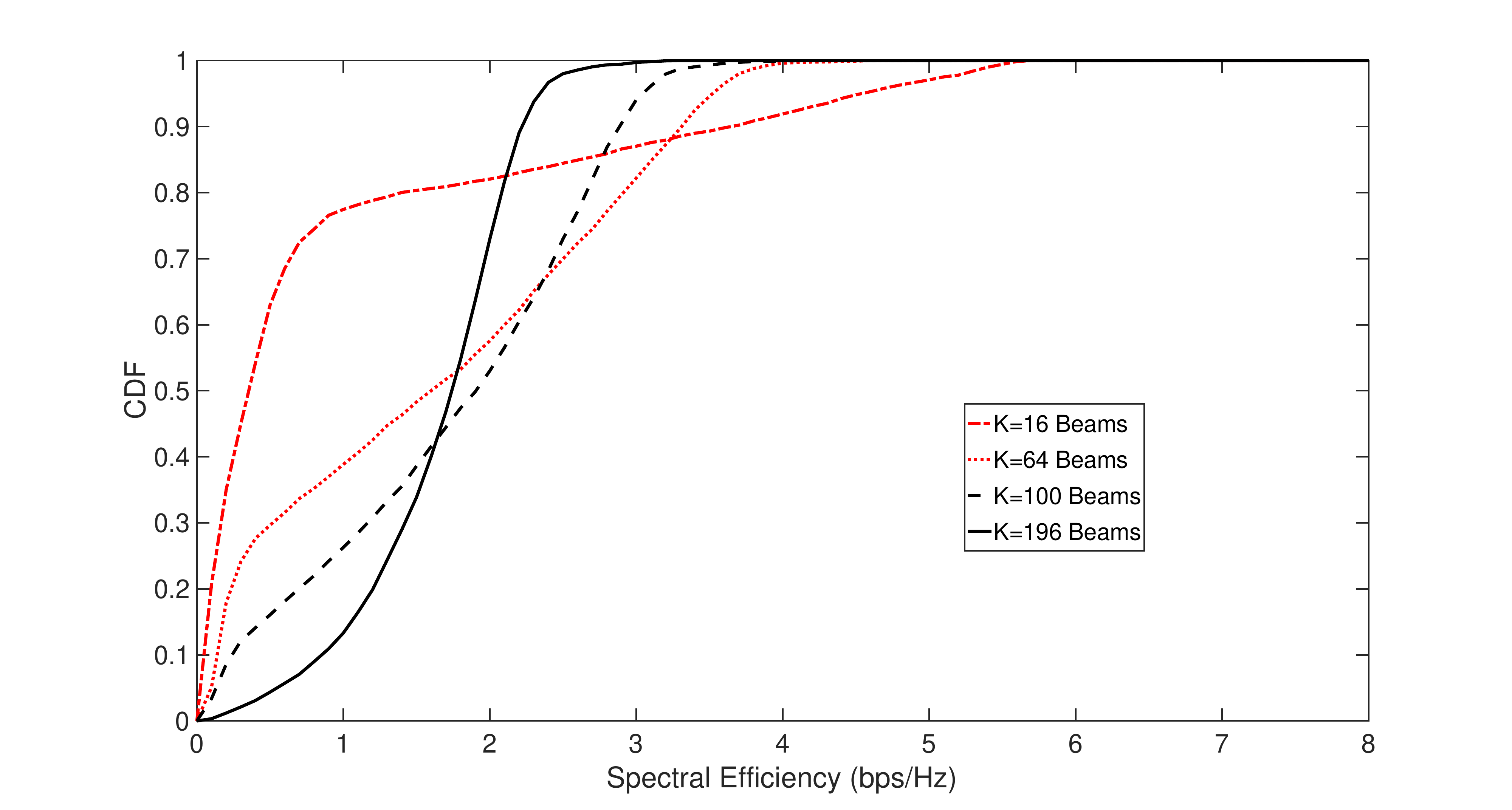}
\caption{$\theta_{\text{3dB}}=25$}
\label{CDF_opt_25}
\end{subfigure}
\caption{CDF of spectral efficiencies of users, with a total transmit power of $50~\rm{dBm}$ at the HAPS and $M_k=64$ selected antenna elements for each beam. We conducted simulations for four different quantities of beams — specifically, $K=16,~ 64, ~100,$ and $196$. Also, $5,000$ users were uniformly distributed in a square urban area measuring $60~\mathrm{km} \times 60~\mathrm{km}$.}
\label{CDF_opt}
\end{figure}

Fig. \ref{CDF_opt} displays the CDF of the spectral efficiencies of the users, with a total transmit power of $50~\rm{dBm}$ at the HAPS and $M_k=64$ antenna elements selected for each beam. This figure illustrates spectral efficiency distribution for four different quantities of beams — specifically, $K=16,~ 64, ~100,$ and $196$.
For this figure, $K$ beams are first generated to cover $K$ designated locations on the ground and create a square grid within an urban area measuring $60~\mathrm{km} \times 60~\mathrm{km}$. Then, the optimized parameter values for the $K$ beams are determined using Algorithm 3. Next, $5,000$ users are uniformly distributed in the urban area with each user assigned to the nearest beam center on the ground.
Fig. \ref{CDF_opt_10} and Fig. \ref{CDF_opt_25} provide the CDFs for $\theta_{\mathsf{3dB}}=10$ and $\theta_{\mathsf{3dB}}=25$, respectively. It is evident that a greater number of beams ($K$) leads to more uniform user spectral efficiencies. Additionally, it can be observed that $\theta_{\mathsf{3dB}}=10$ yields higher spectral efficiencies than $\theta_{\mathsf{3dB}}=25$, which is mainly due to the limitation of having only $M_k=64$ selected antenna elements, which does not provide sufficient beamforming gain for wider beamwidths.

\begin{figure}[t]
\begin{subfigure}{0.5\textwidth}
\includegraphics[width=1\linewidth, height=4.5cm]{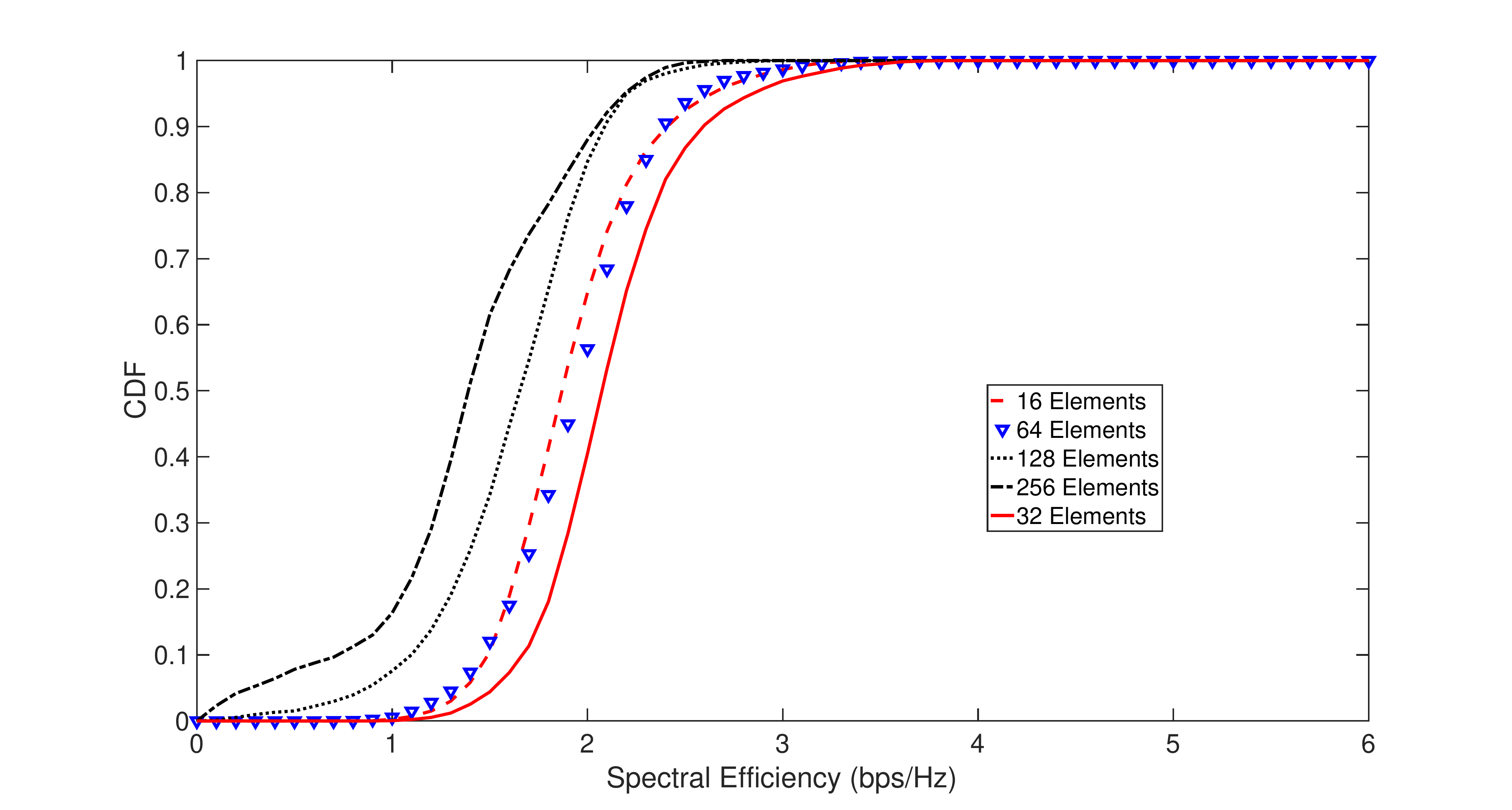} 
\caption{$\theta_{\mathsf{3dB}}=10$}
\label{CDF_opt_K_10}
\end{subfigure}
\begin{subfigure}{0.5\textwidth}
\includegraphics[width=1\linewidth, height=4.5cm]{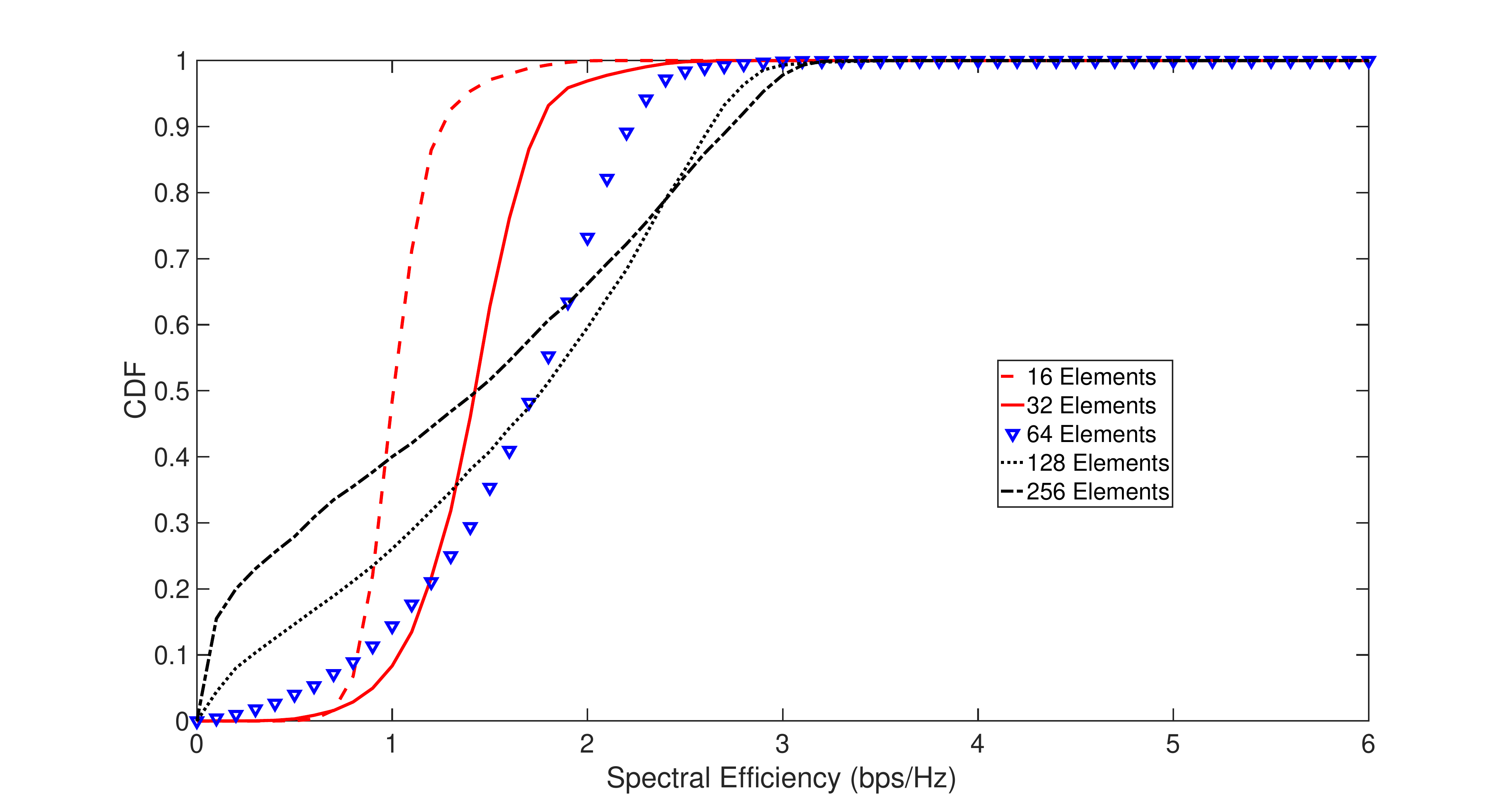}
\caption{$\theta_{\mathsf{3dB}}=25$}
\label{CDF_opt_K_25}
\end{subfigure}
\caption{CDF of spectral efficiencies of users, with a total transmit power of $50~\rm{dBm}$ at the HAPS and $K=196$ beams. We conducted simulations for five different quantities of selected antenna elements, i.e., $M_k=16,~ 32,~ 64,~ 128,$ and $256$. Additionally, $5,000$ users were uniformly distributed in a square urban area measuring $60~\mathrm{km} \times 60~\mathrm{km}$.}
\label{CDF_opt_K}
\end{figure}

Fig. \ref{CDF_opt_K} presents the CDF of the spectral efficiencies of the users assuming $K=196$ beams and the total transmit power of $50~\rm{dBm}$. This figure illustrates spectral efficiency distributions for five different quantities of antenna elements selected — namely, $M_k=16,~ 32,~ 64,~ 128, ~\text{and}~256$.
For this figure, we use the same methodology as the one used for Fig. \ref{CDF_opt}.
Fig. \ref{CDF_opt_K_10} depicts the CDF when the antenna elements have a narrower $3~\rm{dB}$ beamwidths of $\theta_{\mathsf{3dB}}=10$. We observe that the spectral efficiency improves by increasing $M_k$ from $16$ to $32$, but a further increase from $32$ to $256$ results in a reduction of spectral efficiency. 
Fig. \ref{CDF_opt_K_25} displays the CDF when the antenna elements have a wider $3~\rm{dB}$ 
beamwidth of $\theta_{\mathsf{3dB}}=25$. The results reveal that increasing $M_k$ from $16$ to $256$ leads to higher spectral efficiencies due to the array gain achieved by beamforming with a greater number of wider elements. These observations are consistent with the findings shown in Fig. \ref{sum_K_10} and Fig. \ref{sum_K_25}.
Moreover, with an increase of the number of selected antenna elements in Fig. \ref{CDF_opt_K_25}, homogeneity of users' achieved spectral efficiencies decreases. This outcome is attributable to the beamforming gain being more advantageous for those located closer to the beam center, while users located further from it experience lower spectral efficiencies.

\section{Conclusion}
In this paper, we introduced a novel hemispherical antenna array (HAA) for high-altitude platform station (HAPS) to respond the challenges posed by conventional rectangular and cylindrical antenna arrays. The results of our simulations conclusively demonstrate that our proposed HAA scheme outperforms traditional rectangular and cylindrical baseline arrays. Furthermore, our simulation results underscore the significance of antenna element beamwidth. When beamwidth is narrower, the harnessing of array (beamforming) gain is less effective and it becomes more advantageous to craft beams with fewer antenna elements selected. Conversely, when beamwidth is wider, beamforming emerges as a potent tool for generating highly focused high-gain beams with a larger number of selected antenna elements. In an urban area spanning a length of $60~\mathrm{km}$, with a communication bandwidth of $20~\mathrm{MHz}$ and a total power at HAPS set at $50~\mathrm{dBm}$, the proposed approach can efficiently attain sum data rates of up to $14 ~\mathrm{Gbps}$. Finally, in contrast to the baseline schemes, the proposed approach achieves uniform data rates across the entire coverage area.

%
\appendices

\appendices

\section{Proof of Proposition \ref{propos_sinr}}\label{proof_sinr}
To derive each user's SINR, we rewrite the signal received at user $k$, i.e., $y_k$ in (\ref{y_k}),  as
\begin{equation}\label{y_k_vect}
    \begin{split}
&y_{k}=\mathbf{h}_k\mathbf{G}_k(\mathbf{W}\odot \mathbf{A})\mathbf{P}\mathbf{s}+z_{k}\\&=\mathbf{h}_k\mathbf{G}_k(\mathbf{w}_k\odot \mathbf{a}_k)\sqrt{p_k}s_k+\sum_{k'=1,k'\neq k}^{K}\mathbf{h}_k\mathbf{G}_k(\mathbf{w}_{k'}\odot \mathbf{a}_{k'})\sqrt{p_{k'}}s_{k'}+z_k\\&=\mathbf{h}_k\mathbf{G}_k\mathbf{A}_k\mathbf{w}_k\sqrt{p_k}s_k+\sum_{k'=1,k'\neq k}^{K}\mathbf{h}_k\mathbf{G}_k\mathbf{A}_{k'}\mathbf{w}_{k'}\sqrt{p_{k'}}s_{k'}+z_k,
    \end{split}
\end{equation}
where $\mathbf{A}_k=\mathrm{diag}(\mathbf{a}_{k}) \in \mathbb{B}^{M\times M}$. Then, we utilize the use-and-then-forget bound \cite{marzetta2016fundamentals} to derive an achievable SINR. From the last equation in (\ref{y_k_vect}), we write the expectation of the desired signal (DS) for user $k$ as
\begin{equation}
    \begin{split}
\mathrm{E}\{\mathsf{DS}_k\}&=\mathrm{E}\{\mathbf{h}_k\mathbf{G}_k\mathbf{A}_k\mathbf{w}_k\sqrt{p_k}\}.
    \end{split}
\end{equation}

In this paper, we apply matched filtering to beamform the signals toward users. Hence, the PS for each user is calculated based on the conjugate of the user's steering vector, i.e., $\mathbf{w}_{k}=\frac{1}{\sqrt{M_k}}\mathbf{b}_k^*$, where $M_k$ indicates the number of antenna elements selected for user $k$. Therefore, we have
\begin{equation}
\begin{split}
\mathrm{E}\{\mathsf{DS}_k\}&=\sqrt{\frac{p_k}{M_k}}\mathrm{E}\{\mathbf{h}_k\mathbf{G}_k\mathbf{A}_k\mathbf{w}_k\}
=\sqrt{\frac{p_k\beta_k^2}{M_k}}\mathrm{Tr}(\mathbf{G}_k\mathbf{A}_k).
\end{split}
\end{equation}

Next, we derive the variance of the interference term in (\ref{y_k_vect}) as

\begin{equation}
\begin{split}
\mathrm{E}\{I_kI_k^*\}&=\mathrm{E}\{(\sum_{k'=1}^{K}\mathbf{h}_k\mathbf{G}_k\mathbf{A}_{k'}\mathbf{w}_{k'}\sqrt{p_{k'}})(\sum_{k''=1}^{K}\mathbf{h}_k\mathbf{G}_k\mathbf{A}_{k''}\mathbf{w}_{k''}\sqrt{p_{k''}})^*\}\\&=\mathrm{E}\{(\sum_{k'=1}^{K}(\mathbf{h}_k\mathbf{G}_k\mathbf{A}_{k'}\mathbf{w}_{k'}\sqrt{p_{k'}})(\mathbf{h}_k^*\mathbf{G}_k\mathbf{A}_{k'}\mathbf{w}_{k'}^*\sqrt{p_{k'}})\}\\&=\beta_k^2\sum_{k'=1}^K p_{k'}\mathrm{E}\{(\mathbf{h}_k\mathbf{G}_k\mathbf{A}_{k'}\mathbf{w}_{k'})(\mathbf{h}_k^*\mathbf{G}_k\mathbf{A}_{k'}\mathbf{w}_{k'}^*)\}\\&=\beta_k^2\sum_{k'=1}^K \frac{p_{k'}}{M_{k'}}\mathrm{Tr}(\mathbf{G}_k^2\mathbf{A}_{k'}).
\end{split}
\end{equation}

Now, we calculate the SINR of user $k$ from the formulas derived for the desired signal and interference as
\begin{equation}
\begin{split}
\mathsf{SINR}_k&=\frac{\mathrm{E}\{\mathsf{DS}_k\}^2}{\mathrm{E}\{I_kI_k^*\}+\sigma^2}=\frac{\frac{p_k\beta_k^2}{M_k}\mathrm{Tr}(\mathbf{G}_k\mathbf{A}_k)^2}{\beta_k^2\sum_{k'=1}^K \frac{p_{k'}}{M_{k'}}\mathrm{Tr}(\mathbf{G}_k^2\mathbf{A}_{k'})+\sigma^2}.
\end{split}
\end{equation}

Finally, we can write user $k$'s achievable rate as $\log_2(1+\mathsf{SINR}_k)$. This completes the proof.

\section{Proof of Proposition \ref{propos_gain}}\label{proof_G}
First, we derive the angle between user $k$ and antenna $m$, i.e., $\theta_{km}$. To this end, we write their unit vectors in Cartesian coordinates as 
$\mathbf{v}_k=[\sin{\theta_{k}}\cos{\phi_{k}},\sin{\theta_{k}}\sin{\phi_{k}},\cos{\theta_{k}}],
$
and 
$
\mathbf{v}_m=[\sin{\theta_{m}}\cos{\phi_{m}},\sin{\theta_{m}}\sin{\phi_{m}},\cos{\theta_{m}}].
$

 Now, based on the geometric definition of dot product, we can write
$$
\mathbf{v}_m\cdot\mathbf{v}_k=\norm{\mathbf{v}_m}\norm{\mathbf{v}_k}\cos{\theta_{km}}=\cos{\theta_{km}}.
$$

Therefore, we have
\begin{equation}
    \begin{split}  \cos{\theta_{km}}&=\sin{\theta_{k}}\sin{\theta_{m}}\cos{\phi_{k}}\cos{\phi_{m}}\\&+\sin{\theta_{k}}\sin{\theta_{m}}\sin{\phi_{k}}\sin{\phi_{m}}+\cos{\theta_{k}}\cos{\theta_{m}},
    \end{split}
\end{equation}
and hence, $\theta_{km}$ can be derived as \eqref{theta-km}. 

It is clear that if $\theta_{km}>90$, antenna element $m$ will achieve no gain user $k$'s location. For $\theta_{km}<90$, we know that an antenna element's maximum directional gain $G_{\mathsf{E,max}}$ must be equal to $1$ when the $3~\rm{dB}$ beamwidth is equal to $180~\rm{degrees}$. Therefore, we can derive $G_{\mathsf{E,max}}$ as
\begin{equation}    \label{g_max_proof}
    G_{\mathsf{E,max}}=\frac{(180)^2}{\theta_{\mathsf{3dB}}^2}=\frac{32400}{\theta_{\mathsf{3dB}}^2}.
\end{equation}
Finally, we can derive antenna $m$'s loss at user $k$'s location based on \cite{3GPP_haps} as \eqref{gamma}. This completes the proof.

\section{Proof of Proposition \ref{propos_selection_optimal}}\label{proof_selection_optimal}
  In order to prove the proposed Algorithm's optimality by contradiction, we assume a different element for user $k$ that has less gain than proposed in Algorithm 1. This means that $g_{km_{\mathsf{selected}}}<g_{km_{\mathsf{proposed}}}$. Now, we write the SINR expression derived in \eqref{SINR} in scalar form as
    \begin{equation}    \label{SINR_scalar}
    \begin{split}    
&\mathsf{SINR}_k=\frac{\frac{p_k\beta_k^2}{M_k}\mathrm{Tr}(\mathbf{G}_k\mathbf{A}_k)^2}{\beta_k^2\sum_{k'=1}^K \frac{p_{k'}}{M_{k'}}\mathrm{Tr}(\mathbf{G}_k^2\mathbf{A}_{k'})+\sigma^2}\\&=\frac{\frac{\beta_k^2p_{k}}{M_k}(\sum_{m=1}^M g_{km}a_{km})^2}{\beta_k^2\sum_{k'=1}^K\frac{p_{k'}}{M_{k'}}\sum_{m=1}^Mg_{km}^2a_{k'm}+\sigma^2}\\&=\frac{\frac{\beta_k^2p_{k}}{M_k}(\sum_{m=1}^M g_{km}a_{km})^2}{\beta_k^2\frac{p_{k}}{M_{k}}\sum_{m=1}^Mg_{km}^2a_{km}+\beta_k^2\sum_{k'=1,k'\neq k}^K\frac{p_{k'}}{M_{k'}}\sum_{m=1}^Mg_{km}^2a_{k'm}+\sigma^2}.
\end{split}
\end{equation}
When a different element is selected for users $k$, this may have either a larger or smaller gain at other users' locations. When a large number of elements is selected for each user, these effects tend to offset one another. Consequently, the second term in the denominator of \eqref{SINR_scalar} remains relatively constant, even with a different element choice than what is proposed in Algorithm 1. However, the numerator and the first term in the denominator of \eqref{SINR_scalar} decrease in value. Per Lemma 1, reducing the value of each antenna element leads to a decrease in the SINR. This suggests that $\mathsf{SINR}_{km_{\mathsf{selected}}} < \mathsf{SINR}_{km_{\mathsf{proposed}}}$.
As user $k$'s SINR decreases, the minimum SINR among all users also declines. Consequently, the objective function of ($\mathcal{P}$) decreases whenever an antenna element other than the one proposed in Algorithm 1 is chosen. Therefore, whenever there is a large number of selected elements involved, Algorithm 1 is an optimal solution. This completes the proof.

\section{Proof of Proposition \ref{propos_quasi}}\label{proof_quasi}
To prove the quasi-linearity of ($\mathcal{P}1$), we show that the objective function is quasi-linear and the constraints are linear sets. To prove quasi-linearity of the objective function, we just need to prove that its upper-level set (ULS) is a linear set \cite{boyd2020disciplined}. To this end, we show the objective function with $f(\mathbf{P})$. Then, for any $t\in \mathbb{R}_+$, the ULS of the objective function is given by

\begin{equation}\label{uls}
  \begin{split}
      &\mathsf{ULS}(f,t)=\{\mathbf{P}:f(\mathbf{P})>t\}\\&=\bigg\{P:\frac{\frac{p_k\beta_k^2}{M_k}\mathrm{Tr}(\mathbf{G}_k\mathbf{A}_k)^2}{\beta_k^2\sum_{k'=1}^K \frac{p_{k'}}{M_{k'}}\mathrm{Tr}(\mathbf{G}_k^2\mathbf{A}_{k'})+\sigma^2}>t,~\forall k \in \mathcal{K}\bigg\}\\&=\bigg\{P:\frac{p_k\beta_k^2}{M_k}\mathrm{Tr}(\mathbf{G}_k\mathbf{A}_k)^2>t\bigg(\beta_k^2\sum_{k'=1}^K \frac{p_{k'}}{M_{k'}}\mathrm{Tr}(\mathbf{G}_k^2\mathbf{A}_{k'})+\sigma^2\bigg),\\&~~~~\forall k\in \mathcal{K}\bigg\},
  \end{split}  
\end{equation}
which is in the form of an affine function exceeding an affine function of variable $\mathbf{P}$; hence, it is a linear set. Constraint (\ref{eq:constraint-sum_no_W}) is also linear. This completes the proof.




\ifCLASSOPTIONcaptionsoff
  \newpage
\fi



 \bibliographystyle{IEEEtran}
 \bibliography{myref}
\end{document}